\documentclass[sigconf]{acmart}
\renewcommand\footnotetextcopyrightpermission[1]{} % Removes footnote with conference info
\setcopyright{none}
\copyrightyear{YYYY}

\acmConference[]{}{}{}
\acmBooktitle{}
\acmPrice{}
\acmDOI{}
\acmISBN{}

\usepackage{color} 
\usepackage{bbm}
\usepackage{algorithmic} 
\usepackage{algorithm}
\usepackage{enumitem}
\usepackage{subcaption}
\usepackage{caption}
\usepackage{soul}

\newenvironment{ourmatrix}{\begin{pmatrix}}{\end{pmatrix}}
\newenvironment{ourSmallMatrix}{\left(\begin{smallmatrix}}{\end{smallmatrix}\right)}

\newcommand{\ignoretext}[1]{}

\newcommand{\inca}{\textsc{IncA}}
\newcommand{\gopa}{GOPA}
\newcommand{\cordp}{CorDP-DME}
\newcommand{\muffliato}{Muffliato}

\usepackage{xcolor}         % Pour les couleurs étendues comme 'magenta', 'blue', etc.
\usepackage{fancybox}       % Pour \fbox (optionnel si tu veux un cadre plus stylisé)

\newcommand{\partySet}{P}
\newcommand{\partyCnt}{n}
\newcommand{\aParty}{i}
\newcommand{\partyA}{\aParty}
\newcommand{\partyB}{j}

\newcommand{\R}{\mathbb{R}}

\newcommand{\alg}{\mathcal{A}}

\newcommand{\algout}{Y}
\renewcommand{\epsilon}{\varepsilon}  % always use \varepsilon 
\newcommand{\advView}{\mathcal{V}}
\newcommand{\diag}[1]{\mathrm{diag}(#1)}

\newcommand{\corrnodes}{\mathcal{C}}  % set of nodes corrupted by the adversary
\newcommand{\onlSymb}{\mathcal{O}} % online parties symbol
\newcommand{\onlNodes}[1]{\onlSymb^{(#1)}} % set of online nodes at a given iteration
\newcommand{\wOnl}{w_{\onlSymb}}
\newcommand{\nOnlHonest}{n_\onlSymb}

\newcommand{\cancelNoiseSymb}{\eta}
\newcommand{\cancelNoiseMat}{\cancelNoiseSymb}
\newcommand{\cancelNoiseMatHonest}{\cancelNoiseMat^\honestSymb}
\newcommand{\cancelNoise}[2]{\cancelNoiseSymb_{#1,#2}} 
\newcommand{\cancelNoiseVecParty}[1]{\noiseMatrixEl{\aParty}{:}}
\newcommand{\cancelNoiseVec}{\noiseMatrix_{(:)}}

\newcommand{\noiseVec}{\begin{ourSmallMatrix}
		\indNoiseVecH \\
		\cancelNoiseVecHonest
\end{ourSmallMatrix}}
\newcommand{\indnoise}[1]{\indNoiseVec_{#1}}
\newcommand{\N}{\mathcal{N}}
\newcommand{\sdCancel}{\sigma_\Delta}
\newcommand{\indNoiseDecor}{\star}
\newcommand{\sdInd}{\sigma_\indNoiseDecor}

\newcommand{\oneVec}{\mathbbm{1}}
\newcommand{\project}[1]{{{\mathsf{b}}_{#1}}} % "b" for "basis vector".  we could also use b or e or u.
\newcommand{\xVectSpace}{\mathbb{X}}
\newcommand{\xSpace}{\mathcal{X}}
\newcommand{\ySpace}{\mathcal{Y}}

\newcommand{\adjASup}{(A)}
\newcommand{\adjBSup}{(B)}
\newcommand{\xValue}{u} % arbitrary value in the space, not element of \vectPriv
\newcommand{\xValueA}{\xValue^{\adjASup}}
\newcommand{\xValueB}{\xValue^{\adjBSup}}
\newcommand{\valPriv}[1]{\vectPriv_{#1}}
\newcommand{\vectPrivSymb}{\mathsf{x}} 
\newcommand{\vectPriv}{\vectPrivSymb}
\newcommand{\vectPrivNoisySymb}{\tilde{\vectPrivSymb}}
\newcommand{\vectPrivNoisy}{\vectPrivNoisySymb}
\newcommand{\valPrivNoisy}[1]{\vectPrivNoisy_{#1}}
\newcommand{\avg}{\bar{\vectPriv}}

\newcommand{\transitionset}{\mathcal{W}}
\newcommand{\iterCnt}{T}

\newcommand{\iterSet}{[\iterCnt]}

\newcommand{\initIterSet}{[0,\iterCnt]}

\newcommand{\anIter}{t}

\newcommand{\edgeSymb}{E}
\newcommand{\iterSimpEdges}[1]{\edgeSymb_{#1}}
\newcommand{\iterSimpHiddenEdges}[1]{\iterSimpEdges{#1}^\honestSymb}
\newcommand{\flatGraph}{G}
\newcommand{\transMatHidden}{\aTransMat^H}
\newcommand{\transMatHiddenEl}[2]{\aTransMat^H_{#1,#2}}
\newcommand{\flatGraphHidden}{\flatGraph^\honestSymb}

\newcommand{\outNeigh}[2]{N^{(#2)}_{#1\to}}
\newcommand{\incNeigh}[2]{N^{(#2)}_{#1\gets}}

\newcommand{\exec}{\mathcal{E}}

\newcommand{\onlDec}[1]{#1^\onlSymb}
\newcommand{\aTransMat}{W}

\newcommand{\transMat}[1]{\aTransMat_{#1}}
\newcommand{\transMatOnl}[1]{\onlDec{\transMat{#1}}}
\newcommand{\transMatOnlEl}[3]{\onlDec{\transMatEl{#1}{#2}{#3}}}

\newcommand{\transMatEl}[3]{\aTransMat_{#1;#2,#3}}

\newcommand{\transMatTab}{\transMatEl{\anIter}{\partyA}{\partyB}}
\newcommand{\cancelPhase}{\mbox{Mixing}}
\newcommand{\initPhase}{Initialization}
\newcommand{\dissemPhase}{Dissemination} 

\newcommand{\valPart}[2]{z_{#1,#2}}

\newcommand{\partDistrSymb}{\mathcal{D}}
\newcommand{\partDistr}[1]{\partDistrSymb(#1)}
\newcommand{\earlyInjDecor}{\mathrm{EI}}
\newcommand{\incremInjDecor}{\mathrm{Inc}}
\newcommand{\partDistrEarlySymb}{\partDistrSymb_{\earlyInjDecor}}
\newcommand{\partDistrEarly}[1]{\partDistrEarlySymb(#1)}
\newcommand{\partDistrIncremSymb}{\partDistrSymb_{\incremInjDecor}}
\newcommand{\partDistrIncrem}[1]{\partDistrIncremSymb(#1)}
\newcommand{\partDistrOnlSymb}{\partDistrSymb^\onlSymb}
\newcommand{\partDistrOnlParty}[1]{\partDistrOnlSymb(#1, :)}
\newcommand{\partDistrOnl}[2]{\partDistrOnlSymb(#1, #2)}
\newcommand{\partDistrGaussSymb}{\partDistrSymb_{Gauss}}
\newcommand{\partDistrGauss}[3]{\partDistrGaussSymb^{(#2,#3)}(#1)}

\newcommand{\probDistrSpace}[1]{\mathcal{P}\left(#1\right)}

\newcommand{\yVecSymb}{y} 
\newcommand{\yVec}{\yVecSymb}
\newcommand{\yVecIter}[1]{\yVecSymb^{(#1)}}
\newcommand{\yVecIterPrime}[1]{\yVecSymb^{(#1)\prime}}
\newcommand{\yValPrime}[2]{\yVecSymb_{#1}^{(#2)\prime}}
\newcommand{\yVal}[2]{\yVecSymb_{#1}^{(#2)}}
\newcommand{\indNoiseVec}{\noiseMatrix^\indNoiseDecor}
\newcommand{\bigVecPart}[1]{\zeta_{#1}}

\newcommand{\matrixPart}{Z}
\newcommand{\matrixPartOnl}[1]{\matrixPart^{(#1)}}
\newcommand{\matrixPartOnlEl}[3]{\matrixPart^{(#1)}_{#2,#3}}
\newcommand{\matrixPartSpecialOnlEl}[3]{\bar{\matrixPart}^{(#1)}_{#2,#3}}
\newcommand{\matrixPartSpecialOnl}[1]{\bar{\matrixPart}^{(#1)}}
\newcommand{\matrixPartEl}[2]{\matrixPart_{#1,#2}}

\newcommand{\xCoeffVec}{c}
\newcommand{\xCoeff}[1]{\xCoeffVec_{#1}}
\newcommand{\xCoeffOnl}[2]{\xCoeffVec^{(#1)}_{#2}}
\newcommand{\xCoeffOnlVec}[1]{\xCoeffVec^{(#1)}}
\newcommand{\totalWeightVec}{w} 
\newcommand{\totalWeightParty}[1]{\totalWeightVec_{#1}}
\newcommand{\noiseMatrix}{\eta}
\newcommand{\noiseMatrixEl}[2]{\noiseMatrix_{#1,#2}}
\newcommand{\xMatSymb}{B}
\newcommand{\nMatSymb}{A}
\newcommand{\xMat}{\xMatSymb^{*}}
\newcommand{\nMat}{\nMatSymb^{*}}
\newcommand{\xMatIter}[1]{\xMatSymb_{#1}}
\newcommand{\nMatIter}[1]{\nMatSymb_{#1}}
\newcommand{\aggPartySymb}{Y}
\newcommand{\aggParty}[1]{\aggPartySymb^{(#1)}}
\newcommand{\aggPartyEl}[3]{\aggPartySymb_{#1,#2}^{(#3)}}

\newcommand{\honestSymb}{H}
\newcommand{\corrSymb}{\corrnodes}
\newcommand{\propDrop}{\gamma}
\newcommand{\propCor}{\rho}
\newcommand{\honestScalar}[1]{#1_\honestSymb}
\newcommand{\honestVector}[1]{#1^\honestSymb}
\newcommand{\corrVector}[1]{#1^\corrSymb}
\newcommand{\vectPrivHonest}{\honestVector{\vectPriv}}
\newcommand{\vectPrivCorr}{\corrVector{\vectPriv}}
\newcommand{\cancelNoiseVecHonest}{\honestVector{\cancelNoiseVec}}
\newcommand{\cancelNoiseVecCorr}{\corrVector{\cancelNoiseVec}}
\newcommand{\yVecCorr}{\yVecSymb_{\advView}}

\newcommand{\xMatCorr}{\xMatSymb}
\newcommand{\nMatCorr}{\nMatSymb}
\newcommand{\partySetHonest}{\partySet^\honestSymb}
\newcommand{\coalSet}{J}
\newcommand{\nHonest}{\honestScalar{n}}
\newcommand{\hidden}{\mathcal{H}}

\newcommand{\nObs}{m}
\newcommand{\advValView}{\advView_{val}}
\newcommand{\varNoise}{\Sigma_{\noiseMatrix}}
\newcommand{\corrParty}{c}
\newcommand{\nNoise}{k}
\newcommand{\nNoiseHonest}{\honestScalar{\nNoise}}
\newcommand{\indNoiseVecH}{\noiseMatrix^{\indNoiseDecor \honestSymb}}
\newcommand{\indNoiseVecCorr}{\noiseMatrix^{\indNoiseDecor \corrSymb}}

\newcommand{\neighDatasetA}{\vectPriv^{\adjASup}}
\newcommand{\neighDatasetB}{\vectPriv^{\adjBSup}}
\newcommand{\indicator}[1]{\mathbb{I}\left[ #1 \right]}
\newcommand{\advViewP}[1]{\advView\left(#1\right)} % Parametrized view of the 
\newcommand{\advViewExecFunc}[1]{\advView_{#1}}
\newcommand{\advViewExec}[2]{\advViewExecFunc{#1}\left(#2\right)}
\newcommand{\xMatCorrCol}{h}

\newcommand{\dummyVec}{v}
\newcommand{\dummyVecA}{\dummyVec_A}
\newcommand{\dummyVecB}{\dummyVec_B}
\newcommand{\dummyTailBoundA}{w}
\newcommand{\dummyTailBoundB}{\gamma}
\newcommand{\sdTailBoundA}{\sigma_\dummyTailBoundA}
\newcommand{\varMatProd}{\Sigma}

\newcommand{\etaTransEl}[2]{G^{(#1,#2)}}
\newcommand{\edgeVec}[2]{a^{(#1,#2)}}
\newcommand{\edgeVecHat}[2]{\hat{a}^{(#1,#2)}}

\newcommand{\matLemmaA}[2]{\matLemmaASymb^{(#1,#2)}}

\newcommand{\matLemmaASymb}{M}

\newcommand{\noisyVec}{\tilde{\vectPriv}}

\newcommand{\noiseDiff}{\xi}
\newcommand{\indNoiseDiff}{\noiseDiff^\indNoiseDecor}
\newcommand{\cancelNoiseDiff}{\Delta}
\newcommand{\partDistTerm}[1]{\eta_{#1}}
\newcommand{\stripDecor}{{(:)}}
\newcommand{\cancelNoiseDiffVec}{\cancelNoiseDiff_\stripDecor}

\definecolor{darkblue}{rgb}{0.0, 0.0, 0.55}

\begin{document}
	
% If your paper is accepted and the title of your paper is very long,
% the style will print as headings an error message. Use the following
% command to supply a shorter title of your paper so that it can be
% used as headings.
%
%\runningtitle{I use this title instead because the last one was very long}

% If your paper is accepted and the number of authors is large, the
% style will print as headings an error message. Use the following
% command to supply a shorter version of the authors names so that
% they can be used as headings (for example, use only the surnames)
%
%\runningauthor{Surname 1, Surname 2, Surname 3, ...., Surname n}

%%
%% The "title" command has an optional parameter,
%% allowing the author to define a "short title" to be used in page headers.
\title[Dropout-Robust Robust Mechanisms for DP-DME]{Dropout-Robust Mechanisms for Differentially Private and Fully Decentralized Mean Estimation}

%%%%%%%%%%%%%%%% Authors' Info %%%%%%%%%%%%%%%%%
%%
%% The "author" command and its associated commands are used to define
%% the authors and their affiliations.

\author{C\'{e}sar Sabater}
%\orcid{1234-5678-9012}
\affiliation{%
	\institution{Insa-Lyon}
	\city{Lyon}
%	\state{Ohio}
	\country{France}}
\email{cesar.sabater@insa-lyon.fr}

\author{Sonia Ben Mokhtar}
%\orcid{1234-5678-9012}
\affiliation{%
	\institution{Insa-Lyon, CNRS}
	\city{Lyon}
	\country{France}}
\email{sonia.ben-mokhtar@cnrs.fr}

\author{Jan Ramon}
\affiliation{%
	\institution{Inria Lille}
	\city{Lille}
	\country{France}
}
\email{jan.ramon@inria.fr}

  % Making (Fully) Decentralized Learning Simultaneously Private and Accurate, even when  Adversaries are Knowledgeable

\renewcommand{\shortauthors}{Sabater et al.}

\begin{abstract}
 %Obtaining differentially private computations in a decentralized setting presents several challenges.
%A first challenge is to achieve acceptable accuracy and communication cost.
%  A second challenge is to ensure robustness to information leakage from exchanged messages.
%  To deal with these challenges, cryptographic techniques offer promising solutions, but in presence of network failures they often incur a high communication overhead or require a central party to orchestrate the computation. Instead, fully decentralized solutions significantly relax adversarial model assumptions or rely on pairwise canceling noise, which suffers from substantial accuracy degradation if parties unexpectedly disconnect during computation. 
  
Achieving differentially private computations in decentralized settings poses significant challenges, particularly regarding accuracy, communication cost, and robustness against information leakage. While cryptographic solutions offer promise, they often suffer from high communication overhead or require centralization in the presence of network failures. Conversely, existing fully decentralized approaches typically rely on relaxed adversarial models or pairwise noise cancellation, the latter suffering from substantial accuracy degradation if parties unexpectedly disconnect.  
  In this work, we propose \inca{}, a new protocol for fully decentralized mean estimation, a widely used primitive in data-intensive processing. Our protocol, which enforces differential privacy, requires no central orchestration and employs low-variance correlated noise, achieved by incrementally injecting sensitive information into the computation. First, we theoretically demonstrate that, when no parties permanently disconnect, our protocol achieves accuracy comparable to that of a centralized setting—already an improvement over most existing decentralized differentially private techniques. Second, we empirically show that our use of low-variance correlated noise significantly mitigates the accuracy loss experienced by existing techniques in the presence of dropouts.
\end{abstract}

%%
%% Keywords. The author(s) should pick words that accurately describe
%% the work being presented. Separate the keywords with commas.

\maketitle

\section{Introduction}

%\myitem The need for decentralized push-based learning.
%\myitembb tweets on X, communication in space, \ldots
%\myitem Knowledgeable adversaries.
%\myitembb on X, what if somebody knows all your followers?
%\myitembb We assume adversaries are placed randomly in the network: in practice, a 'spy' may be able to facilitate becoming 
%Comment: A possibility is to leave the discussion around push-based protocols later in the paper and instead remain generic in the introduction.

The training of machine learning models commonly relies on the centralized processing of large datasets. While this approach offers simplicity, it presents substantial privacy implications when sensitive data is involved, or when data aggregation is limited by sharing constraints. For instance, private companies or hospitals might be reluctent to share their data for jointly training prediction models due to competitive secrecy or juridic regulations.
For these reasons, algorithms where the data remains local to their owners have gained significant popularity in the last decade. One prominent example is Federated Learning~\cite{mcmahan2017communication, kairouz2021advances}, a framework in which participating entities iteratively perform local model training and subsequently transmit intermediate model updates to a central server for aggregation. While this architecture enhances data governance, it nonetheless presents residual challenges in terms of privacy and robustness. Indeed, computational processes relying on a large participant base introduce vulnerabilities to participant behavior, including failures or corruption by adversaries. This may compromise both the integrity of the outcome and participant privacy through inference attacks~\cite{melis2019exploiting,geiping2020inverting,mrini2024privacy,touat2024scrutinizing,kariyappa2023cocktail,pasquini2023security}. Furthermore, the centralized visibility of all model updates to a single party in federated learning establishes a high-risk single point of failure. 

To address the aforementioned limitations, entirely decentralized machine learning algorithms have been developed~\cite{hegedHus2021decentralized,beltran2023decentralized}. These solutions present a promising alternative to their federated counterparts, as they remove the need for any centralized entity. In this context, decentralized averaging algorithms, which are the primary focus of this paper, constitute a fundamental component enabling parties to average their model parameters or gradients towards model convergence.
Specifically, we concentrate on the canonical task of averaging a set of private values held by participants. Despite its apparent simplicity, averaging serves as a critical primitive in decentralized learning and various other machine learning and data mining tasks, including recommendation systems\cite{nicolas2025secure}, clustering, matrix factorization\cite{nicolas2024differentially}, decision trees, empirical cumulative distribution functions\cite{barczewski2025differentially} and linear regression. Fundamentally, its applicability extends to any task amenable to decomposition into local computations followed by private averaging.

Our objective is then to analyze the privacy of decentralized averaging algorithms through the lens of \emph{differential privacy (DP)}~\cite{dwork2006calibrating}. This framework has emerged as a gold standard in privacy research due to its solid theoretical guarantees. Nevertheless, demonstrating that an algorithm satisfies DP is often challenging, typically necessitating the injection of noise into computations, which degrades accuracy. This challenge is particularly acute within the local differential privacy (LDP) framework~\cite{duchi2013local,kasiviswanathan2011can,kairouz2014extremal,kairouz2016discrete,chen2020breaking}. In LDP, where data is fully privatized before it is used in collaborative computations,
%where no central authority is trusted with participants' private data,
the required noise levels are frequently prohibitively high, rendering computations practically unusable in many real-world scenarios.
The integration of cryptographic primitives, such as secure multiparty computation (specifically, Secure Aggregation \cite{bonawitz2017practical,bell2020secure}) or fully homomorphic encryption, into decentralized computations offers a promising solution to substantially reduce the noise levels required for achieving differential privacy (DP) guarantees. This approach can yield privacy-accuracy trade-offs comparable to those observed in central DP settings \cite{agarwal2021skellam,kairouz2021distributed,heikkila2025using,chen2022poisson,jayaraman2018distributed,dwork2006our}, where a single trusted party or curator is responsible for noise injection. However, a significant drawback of these cryptographic methods is their high computational cost, particularly as the number of involved parties increases, or their reliance on a trusted central entity for orchestration and fault tolerance (e.g., handling participant drop-outs).

In the decentralized setting, several techniques have been proposed to reduce the noise of local DP. These solutions rely on gossip algorithms in which exchanges are similar to the protocols proposed in \cite{boyd2006randomized}. However, they often rely on DP relaxations \cite{cyffers2022privacy,cyffers2022muffliato}, which can limit practical applicability, or necessitate correlated bounded noise (e.g., pairwise canceling Gaussian noise \cite{sabater2022accurate,allouah2024privacy,vithana2025correlated}). The latter approaches either incur high communication costs or demand the cancellation of high-variance noise, which substantially impacts computational efficiency, particularly in the presence of dropouts.

In our work, we propose protocols that (i) have the same privacy accuracy trade-offs of Central DP when no drop-outs are present, which is similar to the trusted curator model of DP, (ii) satisfy the classical notion of DP, (iii) are robust to a percentage of parties colluding to infer information of the honest participants and (iv) have a bounded impact in drop-outs thanks to the use of low variance correlated noise.

\subsection{Contributions}

Our contributions are the following: 
\begin{enumerate}
	\item We propose novel protocol for private averaging.
	It involves iterative exchanges of  messages in a way similar to the  gossip protocols in~\cite{boyd2006randomized} with the addition of correlated noise. The latter  ensures privacy with a minimal harm in accuracy. We provide a generic construction in which the private values as well as the noise injected in the computation are customizable, allowing  to choose the most adequate strategy for different scenarios. Our approach differs from the common technique of previous correlated noise approaches \cite{sabater2022accurate,allouah2024privacy,vithana2025correlated} and Secure Aggregation  \cite{dwork2006our,bonawitz2017practical,bell2020secure,taiello2024let,bell2023acorn}, that are based in pairwise noise masks shared between users. In our approach, each user knows its own noise terms and cancels them in progressive updates.  
	\item When parties do not drop out of the protocol, we prove differential privacy guarantees even if the adversary (i) corrupts a proportion of the participants to passively share their gathered information, (ii) observes a proportion the exchanges, and (iii) knows all the network interactions (i.e., who communicated with whom). 
	While proven to be differentially private, our protocol matches the utility of Secure Aggregation combined with local noise  \cite{agarwal2021skellam,kairouz2021distributed,heikkila2025using,chen2022poisson}  without the use cryptographic primitives or a requiring a central orchestrator. This privacy-utility trade-off is comparable to Central DP. 
	We provide two characterizations of our privacy guarantees: the first is more accurate while the second is more interpretable. For the latter, we show that our guarantees hold when a sufficiently large and diverse set of exchanges have been performed between honest parties.
	\item We prove that, given the graph that models the hidden exchanges between honest parties, certain topologies provide sufficient conditions for privacy or the lack of it. We provide positive  and negative results. The former are for strongly connected topologies.  The latter are for graphs where no party changes their neighbors across iterations and the adversary continually observes at least two honest parties. 
%	\janfoot{The claim is too strong, e.g., is not true for the complete graph.} \cesarfoot{we do prove that static graphs make our protocol not to meet the preconditions of our privacy theorems if at least 2 nodes are continually observed.. It really depends on which complete graph are we talking about, perhaps we have a misunterstanding here.}

%	 If $\nHonest$ is the number of non-corrupted participants, our characterization is based on $\nHonest-1$ linearly independent vectors where the set of vectors that is strongly related to the columns of the weighted adjacency matrices of the decentralized exchanges. \janfoot{The sentence has a grammar problem.  I'm not sure the reader understands what adjacency matrix you're talking about.}
	\item We extend our privacy analysis to the case of dropouts, where correlated noise does not cancel as when all participants finish their contribution correctly. We show that our approach is especially resilient to dropouts when the injection of private values is done incrementally among the noise, therefore keeping correlated DP noise  smaller in size.
	The analysis is done both in theory and in practice. First, we show theoretically  that DP noise that is not supposed to cancel is bounded. Next we show empirically that correlated noise is smaller than with existing techniques in the same setting \cite{sabater2022accurate, cyffers2022muffliato, vithana2025correlated}. 
%	\item Finally, we empirically  show that our protocol exhibits similar communication cost than other correlated noise techniques in the same setting  \cite{sabater2022accurate, vithana2025correlated} and works with correlated noise of smaller variance. Moreover, we show that increasing the communication allows to further reduce the variance of correlated noise, which improves the resilience to dropouts. 
%	
	
%	especially when executed in communication graphs that are matchings of nodes. \janfoot{The reader may not understand 'Matchings of nodes' since usually Gossip algorithms require sufficient mixing.  In fact, your phrase probably just refers to the communication graph of a single iteration. }
\end{enumerate} 

%\sonia{Explain that DP in decentralized settings have been explored (e.g., one may use local DP or Network DP). However, local DP results in a weak privacy/utility tradeoff while network DP has relaxed assumptions. The objective in this paper is to...}
\subsection{Structure of the paper}
The rest of the paper is organized as follows. We first present preliminaries (Section~\ref{sec:prelim}). Then we present our protocol (Section~\ref{sec:protocol}) and its privacy analysis (Section~\ref{sec:privacy}). We further present our empirical evaluation (Section~\ref{sec:exp}) before presenting related work (Section~\ref{sec:relatedwork}) and concluding the paper (Section~\ref{sec:conclusion}).

% !TeX root = main.tex

\section{Preliminaries}
\label{sec:prelim}
We denote the set of integers between $a$ and $b$ with $[a,b]\coloneq \{a, a+1 \ldots b\}$ and the set of first $k$ positive integers by $[k] \coloneq [1, k]$. We define $\oneVec \coloneq (1,\dots, 1)^\top$ and $\project{i}$ is a vector of all $0$s except for the $i$th coordinate which is equal to $1$. The dimensions of $\oneVec$ and $\project{i}$ can be inferred from the context. Unless explicitly stated differently, all vectors are column vectors. $\indicator{\cdot}$ is the indicator function, i.e., $\indicator{True}=1$  and $\indicator{False}=0$.  For an $n \times m$ 
matrix $M$ and subsets of indices $J\subseteq [n]$ and $K \subseteq [m]$, $M_{J,K}$ is the matrix obtained from $M$ by taking the subsets of rows and columns indicated by $J$ and $K$ respectively. This extends analogously to vectors and other kinds of tensors. When used as sub-indexes `$:$' means the set of all indices and $-i$ means the set of all indices minus $i$, e.g., % (some examples are
$M_{:,K} = M_{[n], K}$ %, $M_{i, :} = M_{\{i\}, [m]}$
and $M_{-i, K} = M_{[n]\setminus\{i\}, K}$. When we use ``$:$'' instead of a column index of a matrix, the result is a row vector (e.g., $M_{i,:}$ is a row vector). 
%the $i$-th row of matrix $M$ and  $M_{i,:}^\top$ is a column vector. %). 

\subsection{Problem Statement}
\label{sec:prelim.problem}

We consider a set of parties $\partySet=[\partyCnt]$. Each party $\aParty \in \partySet$ has a private value $\valPriv{\aParty}\in\xSpace$ where $\xSpace$ is a convex set in a vector space $\xVectSpace$.
%$x_i \in [0,1] \subset \R$.
These parties want to collaboratively estimate ${\avg}=\frac{1}{n}\sum_{\aParty \in \partySet} \valPriv{\aParty}$ while keeping each value $\valPriv{\aParty}$ (differentially) private.
%\janfoot{Add if we need it in the end:
	While for simplicity of our explanation we average scalars, our approach can be easily extended %applied without loss of generality
	to vectors (such as machine learning models or gradients).
	%}

\subsection{Differential Privacy}
\label{sec:prelim.dp}

We evaluate the privacy of our protocols using the differential privacy framework (DP) \cite{Dwork2014}. This framework allows one to quantify how distinguishable is the output of a randomized algorithm on two \emph{neighboring} datasets. 
More formally, a dataset is a vector of elements of $\xSpace$, we denote the space of all datasets by $\xSpace^*$.  Two datasets $\neighDatasetA,\neighDatasetB\in\xSpace^*$ are neighboring if there are $\vectPriv^{(0)}\in\xSpace^*$, permutation matrices $Q^{\adjASup}$ and $Q^{\adjBSup}$, and $\xValueA,\xValueB\in\xSpace$ such that $\neighDatasetA=Q^{\adjASup}\left[\begin{array}{c}\xValueA \\ \vectPriv^{(0)} \end{array}\right]$ and $\neighDatasetB=Q^{\adjBSup} \left[\begin{array}{c}\xValueB \\ \vectPriv^{(0)}\end{array}\right]$.
%\cesarfoot{I don't really understand the permutation matrices here. I think we could avoid this notation. Why dont'} 

\begin{definition}
	Let $\epsilon > 0$ and $\delta \in [0,1]$. A randomized algorithm $\alg:\xSpace^*\to\ySpace$ is $(\epsilon, \delta)$-\emph{DP} if for any pair of neighboring datasets $D, D'$ and any subset $\algout\subseteq \ySpace$ of possible outputs
	%\janfoot{Should we distinguish between 'output' and 'outcome'?  We could systematically use 'output', or we could use one term for the intended output and the other for everything an adversary can observe.}
	we have that 
	\begin{equation}
		\Pr ( \alg(D) \in \algout ) \le \exp(\epsilon) \Pr ( \alg(D') \in \algout) + \delta 
		\label{eq:dp.def} 
	\end{equation}
	where the probability is taken over the randomness of  $\alg$. 
	\label{def:dp}
\end{definition} 

In our problem, we set $\xSpace=[0,1]$ for the simplicity of our explanation.  Then, a dataset is a vector $\vectPriv = (x_\aParty)_{i=1}^n \in [0,1]^n$.
%\cesarfoot{actually, our dataset is the vector of private values of honest users. Introduce this either here, or present the threat model (and therefore communication model) before DP}
%\cesarfoot{here we instantiate $\xSpace = [0,1]$ for simplicity. We should say that this does not imply a loss of generality} .
%of private values. % and two datasets are neighboring if the differ in a single component. 
%To achieve DP guarantees, computations can be perturbed with controlled calibrated according to the desired guarantees. 
The output of $\alg$ is the set of all observations an adversary makes, which may exceed the intended input and could include intercepted messages or information obtained from corrupted parties.  This output therefore depends on the threat model considered.  
%The view $\advView$ of an external adversary (i.e., the output of $\alg$) is determined by the trust model.
Consider a particular party $\aParty\in\partySet$.
One can consider two extreme cases.  First, if the threat model allows for all parties $\partySet\setminus \{\aParty\}$ to be corrupted, then the best $\aParty$ can do is to privatize $x_\aParty$ before starting anything, and the best utility which can be obtained for a given privacy level $(\epsilon,\delta)$ is the utility of Local DP  \cite{ldp}.  Second, if the threat model makes more favorable assumptions, the best possible utility which can be reached for a given privacy level is the utility of Central DP, where there exist a trusted entity to whom all parties can send their messages and the adversary cannot see them until this entity releases his output. 

\paragraph{Gaussian Mechanism} Our protocols use Gaussian noise to achieve differential privacy. Using the Gaussian Mechanism\cite{dwork2014algorithmic}, it is possible to compute $(\epsilon, \delta)$-DP estimations of averages with a mean squared error of $2\ln(1.25/\delta)/\epsilon^2 n^2$ in the central model. In the local model this error is $n$ times bigger.
%\cesar{Do we add the gaussian mechanism here?}\janfoot{yes, but we don't claim yet that we will always use Gaussians.}
%\janrm{In the two extremes of the trust spectrum lie Local DP} \cite{ldp}  \janrm{and Central DP} \cite{cdp}\janrm{. In the former all messages of the protocol are included in }$\advView$\janrm{That may not be necessary to get in the LDP regime.} \janrm{while the latter assume that  a \emph{trusted curator}}\janfoot{not necessarily, see e.g., GOPA who achieve Central DP without trusted curator.} \janrm{exist and whose output element of $\advView$.  Algorithms that satisfy Local DP assume minimal trust but  provide estimates of $O(n)$ times more variance than in Central DP for the same $(\epsilon,  \delta)$-DP guarantees} \cite{cite different mechanisms}.
\subsection{Threat Model}
\label{sec:prelim.threat}

\newcommand{\msgAbs}[1]{y_{#1}}
\newcommand{\transcript}{\mathcal{T}}
\newcommand{\msgCnt}{M}
\newcommand{\msgVal}{v}
\newcommand{\obsDPLabel}{Eavesdrop DP}
\newcommand{\obsDPAcron}{E-DP}
\newcommand{\obsDPLarge}[3]{(#1, #2, #3)-\obsDPLabel{}}
\newcommand{\obsDP}[3]{(#1, #2, #3)-\obsDPAcron{}}
\newcommand{\colDPLabel}{Collussion DP}
\newcommand{\colDPAcron}{C-DP}
\newcommand{\colDPLarge}[3]{(#1, #2, #3)-\colDPLabel{}}
\newcommand{\colDP}[3]{(#1, #2, #3)-\colDPAcron{}}

In our work, we assume all parties are \emph{semi-honest} in that they follow the protocol as prescribed. A stochastic decentralized algorithm $\alg$ that takes a dataset $D$ produces a transcript $\transcript(D) = ( \msgAbs{k} )_{k=1}^{\msgCnt}$ of all $\msgCnt$ messages exchanged by the protocol, where  $\msgAbs{k} = ((\partyA, \partyB),\msgVal ) \in (\partySet \times \partySet) \times \{0,1\}^*$ means that party $\partyA$ sent value $\msgVal$ to $\partyB$. We assume that the adversary knows the tuple $\transitionset = (p)_{ (p,v) \in \transcript(D)}$ of pairs sender-receiver of all of user interactions, e.g., as they may corrupt participants or track network traffic. We assume that parties communicate via secure channels, however, we consider that the adversary may succeed to break this security and observe a subset of messages  $\advValView \subseteq \transcript(D)$, e.g., by intercepting messages, exploiting side channels or other attacks. 

\begin{definition}[\obsDPLabel{}]
	We say that a decentralized algorithm $\alg$ satisfies \obsDPLarge{$\epsilon$}{$\delta$}{$\advValView$} (or \obsDP{$\epsilon$}{$\delta$}{$\advValView$}) if it satisfies $(\epsilon, \delta)$-DP for the output $\alg(D) =  (\transitionset,\advValView)$.  
\end{definition}  
Note that  \obsDP{$\epsilon$}{$\delta$}{$\transcript(D)$}  (i.e. where all messages are known by the adversary) is equivalent to Local DP. We also assume an adversary that may corrupt some subset of parties $\corrnodes\subset \partySet$.  If a party $\corrParty \in \corrnodes$ is corrupted, the adversary can learn both its private value $\valPriv{\aParty}$ and its incoming/outgoing messages. In that case, we denote by $\partySetHonest = \partySet \setminus \corrnodes$ to the set of $\nHonest = n - |\corrnodes|$ honest parties that are not controlled by the adversary. 
\begin{definition}[\colDPLabel{}]
	We say that a decentralized algorithm $\alg$ is \colDPLarge{$\epsilon$}{$\delta$}{$\corrnodes$} (or \colDP{$\epsilon$}{$\delta$}{$\corrnodes$}) if it is $(\epsilon, \delta)$-DP for the output $\alg(D^H) =  (\transitionset,\advValView)$, where the input dataset $D^H$ is the set of private values of honest parties $(\valPriv{\aParty})_{\aParty \in \partySetHonest}$ and $\advValView = \{ ((\partyA, \partyB), v) \in \transcript(D) : \partyA \in \corrnodes \text{ or }  \partyB \in \corrnodes  \}$.
\end{definition}  
We remark that parties executing algorithm $\alg$ do not have knowledge of the observed messages $\advValView$ or corrupted  parties $\corrnodes$. 
%\obsDPAcron{} and \colDPAcron{} can be combined if the adversary is hybrid and simultaneously observes messages and corrupt parties. 

The view of the adversary in  \obsDPAcron{} and \colDPAcron{} is similar to other correlated noise techniques such as the presented in \cite{sabater2022accurate,allouah2024privacy} and in Secure Aggregation \cite{bell2020secure}, where the set of interactions $\transitionset$ is known. It is stronger than the view of the adversary in Network DP \cite{cyffers2022privacy} which only knows  $\advValView$ and ignores $\transitionset$. We remark that the knowledge of $\transitionset$ is a dangerous piece of  information as it can be exploited to completely compromise privacy~\cite{mrini2024privacy}. Our adversary  is also stronger than Pairwise Network DP \cite{cyffers2022muffliato} where $\transitionset$ is also unknown and the reported privacy loss of a party is the average over all possible placements of the adversary, whereas it corresponds to the worst-case in \colDPAcron{}.

\subsection{Communication Model}
\label{sec:prelim.comm}

We study  decentralized protocols in the synchronous setting as defined in~\cite{guerraoui2006introduction}. In this setting, parties perform exchanges in $\iterCnt$ iterations and there exists a known finite time bound for a message to reach its destination.

%\janfoot{Here the question arises whether you want to get into the big notation reorganization discussion.  For example, $q$ is not a very natural letter.  It would be nicer to have $T$ iterations and then $t\in[T]$ can be an index.  But then $T$ can not anymore be a transition matrix, so we would need to use another letter.  This may lead to a cascade of notation reorganizations to get the most natural / easy to remember notations, but would require some work from our side to adapt.}
At each iteration $\anIter \in \iterSet$, all parties wake up and interact, among others sending a message to a set of neighbors.  We first model this communication structure.  For every iteration $\anIter \in\iterSet$, let $\iterSimpEdges{\anIter}\subseteq\partySet\times\partySet$ such that $(\partyA,\partyB) \in \iterSimpEdges{\anIter}$ if and only if party $\partyA$ sends a message to $\partyB$ in iteration $\anIter$.  If $\iterSimpEdges{\anIter}=\iterSimpEdges{1}$ for all $\anIter\in\iterSet$, we will say that our protocol has \emph{static} exchanges, otherwise these are \emph{dynamic}. For $\aParty\in\partySet$ and $\anIter \in [\iterCnt]$, let $\incNeigh{\partyA}{\anIter} = \{\partyB \mid (\partyB,\partyA)\in E_t\}$ and $\outNeigh{\partyA}{\anIter} = \{\partyB \mid (\partyA,\partyB)\in E_t\}$ to be respectively the sets of incoming and outgoing neighbors of $\aParty$.

%A transition matrix $\transMat{\anIter}$ is \emph{doubly stochastic} if $\transMat{\anIter} \mathbbm{1} = \mathbbm{1}$ and $\mathbbm{1}^\top \transMat{\anIter} = \mathbbm{1}^\top$. For $t_2\ge t_1$ let $\transMatInt{\iterA}{\iterB} \coloneq \transMat{\iterB} \transMat{\iterB-1}  \dots \transMat{\iterA}$. If transition matrices $\transMat{1}, \dots, \transMat{\iterCnt}$ are doubly stochastic, an upper bound of the approximation error of the estimate    is given by  \cesarfoot{add why a doubly stochastic matrix converges to the average} 
%
%\[ 
%\| \transMatInt{1}{\iterCnt} \vectPriv - \mathbbm{1}^\top \avg  \|_2^2 \le \left( \prod_{\anIter=1}^\iterCnt \lambda_{2,\anIter}  \right)  \|\vectPriv - \mathbbm{1}^\top \avg \|_2^2 
%\]
%where for all $\anIter \in \iterSet$, $\lambda_{2,\anIter}$  is the second largest eigenvalue of $\transMat{\anIter}^\top \transMat{\anIter}$  \cite{boyd2006randomized} (the largest eigenvalue being $1$). 

\subsection{Failures} 
\label{sec:prelim.dropouts} 
We will consider that, at each iteration, parties could be unexpectedly absent of the computation due to a temporary or permanent crash, disconnection from the network, or other problems. We call this behavior a \emph{dropout}. We denote by $\onlNodes{\anIter} \in \partySet$ to the set of parties that did not drop out in the computation at iteration $\anIter \in [\iterCnt]$. We consider that a party $\aParty \in \partySet$ dropped out \emph{permanently} if he is not present in the last iteration (i.e., $\aParty \not\in \onlNodes{\iterCnt})$. We assume that messages have bounded delays~\cite{guerraoui2006introduction} and that there is sufficient time in an iteration that parties can confirm they received a message. Therefore, if a party $\aParty$ send a message to $\partyB$ in an iteration where $\partyB$ dropped out, $\aParty$ will detect it.

%In our work, we consider that a fixed set of nodes $\corrnodes \subset \partySet$ is corrupted by an adversary. These nodes are \emph{semi-honest} as they can collude to exchange all the information they gather but follow the protocol correctly. We also assume that the adversary knows the set  $\transitionset = \{W_1, \dots, W_T\}$ of transition matrices, e.g., by corrupting participants or track network traffic. Let $\msg{i}{j}$ be the set of messages sent from party $i$ to party $j$ during the execution.   We focus on proving that the view
%  \[    \advView = \transitionset \cup \{ x_\corrParty  \}_{\corrParty \in \corrnodes}  \cup \{ \msg{\partyA}{\partyB} \mid \partyA \in \corrnodes \text{ or } \partyB \in \corrnodes \}_{\partyA, \partyB \in \partySet} \] 
%of the aversary is differentially private,.

% !TeX root = main.tex

\section{Protocol}
\label{sec:protocol} 
In this section, we present our protocol. The base protocol is described in Section \ref{sec:protocol.base} and the dropout resistant version is presented in Section \ref{sec:protocol.dropouts}.

\subsection{Base Protocol}
\label{sec:protocol.base}  
We introduce \inca{}, our protocol for  incremental averaging  in Algorithm \ref{alg:ourprotocol}. It consist on three phases: \initPhase{}, \cancelPhase{} and \dissemPhase{}.

\paragraph{Protocol} The \initPhase{} Phase is devoted to sample DP noise and compute the initial messages of each party $\aParty \in \partySet$. In particular, $\aParty$ draws a Gaussian noise sample $\indnoise{\aParty}\sim\mathcal{N}(0,\sdInd^2)$ and a $(\iterCnt+1)$-vector %canceling noisy
$\left(\valPart{\aParty}{\anIter}\right)_{\anIter=0}^\iterCnt$ from a noise  distribution $\partDistr{\valPriv{\aParty}+\indnoise{\aParty}}$ (line \ref{alg.line:draw}). For all $u \in \xVectSpace$, $\partDistr{u}$  satisfies the following property:  
\begin{equation} 
	\sum_{\anIter=0}^\iterCnt {v_\anIter} = u, \quad \text{ for all  } \left(v_\anIter \right)_{\anIter=0}^\iterCnt \sim \partDistr{u}. 
	\label{eq:canceling}
\end{equation}
The goal of the $T+1$ terms $\left(\valPart{\aParty}{\anIter}\right)_{\anIter=0}^\iterCnt$ is to inject the value $\valPriv{\aParty}+\indnoise{\aParty}$ into the system incrementally in a randomized way, such that (a) it is hard for an adversary to infer this value and (b) if party $\aParty$ would drop out and would not inject the rest of the value into the system the resulting error remains bounded.
%\janrm{prevent the leakage of sensitive information from intermediate computations/messages  that can be seen by the adversary, without reducing accuracy.} 
The goal of $\indnoise{\aParty}$ is to prevent leakages from the exact average, by outputing only a differentially private approximation of it.  The first message $\yVal{\aParty}{0}$ of  each party $\aParty \in \partySet$ is set to $\valPart{\aParty}{0}$ (line \ref{alg.line:init}). 

Exchanges to ensure privacy are done in the \cancelPhase{} Phase. At each iteration $\anIter \in [\iterCnt]$,  each party  $\aParty \in \partySet$ starts by sending its  message $\yVal{\aParty}{\anIter-1}$ to  outgoing neighbors $\outNeigh{\aParty}{\anIter}$.  In the spirit of gossip algorithms, one can view $\yVal{\aParty}{\anIter-1}$ as $\aParty$'s current estimate of the final output with the information it has so far.
Then, it computes its new message (updated estimate of the final output) $\yVal{\aParty}{\anIter}$ by aggregating incoming messages, his current message and a new part of its private value to inject $\valPart{\aParty}{\anIter}$. For each incoming neighbor $\partyB \in \incNeigh{\aParty}{\anIter}$, message $\yVal{\partyB}{\anIter-1}$  is aggregated with weight  $\transMatEl{\anIter}{\aParty}{\partyB}$ and $\yVal{\aParty}{\anIter}$ with  $\transMatEl{\anIter}{\aParty}{\aParty}$ %. This is described in
(line \ref{alg.line:canceling}). %of the  Algorithm.
At the end of the \cancelPhase{} Phase, each party $\aParty \in \partySet$ will have an estimate $\yVal{\aParty}{\iterCnt}$. By the time iteration $\iterCnt$ is reached, the final messages $(\yVal{\aParty}{\iterCnt})_{\aParty \in \partySet}$ of this phase sufficiently hide the private values due to the gossip mixing and the noise terms $(\indnoise{\aParty})_{\aParty \in \partySet}$ and $(\valPart{\aParty}{\anIter})_{\aParty,\anIter \in \partySet \times [\iterCnt]}$  (see Section \ref{sec:privacy} for details). 
%\janrm{Therefore privacy guarantees will still be present even if the adversary recovers all of them.}\janfoot{It is unclear what ``all of them'' means (all messages, all $z_{i,t}$ \ldots) and it is unclear why this is guaranteed as we didn't introduce sufficient pre-conditions to prove it.  It is better to postpone the argument to a later point.}

In the \dissemPhase{} Phase, since $(\yVal{\aParty}{\iterCnt})_{\aParty \in \partySet}$ is differentially private, the parties can
%use a decentralized protocol to
compute $\frac{1}{n}\sum_{\aParty \in \partySet} \yVal{\aParty}{\iterCnt}$ in the clear.
There are many practical ways to perform this computation. One of them is to use a gossip averaging protocol \cite{boyd2006randomized}, this preserves a similar structure as the previous phases of our protocols.
%weighted aggregation and exchange of our protocols,
%providing convergence guarantees on the spectral properties of transition matrices.  % JR: grammar is incorrect, Let's drop the phrase, it is not essential here.

%  As they have been sufficiently mixed and obfuscated, even if the adversary gains the knowledge of all final messages $\yVal{:}{\iterCnt}$ it would not compromise privacy (see Section \ref{sec:privacy} for details). Therefore, parties can freely share $\yVal{:}{\iterCnt}$. This is precisely what is done in \dissemPhase{} Phase: parties run Algorithm \ref{alg:gossip} with input $(\yVal{\aParty}{\anIter})_{\aParty \in \partySet}$ to obtain $\frac{1}{n}\sum_{\aParty \in \partySet} \yVal{\aParty}{\iterCnt}$. 

\begin{algorithm}[H]
	\begin{algorithmic}[1]
		\STATE{\textbf{Input:} $\iterCnt \in \mathbb{N}$, $\vectPriv = (x_1, \dots, x_n)^\top \in \xSpace^n$, $W_1, \dots, W_{\iterCnt} \in \R^{n\times n}$
			, $\sdInd^2 \in \R$, $\partDistr{\cdot}:\xVectSpace\to\probDistrSpace{\xVectSpace^{T+1}}$}
		\STATE{\textit{\initPhase{} Phase}}
		\FORALL{$\aParty \in \partySet$} 
		\STATE{ Sample  $\indnoise{\aParty} \sim \N(0, \sdInd^2)$ and $(\valPart{\aParty,0} \ldots \valPart{\aParty}{\iterCnt}) \sim \partDistr{x_\aParty+\indnoise{\aParty}}$ \label{alg.line:draw}}
		\STATE{$y^{(0)}_\aParty \gets  \valPart{\aParty}{0}$} \label{alg.line:init}
		\ENDFOR
		\STATE{\textit{\cancelPhase{} Phase}}
		\FOR{$\anIter \in \{1 \dots \iterCnt\}$}
		\FORALL{$\partyA \in \partySet$}
		\STATE{$y^{(\anIter)}_\partyA \gets \left( \sum_{\partyB \in \partySet} \transMatTab y^{(\anIter-1)}_\partyB \right)+  \valPart{\partyA}{\anIter}$} \label{alg.line:canceling} 
		\ENDFOR
		\ENDFOR
		\STATE{\textit{\dissemPhase{} Phase:}
			Parties jointly compute $\frac{1}{n}\sum_{\aParty \in \partySet} \yVal{\aParty}{\iterCnt}$ 
			% 			using transition matrices $W_{\iterCnt+1}, \dots, W_{\iterCnt+\iterCnt'}$ for $\iterCnt'$ iterations
		} 
	\end{algorithmic}		
	\caption{\inca{} Protocol}
	\label{alg:ourprotocol}
\end{algorithm} 

\paragraph{Utility} If for all $\anIter \in [\iterCnt]$ we have that  $\transMat{\anIter}$ satisfies 
\begin{equation}
	\forall \partyB \in \partySet, \quad \sum_{\aParty \in \partySet} \transMatEl{\anIter}{\aParty}{\partyB} = 1 \quad  \text{(column stochasticity)} 
	\label{eq:colstoch} 
\end{equation}
then
\begin{eqnarray}
	\frac{1}{n}\sum_{\aParty \in \partySet} \yVal{\aParty}{\iterCnt} &=& \frac{1}{n}\sum_{\aParty \in \partySet} \valPart{\partyA}{\iterCnt} +  \sum_{\partyB \in \partySet} W_{\iterCnt;i,j} y^{(\iterCnt-1)}_\partyB    \notag \\
	&=& \frac{1}{n} \sum_{\aParty \in \partySet} \valPart{\partyA}{\iterCnt} + \frac{1}{n}  \sum_{\partyB \in \partySet} \left( \sum_{\aParty \in \partySet} W_{\iterCnt;i,j} \right)   y^{(\iterCnt-1)}_\partyB \notag \\\
	\text{(by Eq. \eqref{eq:colstoch})}		&=& \frac{1}{n} \sum_{\aParty \in \partySet} \valPart{\partyA}{\iterCnt} + \frac{1}{n}  \sum_{\partyB \in \partySet}   y^{(\iterCnt-1)}_\partyB  \label{eq:noisyavg.rec} \\
	&\vdots& \notag \\
	\text{(by Eq. \eqref{eq:noisyavg.rec})}		&=& \frac{1}{n} \sum_{\aParty \in \partySet} \sum_{\anIter=0}^\iterCnt \valPart{\partyA}{\anIter}  \notag \\
	\text{(by Eq. \eqref{eq:canceling})} &=&  \frac{1}{n} \sum_{\aParty \in \partySet} \valPriv{\aParty} + \indnoise{\aParty}.
	\label{eq:noisyavg} 
\end{eqnarray}
Equation \eqref{eq:noisyavg} shows that the only noise that remains in the final estimate is $\frac{1}{n} \sum_{\aParty \in \partySet} \indnoise{\aParty}$, whose variance is $\sdInd^2/n$. This shows that, if matrices $(\transMat{\anIter})_{\anIter \in [\iterCnt]}$ are column stochastic, then as long as $\partDistrSymb$ satisfies Equation \eqref{eq:canceling} the injected noise of this distribution will not affect accuracy.

\paragraph{Distribution $\partDistrSymb$} There are several ways to model the distribution $\partDistr{\cdot}$. We first show two examples to illustrate the core principle of the distribution. 

\begin{example}[Early Injection: $\partDistrEarlySymb$]
	For some variance $\sdCancel^2 > 0$ and  $v \in \xVectSpace$, $(v_\anIter)_{\anIter =1}^\iterCnt \sim \partDistrEarly{v}$ is defined as follows:
	\begin{enumerate} 
		\item draw $\iterCnt$ i.i.d. samples $\partDistTerm{1}, \dots , \partDistTerm{\iterCnt}$ with distribution $\N(0, \sdCancel^2)$ 
		\item set  $v_0 = v + \sum_{\anIter=1}^\iterCnt \partDistTerm{\anIter}$ and $v_{\anIter} = - \partDistTerm{\anIter}$ for each $\anIter \in [\iterCnt]$
	\end{enumerate}   
	\label{ex:early.inj} 
\end{example}  

Let us instantiate $\partDistrSymb$ of Algorithm \ref{alg:ourprotocol} by $\partDistrEarlySymb$. First, each party $\aParty \in \partySet$ will draw $\iterCnt$ i.i.d noise terms $\cancelNoise{\aParty}{0}, \dots, \cancelNoise{\aParty}{\iterCnt}$ from $\N(0,\sdCancel^2)$. Then we have that 
\begin{align*}  
	\yVal{\aParty}{0} &= \valPriv{\aParty}+\indnoise{\aParty} + \sum_{\anIter=1}^\iterCnt \cancelNoise{\aParty}{\iterCnt}
	\text{\quad and that } \\
	\yVal{\aParty}{\anIter} &= \left( \sum_{\partyB \in \partySet} \transMatTab y^{(\anIter-1)}_\partyB \right) - \cancelNoise{\aParty}{\anIter}
\end{align*}  for all $\anIter \in [\iterCnt]$. In this example, parties  inject all the private value and noise at the beginning of the execution to $\yVal{\aParty}{0}$, and gradually remove each noise term in subsequent messages while mixing them.

Note that, in addition to the independent noise $\indnoise{\aParty}$, each party $\aParty$ injects $\iterCnt$ noise terms. One might naively think that privacy could be obtained by just adding a single extra noise term and progressively removing it instead of adding  many terms as done in the example. However, such approach is likely to compromise privacy. With the knowledge of interactions and weights matrices $(\transMat{1}, \dots, \transMat{\iterCnt})$, the adversary can easily construct a system of linear equations where private values and noise terms are unknowns and each observed message in $\advValView$ can be used to construct one equation of the system (see Section \ref{sec:privacy.abstract}). This attack strategy has shown to be dangerous in \cite{mrini2024privacy}. Therefore, it is crucial that each party has a sufficient number of noise terms to hide private values from adversarial observations.  

\begin{example}[Incremental Injection: $\partDistrIncremSymb$]
	For some variance $\sdCancel^2 > 0$ and  $v \in \xVectSpace$, $(v_\anIter)_{\anIter =1}^\iterCnt \sim \partDistrIncrem{v}$ is defined as follows:
	\begin{enumerate} 
		\item draw $\partDistTerm{1}, \dots , \partDistTerm{\iterCnt}$  as in Example \ref{ex:early.inj} 
		\item set $v_{0} = v/(\iterCnt+1) + \partDistTerm{1}$, $v_{\anIter} = v/(\iterCnt+1) -  \partDistTerm{\anIter} + \partDistTerm{\anIter+1}$ for all $\anIter \in [\iterCnt-1]$ and  $v_{\iterCnt} = v/(\iterCnt+1)-  \partDistTerm{\iterCnt}$
	\end{enumerate} 
	\label{ex:increm.inj} 
\end{example}
%Let $\cancelNoise{\aParty}{1}, \dots , \cancelNoise{\aParty}{\iterCnt}$ be drawn as in Example \ref{ex:early.inj} for all $\aParty \in \partySet$. Then $(\valPart{\aParty}{\anIter})_{\anIter =1}^\iterCnt \sim 
%\partDistr{v}$ implies that $\valPart{\aParty}{0} = v/(\iterCnt+1) + \cancelNoise{\aParty}{1}$,  $\valPart{\aParty}{\anIter} = v/(\iterCnt+1) -  \cancelNoise{\aParty}{\anIter} + \cancelNoise{\aParty}{\anIter+1}$ for all $\anIter \in [\iterCnt-1]$ and $\valPart{\aParty}{\iterCnt} = v/(\iterCnt+1)-  \cancelNoise{\aParty}{\iterCnt}$.  
%\label{ex:increm.inj}
%\end{example}
In Example \ref{ex:increm.inj}, $v$ is spread over all vectors $(v_\anIter)_{\anIter =1}^\iterCnt$ and one noise term is added and canceled at each iteration. When $\partDistrIncremSymb$ is  used in Algorithm \ref{alg:ourprotocol}, each party $\aParty \in \partySet$ draws  i.i.d. noise terms $\cancelNoise{\aParty}{0}, \dots, \cancelNoise{\aParty}{\iterCnt}$ as in the previous example. Then we have that 
\begin{align*}
	\yVal{\aParty}{0} &= \valPriv{\aParty}/(\iterCnt+1) + \cancelNoise{\aParty}{1},\\
	\yVal{\aParty}{\anIter} &= \left( \sum_{\partyB \in \partySet}\transMatTab y^{(\anIter-1)}_\partyB \right) - \cancelNoise{\aParty}{\anIter} + \cancelNoise{\aParty}{\anIter+1} \quad \forall \anIter \in [\iterCnt-1] \text{ and } \\
	\yVal{\aParty}{\iterCnt} &= \left( \sum_{\partyB \in \partySet} \transMatTab y^{(\anIter-1)}_\partyB \right) - \cancelNoise{\aParty}{\anIter}. 
\end{align*}
Party $\aParty$ injects $1/(\iterCnt+1)$-th of  $\valPriv{\aParty}$ per iteration, renewing the noise each time. At iteration $\iterCnt$ no new noise term is injected.  The sensitivity of each  noisy term $\valPart{\aParty}{\anIter}$ to changes in $\valPriv{\aParty}$ is in the worst case smaller than in Example \ref{ex:early.inj} where the private values are injected all at once in the first iteration. This reduces the variance $\sdCancel^2$ required to satisfy differential privacy.  As shown in Equation \eqref{eq:noisyavg}, noise added by $\partDistrSymb$ cancels and its variance $\sdCancel^2$ does not impact accuracy in the absence of failures. However, parties can disconnect during the computation. Therefore, it is important to remain $\sdCancel^2$ as low as possible.

\paragraph{Multivariate Gaussians} Examples \ref{ex:early.inj} and \ref{ex:increm.inj} are part of a more general types of distributions for which we prove differential privacy guarantees. We define it below.  

\begin{definition}[Multivariate Gaussian]
We call a distribution $\partDistrSymb$ is a $(\xCoeffVec, \matrixPart)$-\emph{Gaussian} if for some $\xCoeffVec \in \R^{\iterCnt+1}$ and $\matrixPart \in \R^{(\iterCnt+1)\times \iterCnt}$ there holds $\partDistr{u} =\partDistrGauss{u}{\xCoeffVec}{\matrixPart}$, where sampling $(v_{0},\dots,v_{\iterCnt})\sim \partDistrGauss{u}{\xCoeffVec}{\matrixPart}$ is achieved by
%for all $v\in \xVectSpace$ if where  $\partDistrGauss{v}{\xCoeffVec}{\matrixPart}$ has the following distribution:  a sample $(v_{0},\dots,v_{\iterCnt}) \sim \partDistrGauss{v}{\xCoeffVec}{\matrixPart}$ is generated by
	\begin{enumerate}
		\item Sampling $\partDistTerm{k} \sim \mathcal{N}(0, \sigma_k^2)$ for each $k \in [1,\iterCnt]$
		\item Setting $	v_{\anIter} =  \xCoeff{\anIter}u + \sum_{k=1}^{\iterCnt} \matrixPartEl{\anIter}{k} \partDistTerm{k}$ for each $\anIter \in [0, \iterCnt]$.
	\end{enumerate}
\end{definition}
One can see that $\partDistrGauss{v}{\xCoeffVec}{\matrixPart}$ is the multivariate Gaussian distribution $\mathcal{N}(\xCoeffVec v, \matrixPart 
\Sigma_{g}  \matrixPart^\top)$ where $\Sigma_{g} = \mathrm{diag}(\sigma_1^2,\dots,\sigma_\iterCnt^2)$. 
%It must hold that $\sum_{\anIter=0}^\iterCnt \xCoeff{\anIter} = 1$ and that $\sum_{\anIter=0}^\iterCnt \matrixPartEl{\anIter}{:} = 0$ for $\partDistrGauss{\cdot}{\xCoeffVec}{\matrixPart}$ to satisfy the canceling property of Equation \eqref{eq:canceling}. 
To prove our privacy guarantees, we will need  \emph{Valid} $(\xCoeffVec, \matrixPart)$-Gaussians:
%we will require that such multivariate gaussians are \emph{Valid}. We define such property in the following.  
%We say that a matrix $M$ is surjective if the system $Ax=b$ has at least one valid solution of $x$ for all $b$.
\begin{definition}[Valid $(\xCoeffVec, \matrixPart)$-Gaussian] 
%	\janfoot{I think here and below we should say ``Valid $(c,Z)$-Gaussians'', as defined in the above definition. (Still, this doesn't imply that the name $(c,Z)$-Gaussian is really the ideal name.)}
  A 	$(\xCoeffVec, \matrixPart)$-\emph{Gaussian} $\partDistrSymb$  is \emph{Valid} if (i) it 
%  \janfoot{what exactly should satisfy the property? The distribution with $c$ and $Z$ fixed and $v$ playing the same role as in Eq (2) ?} 
  satisfies the property of Equation \eqref{eq:canceling}, i.e., $\sum_{\anIter=0}^\iterCnt v_\anIter = u$ independently of the $\eta_k$,
  (ii) the matrix $(\xCoeffVec, \matrixPart) \in \R^{(\iterCnt+1)\times (\iterCnt+1)}$ is invertible and (iii) $\matrixPartEl{-\iterCnt}{:} \in \R^{\iterCnt\times\iterCnt}$ is invertible.   
\end{definition}
Condition (i) above implies that $\sum_{\anIter=0}^\iterCnt \xCoeff{\anIter}=1$ and $\forall k: \sum_{\anIter=0}^\iterCnt \matrixPartEl{\anIter}{k} = 0$. Conditions (ii) and (iii) prevent individual noise terms from being guessed if the adversary does not observe all the messages of a party. 
%\janfoot{Please check!}
%\janrm{\emph{Valid} Gaussians ensure that sufficient noise terms are generated such that linear equations observed by the adversary do not make any noise term overdetermined} \janfoot{I'm not sure 'overdetermined' is the technically correct word here.  The point is probably not that individual noise terms can not be guessed, we only want to avoid that the secrets can't be guessed.}\janrm{, preventing the the inference of sensitive information. We will expand on this aspect in Section~}\ref{sec:privacy}.
\newcommand{\lemmaValidGaussStm}{
	$\partDistrEarlySymb$ and $\partDistrIncremSymb$ are Valid Gaussians. }
\begin{lemma}
	\label{lm:validgauss}
	\lemmaValidGaussStm
\end{lemma} 
We prove  Lemma \ref{lm:validgauss} it in Appendix \ref{app:protocol}. 
In practice we'll adopt  $\partDistrIncremSymb$
%the one outlined in Example \ref{ex:increm.inj}
due to its good properties against dropouts. 

\subsection{Protocol for Dropouts} 
\label{sec:protocol.dropouts} 

We now present our protocol in the presence of temporary and permanent dropouts.  For all  $\anIter \in [1,\iterCnt]$, we denote by $\onlNodes{\anIter}$ the set of parties that did not dropout at iteration $\anIter$. Our protocol is described in Algorithm \ref{alg:ourprotocol.drop}. We consider that all parties in $\partySet$ have at least started the protocol and computed its first message (i.e., $\onlNodes{0} = \partySet$), otherwise they are not considered as part of the protocol. 

The \initPhase{} Phase only requires minor changes with respect to Algorithm \ref{alg:ourprotocol} on the sampling of correlated noise. Since a party $\aParty \in \partySet$ might dropout, not all samples $(\valPart{\anIter}{\aParty})_{\anIter \in [0,\iterCnt]}$ are generated in advance, as it is more convenient to do it per iteration, depending on the online history $\aParty$.  Only the first term $\valPart{\aParty}{0}$ is sampled at \initPhase{} Phase. 
%Distribution $\partDistrOnlSymb$ differs  from  $\partDistrSymb$ in later samples, but the first term $\valPart{\aParty}{0}$ has the same distribution as when using   $\partDistrSymb$. 

% from distribution and $\valPart{\aParty}{\anIter}$ is sampled at each iteration $\anIter \in [1,\iterCnt]$. \begin{environment-name}
	% 	content
	% \end{environment-name}

At each iteration $\anIter \in [1, \iterCnt]$ of the \cancelPhase{} Phase,  only active parties in $\onlNodes{\anIter}$ exchange messages. Each party $\partyA \in \onlNodes{\anIter}$ receives its incoming messages of online parties, sends its outgoing messages and detects that its outgoing messages are not received by $\partyB \in \outNeigh{\aParty}{\anIter} \setminus \onlNodes{\anIter}$ due to dropout. Matrix  $\transMatOnl{\anIter}$ has the weights of real exchanges, correcting  the attempted exchanges $\transMat{\anIter}$ according to $\onlNodes{\anIter}$. The weights of incoming neighbors that did not drop out remain unchanged (line \ref{alg.line:ourprotocol.drop.WOnl.1}), otherwise they are set to $0$ (line \ref{alg.line:ourprotocol.drop.WOnl.2}). To keep the column stochasticity of $\transMatOnl{\anIter}$, the weight $\transMatOnlEl{\anIter}{\partyA}{\partyA}$ of its own message $\yVal{\partyA}{\anIter-1}$ is increased with the weights of value that $\partyA$ could not send due to dropout (line \ref{alg.line:ourprotocol.drop.WOnl.3}). 

Next, each party $\partyA$ samples a new term $\valPart{\partyA}{\anIter}$ from the distribution $\partDistrOnl{\partyA}{\anIter}$ (line \ref{alg.line:ourprotocol.drop.valPart}).  Each party adapts $\partDistrSymb$ to its online history. At the end of the \cancelPhase{} phase, noise terms $(\valPart{\aParty}{\anIter})_{\anIter \in [0,\iterCnt]}$ follow distribution $\partDistrOnlParty{\aParty}$, which we will explain later.

After that, party $\partyA$ aggregates incoming messages and noise terms to compute $\yVal{\partyA}{\anIter}$ as in the original protocol (line \ref{alg.line:ourprotocol.drop.canceling}). If party $\partyA$ dropped out, it will not be able to perform any exchange or sample any noise term. Therefore it will consider the last message he computed as his current message  (line \ref{alg.line:ourprotocol.drop.canceling.drop}). 
%Parties that did not exchange any message with neighbors because they all dropped out will also be inactive and considered outside $\onlNodes{\anIter}$. 
This concludes the \cancelPhase{} Phase.

In the \dissemPhase{} phase, parties average the values $(\yVal{\partyA}{\iterCnt})_{\aParty \in \onlNodes{\iterCnt}}$.  When a party $\aParty$ drops out during an iteration, it only partially injects its value $\valPriv{\aParty}$. Moreover, if $\aParty$ permanently drops out at some iteration $\anIter \in [1, T]$, a part of the private values of (a subset of) all parties is lost due to the mixing nature of our protocol, along with the message $\yVal{\partyA}{\anIter}$, which is never sent. Therefore, when computing the average of $(\yVal{\partyA}{\iterCnt})_{\aParty \in \onlNodes{\iterCnt}}$, dividing $\sum_{\aParty \in \onlNodes{\iterCnt}} \yVal{\partyA}{\iterCnt}$ by $n$ is less accurate than dividing by the total weight of private values injected. 

We adopt the latter approach, which is feasible by adding some extra information to the messages.
Even if $(\valPriv{\aParty})_{\aParty \in \partySet}$ are secret, parties can still track the total weight of private values in each message if they share their own weight alongside their messages and update it when they inject part of their private values or aggregate incoming messages. This does not constitute a privacy breach, as such weights are already known to the adversary (as encoded in Equation~\eqref{eq:linear.complete} and explained in Section~\ref{sec:privacy.knowledge}) and only depend on the status of parties being online of offline rather than the private values themselves. 

We assume that $(\yVal{\partyA}{\iterCnt})_{\partyA \in \onlNodes{\iterCnt}}$ is visible to the adversary. For resilience to dropouts, instead of using the method in \cite{boyd2006randomized}, one can employ other non-private epidemic dissemination protocols \cite{pittel1987spreading,demers1987epidemic} that have negligible probability of information loss due to dropout. In the unlikely case that a message from party $\aParty$ is lost in the \dissemPhase{} phase due to a dropout, we consider $\aParty \not\in \onlNodes{\iterCnt}$.

\begin{algorithm}[H]
	\begin{algorithmic}[1]
		\STATE{\textbf{Input:} $\iterCnt \in \mathbb{N}$, $\vectPriv = (x_1, \dots, x_n)^\top \in \xSpace^n$, $W_1, \dots, W_{\iterCnt} \in \R^{n\times n}$
			, $\sdInd^2 \in \R$, $\partDistrSymb: \xVectSpace \to\probDistrSpace{\xVectSpace^{T+1}}$}
		\STATE{\textit{\initPhase{} Phase}}
		\FORALL{$\aParty \in \partySet$} 
		\STATE{Sample $\indnoise{\aParty} \sim \N(0, \sdInd^2)$ and $\valPart{\aParty}{0}  \sim \partDistrOnl{\aParty}{0}$ \label{alg.line:ourprotocol.drop.draw}}
		\STATE{$y^{(0)}_\aParty \gets  \valPart{\aParty}{0}$} \label{alg.line:ourprotocol.drop.init}
		\ENDFOR
		\STATE{\textit{\cancelPhase{} Phase}}
		\FOR{$\anIter \in \{1 \dots \iterCnt\}$}
		\FORALL{$\partyA \in \partySet$}
		\IF{$\partyA \in \onlNodes{\anIter}$}
		\STATE{Set $\transMatOnlEl{\anIter}{\partyA}{\partyB} \gets \transMatEl{\anIter}{\partyA}{\partyB}$ \quad for all $\partyB \in \onlNodes{\anIter}$}
		\label{alg.line:ourprotocol.drop.WOnl.1}
		\STATE{Set $\transMatOnlEl{\anIter}{\partyA}{\partyB} \gets 0$ \qquad\quad   for all $\partyB \in \partySet \setminus \onlNodes{\anIter}\setminus \{\partyA\} $}
		\label{alg.line:ourprotocol.drop.WOnl.2}
		\STATE{Set $\transMatOnlEl{\anIter}{\partyA}{\partyA} \gets  \sum_{\partyB \in \partySet \setminus \onlNodes{\anIter}} \transMatEl{\anIter}{\partyB}{\partyA}$}  
		\label{alg.line:ourprotocol.drop.WOnl.3}
		\STATE{Sample $\valPart{\aParty}{\anIter}  \sim \partDistrOnl{\aParty}{\anIter}$ }
		\label{alg.line:ourprotocol.drop.valPart}
		\STATE{$y^{(\anIter)}_\partyA \gets \left( \sum_{\partyB \in \partySet} \transMatOnlEl{\anIter}{\partyA}{\partyB} y^{(\anIter-1)}_\partyB \right)+  \valPart{\partyA}{\anIter}$}
		\label{alg.line:ourprotocol.drop.canceling} 
		\ELSE
		\STATE{$y^{(\anIter)}_\partyA \gets y^{(\anIter-1)}_\partyA$}
		\label{alg.line:ourprotocol.drop.canceling.drop} 
		\ENDIF
		\ENDFOR
		\ENDFOR
		\STATE{\textit{\dissemPhase{} Phase}: Parties jointly average $(\yVal{\partyA}{\iterCnt})_{\partyA \in \onlNodes{\iterCnt}}$.
			% 			using transition matrices $W_{\iterCnt+1}, \dots, W_{\iterCnt+\iterCnt'}$ for $\iterCnt'$ iterations
		} 
	\end{algorithmic}		
	\caption{\inca{} for dropouts}
	\label{alg:ourprotocol.drop}
\end{algorithm}

\paragraph{Distribution $\partDistrOnlSymb$} Now we describe how parties use multivariate Gaussians when dropouts occur. 
Let  $\partDistrSymb = \partDistrGauss{v}{\xCoeffVec}{\matrixPart}$ and $\onlSymb = (\onlNodes{\aParty})_{\aParty \in \partySet}$ be the dropout history. Then, for each party $\aParty \in \partySet$, $\partDistrOnlParty{\aParty} = \partDistrGauss{\valPriv{\aParty} + \indnoise{\aParty}}{\xCoeffOnlVec{\aParty}}{\matrixPartOnl{\aParty}}$ where 
\begin{eqnarray*}
	\xCoeffOnl{\aParty}{\anIter} = 0,  &\matrixPartOnlEl{\aParty}{\anIter}{:} = 0&  \quad \text{if $\aParty \not\in \onlNodes{\anIter}$ and $\anIter \in [0,\iterCnt]$}\\ 
	\xCoeffOnl{\aParty}{\anIter} = \xCoeff{\anIter'},  & \matrixPartOnlEl{\aParty}{\anIter'}{:} = \matrixPartEl{\anIter}{:} & \quad  \text{if $\aParty \in \onlNodes{\anIter}$ and $\anIter \in [0,\iterCnt-1]$} \\
	\xCoeffOnl{\aParty}{\iterCnt} = \xCoeff{\iterCnt'},  &\matrixPartOnlEl{\aParty}{\iterCnt}{:} = -\sum_{\anIter'=0}^{\iterCnt'-1} \matrixPartEl{\anIter'}{:}& \quad  \text{if $\aParty \in \onlNodes{\iterCnt}$} 
\end{eqnarray*}
and $\anIter' = \sum_{k=0}^\anIter \indicator{\partyA \in \onlNodes{k}}$ for all $\anIter \in [0, \iterCnt]$. Essentially, parties do not inject anything into the system if they  dropped out and, if they are online in the last iteration, they cancel all previously injected noise terms. 
%$\partDistrOnlSymb$ does not satisfy Equation \eqref{eq:canceling}. However it will satisfy the following weaker property. 

Replacing $\valPart{\aParty}{\anIter}$ using the definition of $\partDistrOnlSymb$, we have the following equations for the messages:
\begin{align}
	\yVal{\aParty}{0} &= \xCoeffOnl{\aParty}{0} (\valPriv{\aParty}+\indnoise{\aParty}) + \sum_{k=1}^\iterCnt  \matrixPartOnlEl{\aParty}{0}{k} \cancelNoise{\aParty}{k} \label{eq:iter.gauss.init}  \\ 
	\yVal{\aParty}{\anIter } &= \sum_{\partyB =1}^n \transMatOnlEl{\anIter}{\aParty}{\partyB} \yVal{\partyB}{\anIter} +  \xCoeffOnl{\aParty}{\anIter} (\valPriv{\aParty}+\indnoise{\aParty}) + \sum_{k=1}^\iterCnt  \matrixPartOnlEl{\aParty}{\anIter}{k} \cancelNoise{\aParty}{k} \label{eq:iter.gauss.cancel} \\   & \qquad \text{for all $\anIter \in [1, \iterCnt].$} \notag  
\end{align}  
These equations determine the view of the adversary.

For each $\aParty \in \onlNodes{\iterCnt}$, we have that 
\begin{equation}  
	\sum_{\anIter=0}^\iterCnt \valPart{\aParty}{\anIter} =   \totalWeightParty{\aParty} (\valPriv{\aParty} + \indnoise{\aParty}) 
	\label{eq:canceling.drop} 
\end{equation} 
where $\totalWeightVec \in \R^{n}$ is such that 
\begin{equation}
	\totalWeightParty{\aParty} = \sum_{\anIter=0}^\iterCnt \xCoeffOnl{\aParty}{\anIter}.
	\label{eq:totalWeight}
\end{equation} If party $\aParty$ did not permanently drop out, she/he will be able to inject a proportion $\totalWeightParty{\aParty}$ of its private value $\valPriv{\aParty}$ plus independent noise $\indnoise{\aParty}$. 

\paragraph{Utility} If all parties have the same probability of dropout, our estimate will be unbiased. If there are only temporary drop outs, i.e. $\onlNodes{\iterCnt} = \partySet$, all correlated noise is canceled. By following a similar analysis as the one in Equation \eqref{eq:noisyavg}, we have that  
\[ 
\sum_{\aParty=1}^n  \yVal{\aParty}{\iterCnt} = \sum_{\aParty=1}^n \totalWeightParty{\aParty}(\valPriv{\aParty}+ \indnoise{\aParty}). 
\] 
If permanent dropouts occur, two main disruptions happen. First, parts of the private values with non-zero coefficients in the dropped messages are lost. Second, correlated noise is not completely canceled. Therefore, the properties in equations \eqref{eq:colstoch} and \eqref{eq:canceling.drop} do not hold. We empirically show the final accuracy in Section \ref{sec:exp.dropouts}.

%\subsection{Correlated Gaussian Noise}
%\label{sec:protocol.gauss} 

%In practice we set in most cases $\sigma_\anIter^2 = \sdCancel^2$ for all $\anIter \in [1,\iterCnt]$ and analyze privacy in function of possible values of  $\sdCancel$. 

% When dropouts occur, per-party distributions satisfy $\partDistrOnlParty{\partyA} = \partDistrGauss{\valPriv{\partyA}+\indnoise{\partyA}}{\xCoeffOnlVec{\aParty}}{\matrixPartOnl{\partyA}}$, where $\xCoeffOnlVec{\partyA} \in \R^{\iterCnt+1}$ and $\matrixPartOnl{\partyA} \in \R^{(\iterCnt+1)\times \iterCnt}$ adapt to the online history of party $\partyA$ in a similar way than $\partDistrOnlEarlySymb$  and $\partDistrOnlIncremSymb$ to $\partDistrEarlySymb$ and $\partDistrIncremSymb$ respectively.
%Where for all $\anIter \in [0, \iterCnt]$, $\xCoeffOnl{\aParty}{\anIter}$ is the proportion of $\valPriv{\aParty}+\indnoise{\aParty}$ that party $\aParty$ injects to the system at iteration $\anIter$ and $\matrixPartOnlEl{\aParty}{\anIter}{k}$ is the coefficient of noise term $\cancelNoise{\aParty}{k}$ when injected at iteration $\anIter$.  

% !TeX root = main.tex

\section{Privacy Analysis} 
\label{sec:privacy} 

We now present our privacy results. In Section \ref{sec:privacy.knowledge} we show how the adversary's knowledge can be structured as a set of linear equations. Next, in Section \ref{sec:privacy.abstract} we present an abstract result that accurately accounts for the privacy loss of \inca{}. After that, we present in Section \ref{sec:privacy.optimal} interpretable results on the conditions to obtain differential privacy with accuracy comparable to that of Central DP when no dropouts occur. Finally, in Section \ref{sec:privacy.negative}, we present sufficient positive and negative conditions on the graph that models parties interactions to obtain privacy. 

\subsection{Knowledge of the Adversary}
\label{sec:privacy.knowledge}  

We prove \obsDPLabel{} and \colDPLabel{} guarantees. Hence, we  assume that the adversary knows the set  $\transitionset = \{\transMat{1}, \dots, \transMat{\iterCnt}\}$ and online activity $\onlSymb = (\onlNodes{\anIter})_{\anIter \in [1,\iterCnt]}$. For the subset of values of messages it knows, we will slightly abuse notation with respect to the definition of $\advValView$ in  \obsDPAcron{} and \colDPAcron{} so that   $\advValView \subseteq  \partySet \times  \initIterSet$  denotes the subset of observed iteration-party messages. 
Hence the messages $\{((\aParty, \anIter),\yVal{\aParty}{\anIter})\}_{(\aParty,\anIter)\in\advValView}$ will be known. 

For \obsDPAcron{}, we assume that the adversary can observe the final messages of the \cancelPhase{} Phase of all parties, i.e.,  
\[
\advValView \supseteq \{ (\aParty, \iterCnt) \}_{\aParty \in \partySet}.
\]
For 
\colDPAcron{}, they will additionally see all messages seen by corrupted parties. That is, 
\[
\advValView = \{(\aParty, \anIter) \in  \partySet \times [0,\iterCnt-1]: (\{\aParty\} \cup \outNeigh{\aParty}{\anIter+1}) \cap \corrnodes \not{=} \emptyset \} \cup \{ (\aParty, \iterCnt) \}_{\aParty \in \partySet}.
\]

\paragraph{Knowledge as Linear Equations}

%\subsection{Knowledge of the Adversary}
%\label{sec:privacy.knowledge}

The knowledge of the adversary can be structured as a set of linear equations where the unknowns are private values $\vectPriv$ and noise terms $(\indNoiseVec,  \cancelNoiseMat)$. We show below how  these equations are constructed. 

Let  $\nNoiseHonest \coloneq \nHonest \iterCnt$  be the number of canceling noise terms  unknown by the adversary before they make any observations. We denote by  $\vectPrivHonest \coloneq (\valPriv{\aParty})_{\aParty \in \partySetHonest} \in \xSpace^{\nHonest}$, 
$\indNoiseVecH \coloneq (\indnoise{\aParty})_{\aParty \in \partySetHonest} \in \xVectSpace^\nHonest$ and $\cancelNoiseVecHonest \coloneq   (\cancelNoiseVecParty{\aParty})_{\aParty \in \partySetHonest} \in \xVectSpace^{\nNoiseHonest}$ 
the vectors of private values and noise terms of honest parties. 
%$\yVecCorr \coloneq \left( \yVal{\aParty}{\anIter} \right)_{ (\aParty, \anIter) \in \advValView } \in \xVectSpace^{\nObs}$ 
Then, from an execution $\exec = (\transitionset, \onlSymb, \advValView)$ of Algorithm \ref{alg:ourprotocol.drop}, the adversary can construct  a set of linear equations described by 
\begin{equation}
	\xMatCorr (\vectPrivHonest + \indNoiseVecH) + \nMatCorr \cancelNoiseVecHonest = \yVecCorr
	\label{eq:linear.adv}
\end{equation}
where for $\nObs \coloneq |\advValView| \le \nHonest \iterCnt$, $\xMatCorr \in \R^{\nObs \times \nHonest}$ and $\nMatCorr \in \R^{m \times\nNoiseHonest}$  are matrices with public coefficients such that $(\xMatCorr, \nMatCorr)$ is a full rank, $\yVecCorr  \in \xVectSpace^{\nObs}$ 
is the set of messages observed by the adversary after removing known constants and redundant information (e.g., linearly dependent equations),   and $(\vectPrivHonest, \indNoiseVecH, \cancelNoiseVecHonest)$ are the unknowns. In Appendix \ref{app:privacy.knowledge} we show the details on how to construct the system of Equation \eqref{eq:linear.adv}. 
%However, the involved notation presented therein will mostly not be used in the rest of the paper. 
%Now that we have described the knowledge of the adversary we are ready to present our abstract privacy guarantees.

\paragraph{Discussion} Let's inspect the system of Equation \eqref{eq:linear.adv}. For each $\aParty \in \partySetHonest$, variables $\valPriv{\aParty}$ and  $\indnoise{\aParty}$ only appear coupled in the term  $\valPriv{\aParty} + \indnoise{\aParty}$. Therefore, the adversary will never be able to isolate them and we can consider them as a single variable $\valPrivNoisy{\aParty} = \valPriv{\aParty} + \indnoise{\aParty}$. In the worst case, the adversary observes all the messages of honest parties, meaning that $\nObs = \nHonest + \nNoiseHonest$ and that  $(\xMatCorr, \nMatCorr)$ is a square invertible matrix. Then, the adversary is able to recover  $\valPrivNoisy{\aParty}$ for all  $\aParty \in \partySetHonest$ by computing $(\xMatCorr, \nMatCorr)^{-1} \yVecCorr$. In that case, the only way to protect private values is to set $\indnoise{\aParty}$ large enough, i.e. using the same noise required for Local DP. Hence,  we emphasize the importance of setting $\nHonest + \nNoiseHonest > \nObs$ to prevent this situation, which justifies the large amount of noise terms (one per iteration) each party generates. 
%In sections .., we show more precise conditions for worst cases and safety.  

\subsection{Abstract Result} 
\label{sec:privacy.abstract}

We first present an abstract privacy result that provides a relationship between the variance of added noise terms and $(\epsilon, \delta)$-DP guarantees for a given execution of Algorithm \ref{alg:ourprotocol.drop}. We define $\varNoise \in \xVectSpace^{(\nHonest+ \nNoiseHonest) \times (\nHonest + \nNoiseHonest)}$ to be the covariance matrix of  $\left(\begin{smallmatrix} \indNoiseVecH \\ \cancelNoiseVecHonest \end{smallmatrix}\right)$. Note that $\varNoise$ is diagonal as all noise terms are independent. %Results stated in Theorem \ref{thm:dp.abstract}. 

\newcommand{\abstractThmStatement}{Let $\exec = (\transitionset, \onlSymb, \advValView)$ be an execution of Algorithm \ref{alg:ourprotocol.drop}  where $\partDistrSymb$ is a Valid  $(\xCoeffVec, \matrixPart)$-Gaussian and  $\advValView$ be defined by the observations of the adversary (respectively by the corruption of a set $\corrnodes$ of parties).  Let $\xMatCorr$, $\nMatCorr$ be derived from $\exec$ as described in Equation \eqref{eq:linear.adv}. Let  $\varNoise$ be the covariance matrix of noise terms $\left(\begin{smallmatrix} \indNoiseVecH \\ \cancelNoiseVecHonest \end{smallmatrix}\right)$, which have positive variance. 
	Then, execution $\exec$ is  \obsDP{$\epsilon$}{$\delta$}{$\advValView$}  (respectively \colDP{$\epsilon$}{$\delta$}{$\corrnodes$})  if all noise terms have positive variance and 
	\begin{equation}
		\xMatCorrCol^\top \varMatProd^{-1} \xMatCorrCol \le  \frac{\epsilon^2}{c^2} \text{\qquad for each column $\xMatCorrCol$ of $\xMatCorr$} 
		\label{eq:dp.abstract}  
	\end{equation}
	where  $c^2 > 2 \ln(1.25 / \delta)$ and  
	\begin{equation}
		\varMatProd \coloneq (\xMatCorr, \nMatCorr) \varNoise(\xMatCorr, \nMatCorr)^\top \in \R^{\nObs \times \nObs} 
		\label{eq:varMatProd}.
	\end{equation}
}
\begin{theorem}  
	\abstractThmStatement
	\label{thm:dp.abstract}  
\end{theorem}  
% Theorem \ref{thm:dp.abstract} establishes a relation between the variance of noise terms and the obtained $(\epsilon, \delta)$-DP guarantees given an execution of Algorithm \ref{alg:ourprotocol.drop}. 

We prove this result in Appendix \ref{app:privacy.abstract}. Note that the invertibility of $\varMatProd$  is not an additional constraint: given that $(\xMatCorr, \nMatCorr)$ is full rank, $\varMatProd$  is always invertible if all noise terms have positive variance. 
%Notably, the accounting of DP budget of our proof does not rely on DP black box composition rules that consider  worst-case leakages. Instead, we analyze the outcome of as a joint set of correlated variables, where the structure of the correlation is determined by the gossip exchanges. This allows to account for the privacy loss more accurately  than by the use of composition rules which, consider different pieces of information as black-box sources of privacy loss and therefore attempt to protect against non-existent  worst-case scenarios. Additionally, it protects of this extra knowledge of the adversary that can be exploited in attacks such as the proposed in \cite{mrini2024privacy}, when techniques do not consider this knowledge.  
Using Theorem \ref{thm:dp.abstract}, one can set variances and obtain $\epsilon$ and $\delta$. 
We now show the inverse direction, i.e., how to obtain $\varNoise$ given $\epsilon$ and $\delta$  using  convex optimization. 

A symmetric matrix $M$ is positive semi-definite, denoted $M \succeq 0$, if  $x^\top M x \ge  0$ for every non-zero vector $x$. $M$ is positive definite, or $M \succ 0$, if $x^\top M x >  0$ for every non-zero vector $x$.  

\newcommand{\crlConvexEq}{
\begin{equation}
%	\label{eq:psd} 
	\begin{ourmatrix}
		\varMatProd & \xMatCorrCol \\
		\xMatCorrCol^\top & \epsilon^2/  c^2 
	\end{ourmatrix} \succeq 0 
\end{equation}
}

\newcommand{\crlConvexStm}{
Let $\epsilon,\delta\in(0,1)$. 
Let $\xMatCorr$, $\nMatCorr$ be associated with an execution of Algorithm \ref{alg:ourprotocol.drop} be derived from $(\transitionset, \onlSymb, \advValView)$ (or from $(\transitionset, \onlSymb, \corrnodes)$) as defined in Theorem \ref{thm:dp.abstract}. Let $\varMatProd$ also be defined as in Theorem \ref{thm:dp.abstract}. Then, the execution is  \obsDP{$\epsilon$}{$\delta$}{$\advValView$}  (respectively \colDP{$\epsilon$}{$\delta$}{$\corrnodes$})  if all noise terms have positive variance and   
  \crlConvexEq
  for each column $\xMatCorrCol$ of $\xMatCorr$, where $c^2 > 2 \ln(1.25 / \delta)$. 
}
\begin{corollary}
	\label{crl:sdp}
	\crlConvexStm 
\end{corollary}  

We prove it in Appendix \ref{app:privacy.abstract}. It implies that for fixed $(\epsilon, \delta)$  the possible values of $\varNoise$ can be computed with a convex program.

%Results of Theorem \ref{thm:dp.abstract} work for any strategy as long as $\nMatCorr \varNoise (\nMatCorr)^\top$ is invertible and $\sdInd$ is large enough. We show below  that this is always the case in our strategy.
%\begin{corollary}
%Let $\xMatCorr$, $\nMatCorr$ be as defined in Equation \eqref{eq:linear.adv} and $\varNoise \in \R^{\nNoiseHonest \times \nNoiseHonest}$ be the covariance matrix of  $\cancelNoiseVecHonest$. Then Algorithm \ref{alg:ourprotocol} is $(\epsilon, \delta)$-differentially private for $\xCoeff{\aParty} = \dots$ and $\vecPart{\anIter} = \dots$ \cesar{add strategy} if it satisfies the conditions of 
%Equation \eqref{eq:dp.abstract} with  $c^2 \ge 2 \ln(1.25 / \delta)$.
%\label{crl:dp.strat}
%\end{corollary}
%\begin{proof}
%	Prove that $\nMatCorr$ is full rank for  $\xCoeff{\aParty} = \dots$ and $\vecPart{\anIter} = \dots$ defined in our strategy. Then prove the non-singularity of $\nMatCorr \varNoise (\nMatCorr)^\top$. 
%\end{proof}  

\subsection{Bounded Independent Noise $\indNoiseVec$} 
\label{sec:privacy.optimal}

%In Section \ref{sec:privacy.abstract}, we showed how to obtain $(\epsilon, \delta)$-DP parameters from the variance of noise terms and vice-versa.
In this section, we show how to achieve such guarantees with bounded variance $\sdInd^2$ of the independent noise $\indNoiseVec$. This is important, as $\sdInd^2$ completely determines the accuracy of Algorithm \ref{alg:ourprotocol} (see Equation \ref{eq:canceling} and significantly influences that of Algorithm \ref{alg:ourprotocol.drop}, as $\indNoiseVec$  are the only terms that \inca{} does not attempt to cancel. 

From an execution $\exec  = (\transitionset, \onlSymb, \advValView)$ of Algorithm \ref{alg:ourprotocol.drop}, the view of the adversary is completely determined by the values of $\vectPrivHonest$, $\indNoiseVecH$ and $\cancelNoiseVecHonest$ as described by Equation \eqref{eq:linear.adv}. We denote this view by
\[ 
\advViewExec{\exec}{\vectPrivHonest + \indNoiseVecH,  \cancelNoiseMatHonest} = \yVecCorr.
\]
%We use $\vectPrivNoisy^\honestSymb = \vectPrivHonest + \indNoiseVecH$. 
We also sometimes slightly abuse notation by using the matrix $\cancelNoiseMatHonest = (\cancelNoise{\aParty}{:})_{\aParty \in \partySetHonest} \in \xVectSpace^{\nHonest\times \iterCnt}$ instead of vector $\cancelNoiseVecHonest$  (which contains the same elements in a different shape)  as the last parameter of  $\advViewExecFunc{\exec}$.

%Let $\exec$ be an execution of Algorithm \ref{alg:ourprotocol.drop} that consists transition matrices $\transMatOnl{1}, \dots, \transMatOnl{\iterCnt}$, observed messages $\advValView$, corrupted parties $\corrnodes$ and coefficients $(\xCoeffOnlVec{\aParty}, \matrixPartOnl{\aParty})_{\aParty=1}^n$ of robust. The view of the adversary is determined by the values of $\vectPrivHonest$, $\indNoiseVecH$ and  $\cancelNoiseMatHonest$,  where $\cancelNoiseMatHonest$ is equal to the matrix obtained by keeping the rows of $\cancelNoiseMat$ that correspond to noise terms of  the honest parties. Recalling $\yVecCorr$ from Equation \eqref{eq:linear.adv}, we denote the view of the adversary  by 
%$\advViewExec{\exec}{\vectPrivHonest + \indNoiseVecH,  \cancelNoiseMatHonest} = \yVecCorr$. We use $\vectPrivNoisy^\honestSymb = \vectPrivHonest + \indNoiseVecH$ as a single parameter as the adversary never sees the summands of this term separately. 
%
%We also sometimes slightly abuse notation use $\cancelNoiseVecHonest$ instead of $\cancelNoiseMatHonest$ as the last parameter of  $\advViewExecFunc{\exec}$.

\newcommand{\stmIndLemmaOld}{
	Let $0\le \anIter \le \iterCnt-1$ and $\aParty\in\partySetHonest$. 
	If $(\aParty, \anIter) \not\in \advValView$ and if $\matLemmaA{\aParty}{\anIter}$ and $(\xCoeffVec, \matrixPart)$ are invertible, then for any $\vectPriv\in\xSpace^n$, $\indNoiseVec \in \xVectSpace^n$, $\noiseMatrix \in \xVectSpace^{n\times\iterCnt}$ and $\beta\in\mathbb{R}$ there holds
	\[\advViewP{\vectPriv, \indNoiseVec, \noiseMatrix } = \advViewP{\vectPriv+\edgeVec{\aParty}{\anIter}\beta,\indNoiseVec, \noiseMatrix-\etaTransEl{\aParty}{\anIter}\beta}\]  
}

%In this section, we give  explicit bounds in the variance of noise terms to satisfy privacy. 
We start by presenting Lemma \ref{lm:dp.noiseDiff.abstract}, which relates the change required in noise terms $\indNoiseVecH$ and $\cancelNoiseVecHonest$ such that executing the protocol with neighboring datasets  $\neighDatasetA$ or $\neighDatasetB$ remain indistinguishable to the adversary and the required variance of   $\indNoiseVecH$ and  $\cancelNoiseVecHonest$ to satisfy $(\epsilon,\delta)$-DP guarantees.

%  on the relation between the indistinguishability of the use of different inputs and the variance of the noise to obtain DP guarantees. 
%We first define $\varDiffMaxFunc: \R_+ \times (0,1) \mapsto \R_+$ be a function that takes pairs $(\epsilon, \delta)$ and returns the maximum value $\varDiff$ such that 
%\begin{equation} 
%	\epsilon \geq \varDiff^{1/2}  + \varDiff/2
%	\label{eq:varDiff.1} 
%\end{equation} 
%
%\begin{equation} 
%	\frac{(\epsilon - \varDiff/2)^2}{\varDiff} \geq 2 \ln \left( \frac{2}{\delta \sqrt{2\pi}}\right) 
%	\label{eq:varDiff.2} 
%\end{equation} 

\newcommand{\execDefStm}{Let $\exec = (\transitionset, \onlSymb, \advValView)$ be an execution of Algorithm \ref{alg:ourprotocol.drop}  where $\partDistrSymb$ is a Valid Gaussian and  $\advValView$ be defined by the observations of the advesary (or by the corruption of a set $\corrnodes$ of parties)}

\newcommand{\lemmaNoiseDiffStm}{
	Let $\epsilon, \delta \in (0,1)$. 
	Let $\exec = (\transitionset,\onlSymb,\advValView)$ be defined as in Theorem \ref{thm:dp.abstract}. For a pair of neighboring datasets $\neighDatasetA, \neighDatasetB \in \xVectSpace^{\nHonest}$, let $\indNoiseDiff \in \xVectSpace^{\nHonest}$ and $\cancelNoiseDiffVec \in \xVectSpace^{\nNoiseHonest}$ be such that  
	\[
	\advViewExec{\exec}{\neighDatasetA + \indNoiseVecH, \cancelNoiseVecHonest} =  \advViewExec{\exec}{\neighDatasetB+\indNoiseVecH+\indNoiseDiff, \cancelNoiseVecHonest+\cancelNoiseDiffVec}
	\label{eq:eqViewDiff}
	\]
	for any $\indNoiseVecH$, $\cancelNoiseVecHonest$. Recall that $\varNoise$ is the covariance matrix of $\begin{ourSmallMatrix}
		\indNoiseVecH \\
		\cancelNoiseVecHonest
	\end{ourSmallMatrix}$. 
	Then, execution $\exec$ is  \obsDP{$\epsilon$}{$\delta$}{$\advValView$}  (or \colDP{$\epsilon$}{$\delta$}{$\corrnodes$}) if
        for all such neighboring datasets $\neighDatasetA$ and $\neighDatasetB$
	\[ 
	\left((\indNoiseDiff)^\top, \cancelNoiseDiffVec^\top \right) \varNoise^{-1} \begin{ourmatrix}
		\indNoiseDiff \\
		\cancelNoiseDiffVec \\ 
	\end{ourmatrix}
	\le \frac{\epsilon^2}{c^2}
	\]
	where 
%	\janrm{$\indNoiseDiff$, $\cancelNoiseDiffVec$ are such that $\|\indNoiseDiff\|_2$ and $\|\cancelNoiseDiffVec\|_2$ are the maximum values for all pairs of neighboring datasets} and 
$c^2 > 2 \ln(1.25/\delta)$. 
}

\begin{lemma}
	\label{lm:dp.noiseDiff.abstract} 
	\lemmaNoiseDiffStm
\end{lemma}

We can think of $\indNoiseDiff$ as a maximum difference between data points in the domain (and hence neighboring datasets), i.e., a measure closely related to sensitivity.  In particular, for the sensitivity $\Delta(\hbox{\textsc{IncA}})$ there holds $\Delta(\hbox{\textsc{IncA}})=\max_{\neighDatasetA,\neighDatasetB }\|\hbox{\textsc{IncA}}(\neighDatasetA,\neighDatasetB)\|= \max_{\indNoiseDiff}\|\hbox{\textsc{IncA}}(\indNoiseDiff)\|$.
%\janfoot{Please check an remove if incorrect or if we have the link with sensitivity somewhere else already.}

The main techniques to prove Lemma \ref{lm:dp.noiseDiff.abstract} have been studied in \cite[Theorem 1]{sabater2022accurate}.
We adapt the proof to our setting in Appendix \ref{app:privacy.noiseDiff}.  This lemma  establishes the relation between the required $(\epsilon, \delta)$-DP parameters and magnitude of the change $\indNoiseDiff$ in independent noise term  $\indNoiseVec$ and the change $\cancelNoiseDiffVec$ in canceling noise terms $\cancelNoiseVec$ such that the view the of adversary remains indistinguishable when the input is changed from $\neighDatasetA$ to $\neighDatasetB$.

\paragraph{Correlated Noise Variance $\sdCancel^2$} In the rest of the paper, all noise terms of vector $\cancelNoiseMatHonest$ have variance $\sdCancel^2$. Hence the condition  to satisfy $(\epsilon,\delta)$-DP in Lemma \ref{lm:dp.noiseDiff.abstract} is equivalent to 
\begin{equation} 
	\frac{\|\indNoiseDiff\|^2_2}{\sdInd^2} + \frac{\|\cancelNoiseDiffVec\|^2_2}{\sdCancel^2}  \le \frac{\epsilon^2}{c^2}
	\label{eq:dpThetaMax}  
\end{equation} 
where $c^2 > 2\ln(1.25/\delta)$. 
%\janrm{If  $(\epsilon, \delta)$ are fixed, then $\sdInd^2$ and $\sdCancel^2$ are proportional to $\|\indNoiseDiff\|^2_2$ and $\|\cancelNoiseDiffVec\|_2^2$ respectively. }\janfoot{Not necessarily, even with $\epsilon$, $\delta$, $\xi^*$ and $\Delta_{(:)}$ fixed, it is possible that  $\sdInd^2$ increases and $\sdCancel^2$ decreases or vice versa.}

%As previously said, \inca{} can trivially achieve DP guarantees by scaling independent noise variance $\sdInd^2$ as in Local DP. 
%However, the accuracy depends on $\sdInd^2$. Therefore we are interested in cases where $\sdInd^2$ is small enough to achieve privacy-accuracy trade-offs that are comparable to Central DP. 

From Equation \eqref{eq:dpThetaMax}, we can see that a bounded $\sdInd$  is achievable for sufficiently large $\sdCancel> \|\cancelNoiseDiffVec\|c/\epsilon$ as long as $\indNoiseDiff$ is bounded. 
We break down our problem, first analyzing  in  Lemma \ref{lm:edgeVec} how privacy is amplified by each message that the adversary does not see and in subsequent theorems the conditions on unseen messages to obtain a bounded $\indNoiseDiff$ (and consequently   $\sdInd$).

\newcommand{\stmIndLemma}{Let $\exec = (\transitionset,\onlSymb,\advValView)$ be defined as in Theorem \ref{thm:dp.abstract}. Then if 
  $\yVal{\aParty}{\anIter}$ is not observed by the adversary  (i.e.,  $(\aParty,\anIter) \in \partySetHonest \times [0, \iterCnt-1] \setminus \advValView$)
%\janfoot{We have a type problem here:  $(\aParty,\anIter)$ is a pair, while $[0,\iterCnt-1]$ is a set of numbers, and then $\advValView$ is defined as a subset of $(P\times P)\times {0,1}^*$ is a set of triples.}
  and $\aParty \in \onlNodes{\anIter}$, then there exists $\edgeVec{\aParty}{\anIter} \in \xVectSpace^n$ and $\etaTransEl{\aParty}{\anIter} \in \xVectSpace^{n \times \iterCnt}$  such that for any $\beta \in \R$ we have 
	\[ 
	\advViewExec{\exec}{\vectPrivNoisy^\honestSymb , \cancelNoiseMatHonest}  = 	\advViewExec{\exec}{\vectPrivNoisy^\honestSymb  + \beta \edgeVec{\aParty}{\anIter}, \cancelNoiseMatHonest  - \beta \etaTransEl{\aParty}{\anIter} }.  
	\]
	Moreover, $\edgeVec{\aParty}{\anIter}_\aParty = \frac{\transMatOnlEl{\anIter+1}{\aParty}{\aParty}-1}{\totalWeightParty{\aParty}}$ and $\edgeVec{\aParty}{\anIter}_\partyB = \frac{\transMatOnlEl{\anIter+1}{\partyB}{\aParty}}{\totalWeightParty{\partyB}}$ for all $\partyB \in \partySetHonest \setminus \{\aParty\}$, where $\totalWeightVec\in \R^{n}$ is defined as in   Equation \eqref{eq:totalWeight}.
}

\begin{lemma}
	\stmIndLemma
	\label{lm:edgeVec}
\end{lemma}  

Lemma \ref{lm:edgeVec} establishes that if $\yVal{\aParty}{\anIter}$ is not part of the adversary's observations and was effectively transmitted to other parties,  a 
%\janrm{rect of possible solutions to}
vector in the nullspace of Equation \eqref{eq:linear.adv}, i.e., a vector which added to a solution gives another solution, is $(\edgeVec{\aParty}{\anIter}, -\etaTransEl{\aParty}{\anIter})$.   We prove this in Appendix \ref{app:privacy.indist}. 

In the following theorems bounds on $\sdInd^2$ are the result of constructing a bounded $\indNoiseDiff$, from the possible solutions  $\vectPriv+\indNoiseVec$ of Equation \eqref{eq:linear.adv}. The more messages are unobserved the easier to find smaller values of $\indNoiseDiff$ and $\cancelNoiseDiffVec$ of Lemma \eqref{lm:dp.noiseDiff.abstract}.

First  we show this when \inca{} is executed without dropouts. 
\newcommand{\thmCdpNodropStm}{
	Let $\epsilon, \delta \in (0,1)$. 
 Let $\exec = (\transitionset,\onlSymb,\advValView)$ be defined as in Theorem \ref{thm:dp.abstract} and associated with an execution without dropouts (i.e., $\onlNodes{\anIter} = \partySet$ for all $\anIter \in [0,\iterCnt]$).
	 Let $\hidden =  \partySetHonest\times [0,\iterCnt] \setminus \advValView$  be the set of pairs $(\aParty,\anIter)$ such that $\yVal{\aParty}{\anIter}$ is not seen by the adversary.  For all $(\aParty,\anIter) \in  \hidden$, let $\edgeVec{\aParty}{\anIter}$ be as defined in Lemma \ref{lm:edgeVec}. 
	If   $\{ \edgeVec{\aParty}{\anIter} \}_{(\aParty,\anIter) \in \hidden}$ has at least $\nHonest -1$ independent vectors, there exist $\indNoiseDiff$ and $\cancelNoiseDiffVec$ as defined in Lemma \ref{lm:dp.noiseDiff.abstract} and  $\exec$ satisfies  \obsDP{$\epsilon$}{$\delta$}{$\advValView$}  (or \colDP{$\epsilon$}{$\delta$}{$\corrnodes$}) for any 
	\[ 
	\sdInd^2 >  \frac{c^2}{\nHonest\epsilon^2} 	
	\] 	
	and  
	\[ 
	\sdCancel^2 \ge \frac{\|\cancelNoiseDiffVec\|^2_2}{
		\frac{\epsilon^2}{c^2} - 		\frac{1}{ \nHonest \sdInd^2}},
	\] 
	where $c^2 > 2\ln(1.25/\delta)$. 
}
\begin{theorem}	
	\thmCdpNodropStm
	\label{thm:cdp.nodrop}
\end{theorem}  

We prove this theorem in Appendix \ref{app:privacy.cdp}. By Equation \eqref{eq:noisyavg}, we have that Algorithm \ref{alg:ourprotocol} produces an unbiased estimate of   $\frac{1}{n}\sum_{\aParty \in \partySet} \valPriv{\aParty}$ with variance $\sdInd^2/n$. Under the conditions of Theorem \ref{thm:cdp.nodrop} this variance of Algorithm \ref{alg:ourprotocol}'s output can be made arbitrarily close to 
\[ 
\frac{c^2}{\nHonest n \epsilon^2}.
\]
The above matches the error of Secure Aggregation plus DP noise \cite{kairouz2021distributed,agarwal2021skellam,chen2022poisson} and 
Central DP when $\nHonest = n$. In the case where  $\nHonest < n$, \inca{} order-wisely matches the $O(1/n^2)$ mean squared error of Central DP. 
Protocols that exactly obtain Central DP accuracy in the presence of corrupted parties require a large communication cost \cite{sabater2023private,jayaraman2018distributed,dwork2006our}. 
In the remaining of Section \ref{sec:exp.dropouts} we show a bound on $\sdInd^2$ in the presence of dropouts.

\newcommand{\thmTotalSumStm}{ 
	Let $\epsilon,\delta \in (0,1)$. Let $\exec = (\transitionset,\onlSymb,\advValView)$ be defined as in Theorem \ref{thm:dp.abstract}. Let  $\wOnl = \sum_{\aParty \in \partySetHonest} \totalWeightParty{\aParty}$, where $\totalWeightVec$ is defined in Equation \eqref{eq:totalWeight}. Let $\hidden =  \partySetHonest\times [0,\iterCnt] \setminus \advValView$  be the set of pairs $(\aParty,\anIter)$ such that $\yVal{\aParty}{\anIter}$ is not seen by the adversary.  For all $(\aParty,\anIter) \in  \hidden$, let $\edgeVec{\aParty}{\anIter}$ be as defined in Lemma \ref{lm:edgeVec}. 
	If   $\{ \edgeVec{\aParty}{\anIter} \}_{(\aParty,\anIter) \in \hidden}$ has at least $|\nHonest| -1$ independent vectors, then
	there exist $\indNoiseDiff$ and $\cancelNoiseDiffVec$ as defined in Lemma \ref{lm:dp.noiseDiff.abstract} such that  $\exec$ satisfies  \obsDP{$\epsilon$}{$\delta$}{$\advValView$}  (or \colDP{$\epsilon$}{$\delta$}{$\corrnodes$})
	for any  
	\[ 
	\sdInd^2 >  \frac{(\nHonest-1)c^2}{(\wOnl-1)^2 \epsilon^2} 
	\] 
	
	and  
	\[ 
	\sdCancel^2 \ge \frac{\|\cancelNoiseDiffVec\|^2_2}{
		\frac{\epsilon^2}{c^2} - 		\frac{(\nHonest-1)}{ (\wOnl-1)^2 \sdInd^2}},
	\] 
	where $c^2 > 2\ln(1.25/\delta)$. 
}

\begin{theorem} 
	\thmTotalSumStm{} 
	\label{thm:totalsum}  
\end{theorem} 

Theorem \ref{thm:totalsum} shows the needed increase of $\sdInd^2$ due to drop-out. We prove it in Appendix \ref{app:privacy.cdp}. As $\totalWeightParty{\aParty}$ is the proportion of $\valPriv{\aParty}$ injected to the computation for each $\aParty \in \partySetHonest$,  $\wOnl \in [\nHonest/\iterCnt, \nHonest]$ is the total amount for all private values. It is proportional to the online time of all parties. As  $\wOnl \in  O(n)$, then 
$\sdInd^2 \in O(c^2/n  \epsilon^2)$ which order-wisely matches $\sdInd^2$ in the setting without drop-outs. 

If a honest party drops-out permanently in an early iteration, it is possible that it was not able to send a sufficient number of messages in order for 
$\{ \edgeVec{\aParty}{\anIter} \}_{(\aParty,\anIter) \in \hidden}$ to have at least $|\nHonest| -1$ independent vectors. This would make Theorem \ref{thm:totalsum} inapplicable. However, as it has not completely canceled his correlated noise due to the permanent drop-out, it will remain protected by the uncanceled correlated noise terms even if $\sdInd^2$ is bounded. We reflect that possibility in the following theorem. 

\newcommand{\thmPermDropStm}{Let $\epsilon,\delta \in (0,1)$. Let $\exec = (\transitionset,\onlSymb,\advValView)$ be defined as in Theorem \ref{thm:dp.abstract}. Let $\coalSet \subseteq \partySetHonest$ be a coalition that contain all honest parties that did not drop-out permanently (i.e. $\onlNodes{\iterCnt} \cap \partySetHonest \subseteq \coalSet)$. 
	 Let  $\wOnl = \sum_{\aParty \in \partySetHonest} \totalWeightParty{\aParty}$, where $\totalWeightVec$ is defined in Equation \eqref{eq:totalWeight}
	 and  
	\[ 
	\hidden = \{ (\aParty, \anIter) \in \partySetHonest\times [0,\iterCnt] \setminus \advValView : \aParty \in \coalSet \land  \outNeigh{\aParty}{\anIter} \subseteq \coalSet \}.
	\]
	For all $(\aParty,\anIter) \in  \hidden$, let  $\edgeVec{\aParty}{\anIter}$ be as defined in Lemma \ref{lm:edgeVec}. 
	If   $\{ \edgeVec{\aParty}{\anIter} \}_{(\aParty,\anIter) \in \hidden}$ has at least $|\coalSet| -1$ independent vectors, 
	then
	there exist $\indNoiseDiff$ and $\cancelNoiseDiffVec$ as defined in Lemma \ref{lm:dp.noiseDiff.abstract} and  $\exec$ satisfies  \obsDP{$\epsilon$}{$\delta$}{$\advValView$}  (or \colDP{$\epsilon$}{$\delta$}{$\corrnodes$})
	for any  
	\[ 
	\sdInd^2 >  \frac{(|\coalSet|-1)c^2}{(\wOnl-1)^2 \epsilon^2} 
	\] 
	and  
	\[ 
	\sdCancel^2 \ge \frac{\|\cancelNoiseDiffVec\|^2_2}{
		\frac{\epsilon^2}{c^2} - 		\frac{(|\coalSet|-1)}{ (\wOnl-1)^2 \sdInd^2}},
	\] 
	where $c^2 > 2\ln(1.25/\delta)$. 
}

\begin{theorem}
	\thmPermDropStm{}
	\label{thm:coalition} 
\end{theorem}  

We prove this theorem in Appendix \ref{app:privacy.cdp}. 
Here, $\hidden$ is 
the set of pairs $(\aParty,\anIter)$ such that $\yVal{\aParty}{\anIter}$ has not been seen by the adversary and only sent between the members of $\coalSet$.
Theorem \ref{thm:coalition} shows how $\sdInd^2$ increases if honest parties cannot join the coalition $\coalSet$, which is an ``adequately connected component'' of honest users. 
%This could happen when a party drops out early and does not rejoin.
%\janfoot{It is unclear here why a party would be unable to join the coalition, except for the condition that they should be online for the dissemination phase (=did not drop out permanently.  It is unclear why a party which has an ``early drop out'' but later rejoins would not be able to join the coalition.} 
However, if over time a party performs sufficient exchanges with other parties, they are likely to be part of $\coalSet$.   
%\janfoot{Is 'joining the coalition' a choice or something for which the parties need to take a distinguished action?  If no, it is unclear what 'likely' means here.}
%
%Parties that dropped out too early will probably not be able to be part of coalition $\coalSet$ of Theorem \ref{thm:coalition}. However, $\sdInd$ remains bounded due to the extra protection of uncanceled noise of permanent dropouts. 

So far we have shown conditions under which DP can be achieved while $\sdInd^2$, which is an important factor to determine output accuracy, remains bounded.
%\janfoot{I don't understand what follows.  I try to formulate an alternative, please check whether that is what you mean.}
%\janrm{ However, $\sdCancel^2$ still depends on $\cancelNoiseDiffVec$, which does not have an explicit dependency in the rest of the parameters. Still, as $\iterCnt$ increases, the change in noise $\cancelNoiseDiffVec$ distributes over a larger number of terms of $\cancelNoiseVec$, reducing the change on each sample, which reduces $\| \cancelNoiseDiffVec \|_2^2$. In} 
%\janadd{
  However, $\sdCancel^2$ still depends on $\cancelNoiseDiffVec$.
  A large $\sdCancel^2$ implies a larger risk (in terms of additional error) in case a party drops out permanently.
  One can reduce  $\sdCancel^2$ by increasing $\iterCnt$.  Indeed, more iterations means that in every iteration a smaller fraction of $\valPriv{\aParty}$ is injected, and hence the needed $\|\indNoiseDiff\|$ and $\cancelNoiseDiffVec$ become smaller.  While this dependency on $\iterCnt$ isn't explicit from the above theorems, in
%}
Section \ref{sec:exp.dropouts}, we empirically show
%\janrm{
%	that $\sdCancel^2$ is sufficiently small so that \inca{} remains accurate in the case of dropouts. 
%}
%\janrm{
	how $\sdCancel^2$ depends on $\iterCnt$.
%}

\subsection{Topological Conditions for Privacy} 
\label{sec:privacy.negative}  

We show that certain conditions on the underlying topology of the parties' exchanges
already determine the number of linearly independent vectors of
$\{ \edgeVec{\aParty}{\anIter} \}_{(\aParty,\anIter) \in \hidden}$ for Theorems \ref{thm:cdp.nodrop} and \ref{thm:totalsum}. 

For all $\anIter \in [1, \iterCnt]$, let 
\[ \iterSimpHiddenEdges{\anIter} = \{ (\aParty,\partyB) \in \partySetHonest \times \partySetHonest : \transMatOnlEl{\anIter}{\partyB}{\aParty
} > 0 \land (\aParty, \anIter) \not\in \advValView  \}
\] 
be the set of edges whose messages the adversary has not seen. 
Let 
\begin{equation}
	\flatGraphHidden = \left(\partySetHonest, \bigcup_{\anIter=1}^{\iterCnt} \iterSimpHiddenEdges{\anIter}\right) 
	\label{eq:flatgraph.hidden} 
\end{equation} be the graph induced by these edges, which combines the unseen exchanges across iterations.  

\newcommand{\thmStronglyConnectedTotSumStm}{
	Let $\epsilon,\delta \in (0,1)$. Let $\exec = (\transitionset,\onlSymb,\advValView)$ be defined as in Theorem \ref{thm:dp.abstract}, $\wOnl = \sum_{\aParty \in \partySetHonest} \totalWeightParty{\aParty}$ and $\flatGraphHidden$ as defined in Equation \eqref{eq:flatgraph.hidden}. 
	If $\flatGraphHidden$ is strongly connected, then 
	$\exec$ satisfies  \obsDP{$\epsilon$}{$\delta$}{$\advValView$}  (or \colDP{$\epsilon$}{$\delta$}{$\corrnodes$})
	for $\sdInd^2$ and   $\sdCancel^2$ defined as in Theorem \ref{thm:totalsum}. 
}

\begin{theorem}
	\thmStronglyConnectedTotSumStm
	\label{thm.stronglyConnected.totalSum}
\end{theorem}  
We prove the theorem in Appendix \ref{app:privacy.topologies}. Essentially, we show  that if $\flatGraphHidden$ is strongly connected, $\{ \edgeVec{\aParty}{\anIter} \}_{(\aParty,\anIter) \in \hidden}$ has at least $\nHonest -1$ linearly independent vectors and therefore theorems \ref{thm:cdp.nodrop} and \ref{thm:totalsum} are applicable \footnote{Theorem \ref{thm:cdp.nodrop}  is a special case of Theorem \ref{thm:totalsum}, therefore the same reasoning applies.}. Analogous results can be obtained if the graph induced by the unseen interactions inside a coalition $\coalSet$ is strongly connected for  Theorem \ref{thm:coalition}  to apply. 

Finally, we provide a negative result on static interactions: that is, when all parties interact with the same neighbors in all iterations, then it is difficult to obtain privacy using our theorems. In particular if all the messages of two parties are observed, the preconditions of our main results are not met. 

\newcommand{\thmStaticStm}{
Let $\exec = (\transitionset,\onlSymb,\advValView)$ be defined as in Theorem \ref{thm:dp.abstract} 
%	and $\wOnl = \sum_{\aParty \in \partySetHonest} \totalWeightParty{\aParty}$ and $\flatGraphHidden$ as defined in Equation \eqref{eq:flatgraph.hidden}. 
Let $\hidden =  \partySetHonest\times [0,\iterCnt] \setminus \advValView$  be the set of pairs $(\aParty,\anIter)$ such that $\yVal{\aParty}{\anIter}$ is not seen by the adversary.  For all $(\aParty,\anIter) \in  \hidden$, let $\edgeVec{\aParty}{\anIter}$ be as defined in Lemma \ref{lm:edgeVec}. 
If parties do not change neighbors across iterations, (i.e., $\transMatOnl{\anIter} = \transMatOnl{1}$ for all $\anIter \in [2,\iterCnt]$) and there exist $\partyA,\partyB \in \partySetHonest$ such that $(\partyA,\anIter) \in \advValView$ and $(\partyB,\anIter) \in \advValView$ for all $\anIter  \in [0,\iterCnt]$, we have that 
$\{ \edgeVec{\aParty}{\anIter} \}_{(\aParty,\anIter) \in \hidden}$ has less than $\nHonest -1$ independent vectors.  
}
\begin{theorem}
\thmStaticStm 
	\label{thm:static} 
\end{theorem}  
We prove it in Appendix \ref{app:privacy.topologies}. Essentially if all messages of two honest parties are observed and no party changes neighbors in the entire execution, it is not possible to meet the conditions for privacy with bounded $\sdInd^2$. However, changes in the neighborhood can occur accidentally with dropouts. In that case, negative conditions do not apply. Note that when parties do not drop out, negative conditions are obtained easily if at least one corrupted party exchanges messages with two honest parties. Therefore, to increase the amount of different vectors in $\{ \edgeVec{\aParty}{\anIter}\}_{(\aParty, \anIter) \in \hidden}$, it is beneficial that parties change neighbors as much as possible. This aligns with \cite{touat2024scrutinizing}, which shows that dynamic networks are beneficial for privacy in decentralized learning as they increase the mixing speed of messages.  
%\begin{itemize}
%\item if graph is disconnected, not possible to obtain l.i. vectors
%\item if the graph is weakly connected and the degree of each node is at most 1 on each iteration, then privacy is guaranteed 
%	\end{itemize}  

% !TeX root = main.tex

\section{Empirical Evaluation}
\label{sec:exp}

%In Section \ref{sec:privacy}, we presented the conditions on the variance of noise terms and exchanges such that \inca{} is differentially private. 
In this section, we empirically evaluate \inca{}. 
The main question we aim to answer is: 

\emph{How does the mean squared error and communication cost of \inca{} compare with existing differentially private and decentralized techniques? }

We perform our evaluation under different adversarial settings, dropout regimes and communication parameters  of \inca{}.  
%We retain the settings that perform best compared to existing techniques. 
In all cases, parties will generate their per-iteration neighbors by choosing $k$ random parties from $\partySet$ to be their outgoing neighbors. 
Message weights are chosen evenly: for all $\anIter \in [1,\iterCnt]$,  $\transMatEl{\anIter}{\partyB}{\aParty} = 1/(k+1)$ for each $\partyB \in \outNeigh{\aParty}{\anIter}$ and  $\transMatEl{\anIter}{\aParty}{\aParty} = 1/(k+1)$.  We measure the mean squared error (MSE) of the privacy preserving estimations with respect to the ground truth over a set of samples. We measure the  communication cost of a protocol by counting the number of exchanges and iterations protocols perform. 

First, in Section \ref{sec:exp.nodrop} we perform our evaluation  when parties do not drop-out. Then we do it  in the presence of permanent drop-outs in Section \ref{sec:exp.dropouts}.

\subsection{Performance without Dropouts}
\label{sec:exp.nodrop}

\begin{figure}
	\includegraphics[width=0.4\textwidth]{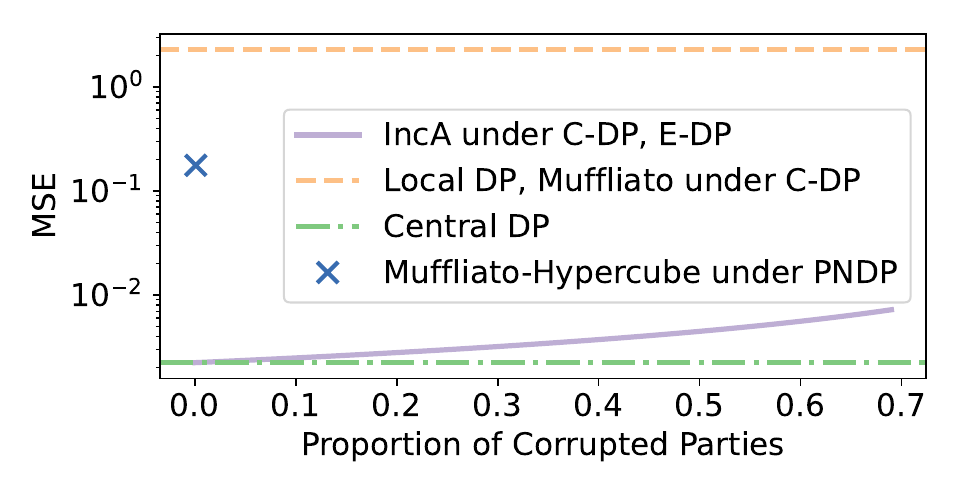}
	\caption{MSE of Local DP, Central DP, \inca{} and Muffliato using a Hypercube graph for $n=2^{10}$, $\epsilon=0.1$ and   $\delta=10^{-5}$. The MSE of \inca{} is in function of $\nHonest$ in the $x$-axis. } 
	\label{fig:accuracy1}
\end{figure} 

We analyze the performance of Algorithm \ref{alg:ourprotocol} in the absence of dropouts.  We start by discussing the error of \inca{}, given by Theorem \ref{thm:cdp.nodrop}. After that, we show the communication effort required such that preconditions to apply this theorem are met.

%  applies when conditions [point to condition equations ] are met.  We first illustrate the accuracy we obtain under these conditions an then show the communication cost to obtain it. 

\paragraph{Mean Squared Error} As a result of Theorem \ref{thm:cdp.nodrop}, we concluded in Section \ref{sec:privacy.optimal} that Algorithm \ref{alg:ourprotocol} produces an estimate with MSE equal to $\frac{c^2}{\nHonest n \epsilon^2}$. 
%This \emph{closely matches Central DP} when parties do not collude and $\nHonest = n-1$ (as we account for the degradation of privacy from the knowledge of every single party). 
In Figure \ref{fig:accuracy1}, we illustrate this for $n=1024$,  comparing \inca{},  Central DP, Local DP and Muffliato\cite{cyffers2022muffliato}. The latter is an approach for differentially private averaging that relaxes the adversary by reporting the mean privacy loss under Pairwise Network DP (PNDP). We show the error of \inca{} under \obsDPAcron{}, where there are no corrupted parties and \colDPAcron{}, where the proportion of corrupted parties varies in the $x$-axis from 0 to 0.7 (the error under \obsDPAcron{} correspond to the left extreme of the \inca{} curve). This shows the degradation of \inca{} due to collusion, which allows the adversary to know a proportion of independent noise terms $\{\indnoise{\aParty} \}_{\aParty \in \corrnodes}$. 
%As collusion  gets larger, parties add a higher amount of noise to compensate for the part that is learned by the adversary.  
The error of \inca{} is much closer to Central DP than to Local DP even when more than $70\%$ of the parties are corrupted.  
% This holds regardless of the value of $\epsilon$ and $\delta$, as  the curves would  keep the same proportion if these parameters are changed.  
For Muffliato, we use hypercube topologies (which provide the best privacy-utility trade-offs) under no collusion. As discussed in Section \ref{sec:prelim.threat}, the adversary of PNDP is substantially weaker than \obsDPAcron{} and \colDPAcron{}. We can see that the error of Muffliato is significantly higher than \inca{}, even with a weaker adversary. Moreover, when executed under \obsDPAcron{} or \colDPAcron{}, Muffliato matches the error of Local DP. We provide more details on the calculations for Muffliato in Appendix \ref{app:exp.muffliato}.  
When executed without dropouts, existent correlated noise techniques have the same error than ours \cite{sabater2022accurate,vithana2025correlated}. However, they are significantly impacted by dropouts, as we show in Section \ref{sec:exp.dropouts}. 

\paragraph{Communication Cost} The more iterations $\iterCnt$ that \inca{} performs, the higher is the likelihood  to obtain $\nHonest-1$ linearly independent vectors in $\{ \edgeVec{\aParty}{\anIter} \}_{(\aParty,\anIter) \in \hidden}$ in order to have the error given by Theorem \ref{thm:cdp.nodrop}.
%as shown in Figure \ref{fig:accuracy1}.
We analyze the communication effort required to obtain these conditions. We evaluate \inca{} when each party communicates with $k \in \{1,2,3,4,5\}$ neighbors per iteration. We consider  $\iterCnt \in \{2,4,6,8,10\}$, where  50\% of messages per iteration chosen randomly are observed by the adversary or 30\% of the parties are chosen randomly to collude with the adversary.   For each parameter combination we run the protocol 100 times and count how many times preconditions of Theorem \ref{thm:cdp.nodrop} are met. 

% $\nHonest-1$ linearly independent vectors in $\{ \edgeVec{\aParty}{\anIter} \}_{(\aParty,\anIter) \in \hidden}$.

Our results are shown for $n = 100$ in Figure \ref{fig:extract.success.nodrop}, with the  success rate as a function of the number of iterations ($\iterCnt$) or total number of messages ($k\iterCnt$) per party. The left and center figures show results when the adversary only observes messages (\obsDPLabel{}). In Figure \ref{fig:extract.success.nodrop.obs.it}, we can see that the higher $k$ is, the lower is the number of iterations $\iterCnt$ required for \inca{} to get 100\% success. However, as shown in Figure  \ref{fig:extract.success.nodrop.obs.msg}, lower $k$ requires less number of messages in total.
In Figure \ref{fig:extract.success.nodrop.corr.it}, we switch to \colDPLabel{}, where parties are corrupted. When  $k$ is lower, it is easier to achieve higher success rate. 
This happens because bigger values of $k$ reduce the size  of $(\edgeVec{\aParty}{\anIter})_{(\aParty,\anIter)\in \hidden}$ (and therefore the possibility that this set has $\nHonest-1$ linearly independent vectors), as the corrupted nodes receive more messages. 
This results hold for any Valid $(\xCoeffVec,\matrixPart)$-Gaussian $\partDistrSymb$, as when there are no dropouts, $\{ \edgeVec{\aParty}{\anIter} \}_{(\aParty,\anIter) \in \hidden}$ only depends on the communication matrices $\transitionset$.
%In Appendix \ref{app:exp}, we provide an extensive analysis of this experiment varying the percentage of observed messages and corrupted nodes, in which we arrive to similar conclusions. \cesarfoot{remove if this appendix is not ready for the deadline}

\begin{figure*}[htbp] 
	\begin{subfigure}[b]{0.3\textwidth}
		\includegraphics[width=\textwidth]{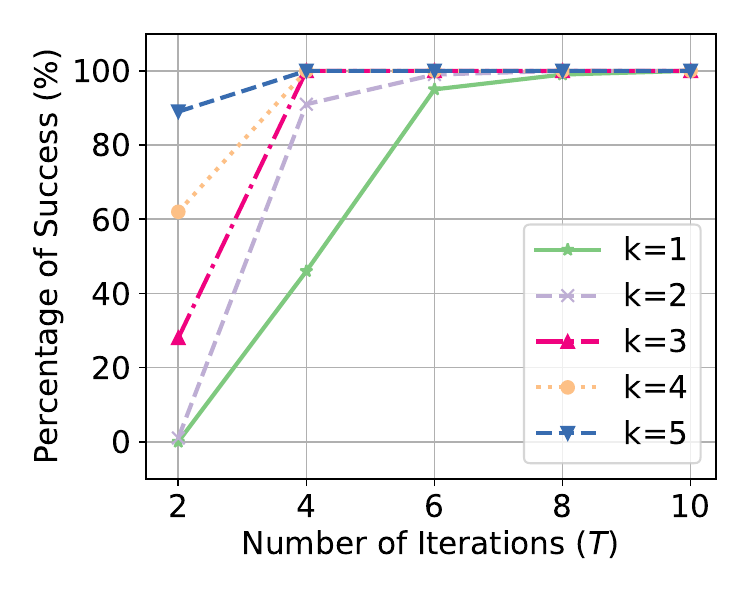}
		\caption{50\% of observed messages}
		\label{fig:extract.success.nodrop.obs.it}
	\end{subfigure}
	\begin{subfigure}[b]{0.3\textwidth}
		\includegraphics[width=\textwidth]{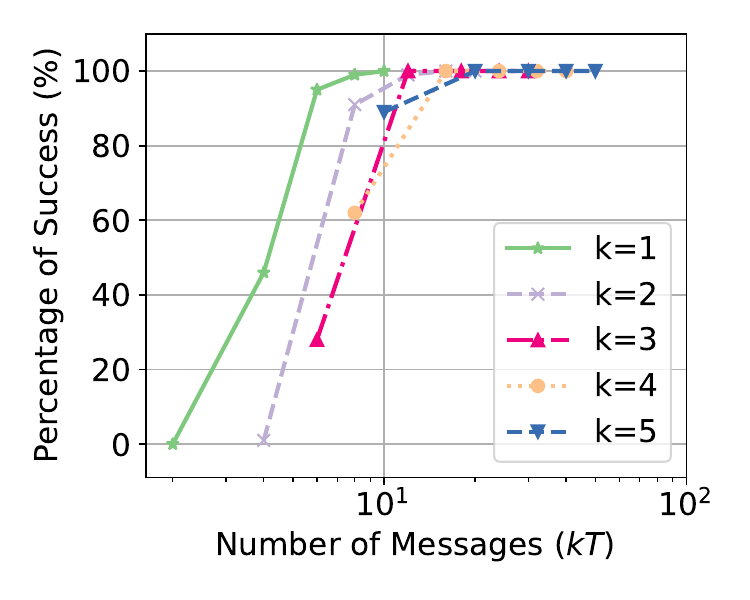}
		\caption{50\% of observed messages}
		\label{fig:extract.success.nodrop.obs.msg}
	\end{subfigure}
	\begin{subfigure}[b]{0.3\textwidth}
		\includegraphics[width=\textwidth]{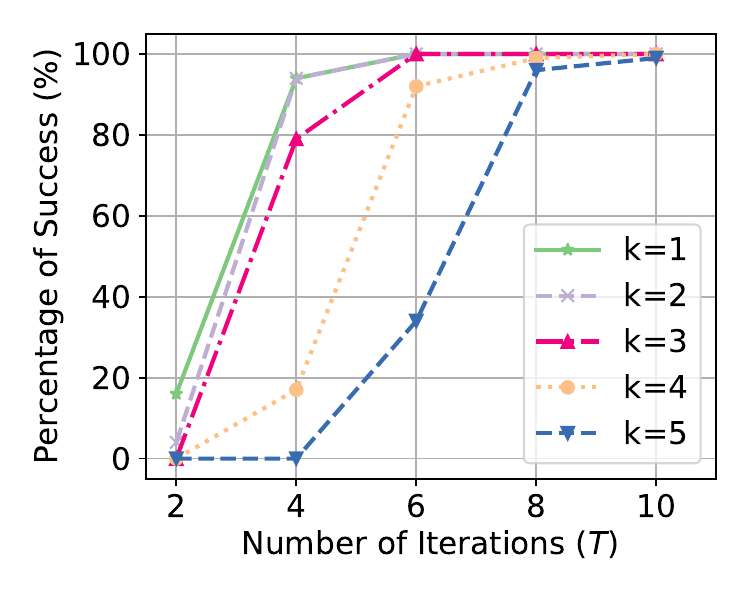}
		\caption{30\% of corrupted parties}
		\label{fig:extract.success.nodrop.corr.it}
	\end{subfigure}
	\caption{Percentage of success in function of the number of iterations (figures \ref{fig:extract.success.nodrop.obs.it} and  \ref{fig:extract.success.nodrop.corr.it}) and the number of messages (Figure \ref{fig:extract.success.nodrop.obs.it})  to meet the preconditions of Theorem \ref{thm:cdp.nodrop} for $k \in \{1,2,3,4,5\}$, where the adversary only observes messages or corrupts a subset of parties.  The total number of parties is $n=100$.}
	\label{fig:extract.success.nodrop}
\end{figure*}

\paragraph{Larger Graphs and Safer Parameters} To conclude Section \ref{sec:exp.nodrop}, we show how the amount of parties affect the communication cost. We analyze the number of messages required to apply Theorem \ref{thm:cdp.nodrop} when $n$ increases. For all sets of parameters, we run the protocol $10^5$ times and in all cases we obtain 100\% success. We focus on \colDPLabel{} with 50\% of corrupted parties. As seen in our previous experiment, the best $k$ under \colDPAcron is equal to $1$ as it reduces the number of messages and $\iterCnt$. In Figure \ref{fig:k-succeed}, the number of parties $n$ is chosen from  $\{100,500, 1000, 5000\}$.  Even with $5000$ parties, the communication cost is low: the amount of messages (and $\iterCnt$) required is slightly higher than 15. For this experiment, we made a slight modification in the way parties made exchanges: at each iteration, each party chose an outgoing neighbor that he/she had not chosen in any previous iteration.

\begin{figure}
	\includegraphics[width=0.35\textwidth]{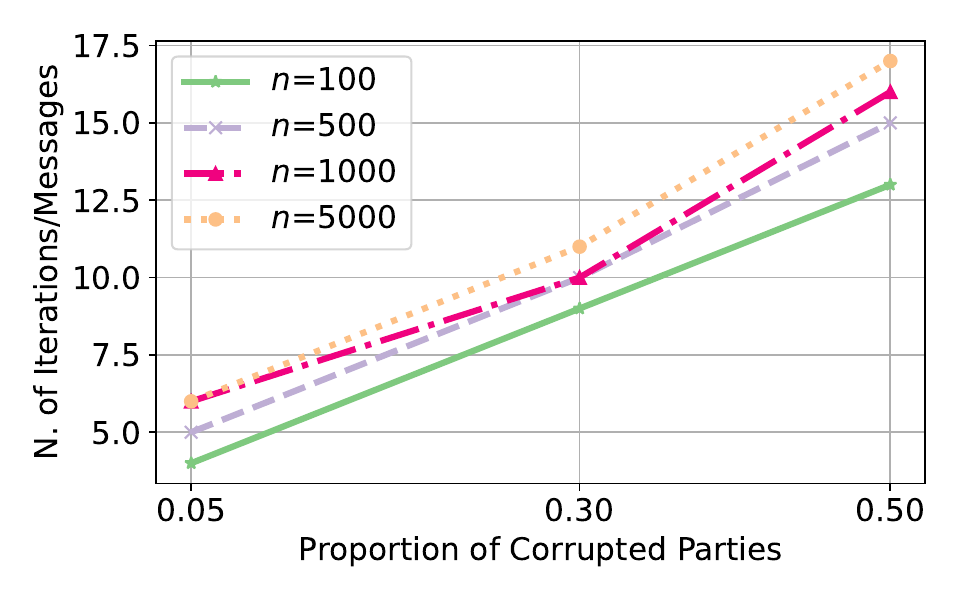}
	\caption{Minimum number of iterations (and messages since $k=1$) of \inca{}  such that it obtains $100\%$ success over $10^5$ runs of the protocol for $n \in \{100,500,1000,5000\}$ considering a proportion of corrupted nodes between $0.05\%$ and $0.5$.} 
	\label{fig:k-succeed}
\end{figure}

\subsection{Performance when Dropouts Occur} 
\label{sec:exp.dropouts}

%Now we analyze the performance of Algorithm \ref{alg:ourprotocol.drop}, where drop-outs can occur under \colDPLabel{}. We will consider that, during the execution of the protocol, a  proportion $\propDrop \in (0,1)$  random parties will permanently drop-out at a random iteration $\anIter \in [1,\iterCnt]$. 

We now analyze \inca{} in the presence of  dropouts under the \colDPLabel{} threat model. We assume that a proportion $\propDrop \in (0,1)$ of parties permanently drop out at a random iteration.
We use the $\partDistrIncremSymb$  (Example~\ref{ex:increm.inj}) as parameter for $\partDistrSymb$, in which private values are evenly injected across all iterations, thereby reducing the required $\sdCancel^2$. The error of \inca{} under dropouts depends on several factors. They consist on
\begin{enumerate}
	\item the variance $\sdInd^2$ of terms $\indNoiseVec$, which are not canceled   \label{it:dropfactor.1} 
	\item the variance of $\sdCancel^2$ terms $\cancelNoiseMat$  \label{it:dropfactor.2} 
	\item the amount of uncanceled terms $\cancelNoiseMat$  \label{it:dropfactor.3} 
	\item the coefficients $\totalWeightVec$ of $\vectPriv+\indNoiseVec$
%	, as they will define how much of the private values and independent noise will be injected in the protocol.   
	\label{it:dropfactor.4} 
\end{enumerate}
Theorems~\ref{thm:totalsum} and \ref{thm:coalition} provide bounds on $\sdInd^2$ (factor \ref{it:dropfactor.1}). We experimentally measure the impact of factors \ref{it:dropfactor.1}-\ref{it:dropfactor.4} in two experiments. The first experiment determines $\sdCancel^2$ (factor \ref{it:dropfactor.2}) and in the second we jointly assess the relation of all  factors. 

%Our theorems \ref{thm:totalsum} and \ref{thm:coalition} for Algorithm \ref{alg:ourprotocol.drop} ensure a bounded $\sdInd$, which largely impacts accuracy but do not show how the uncanceled terms of $\cancelNoiseVec$ or the partial injection of $\vectPriv$ affect the estimation. We study these factors experimentally in the following via two experiments. The first experiment  studies which values of $\sdCancel^2$ are safe to obtain DP guarantees. 

%In the second experiment, we use safe parameters of $\sdInd^2$ given by Theorem  \ref{thm:coalition} and $\sdCancel^2$ given by the first experiment to perform estimations and measure the MSE. For both experiments, we consider a proportion of corrupted parties. We run  Algorithm \ref{alg:ourprotocol.drop} with distribution $\partDistrIncremSymb$ of Example \ref{ex:early.inj.drop}, which is tailored for drop-outs. 
\newcommand{\indVarDeg}{\alpha}

\paragraph{First Experiment} For each run, we fix parameters  $\epsilon$, $\delta$, $n$, $k$, $\iterCnt$, the proportion of corrupted parties $\propCor \in [0,1)$, the dropout rate $\propDrop$, and the independent noise variance  $\sdInd^2$.
Then, we compute the required value of $\sdCancel^2$ such that \inca{} satisfies \colDP{$\epsilon$}{$\delta$}{$\corrnodes$} according to Theorem~\ref{thm:coalition}.
 We set $\sdInd^2 = \indVarDeg 2\ln(1.25/\delta)/\nOnlHonest \epsilon^2$ where $\nOnlHonest = n( 1- \propDrop - \propCor)$ and a degradation parameter $\indVarDeg$. The quantity $2\ln(1.25/\delta)/\nOnlHonest \epsilon^2$ is the theoretical lower bound of $\sdInd^2$ according to Theorem \ref{thm:coalition} when $\nOnlHonest$ honest users are online all iterations. As $\sdInd^2$ must be bigger than the theoretical bound, we chose $\indVarDeg=1.3$, providing a good balance between $\sdInd^2$ and $\sdCancel^2$. 
For each set of parameters, we perform 1000 runs for each parameter set and keep the worst case $\sdCancel^2$.

\paragraph{Second Experiment} For the second experiment, each run the protocol estimates the average of a random vector $\vectPriv \in [0,1]^n$ using the same sets of parameters than the first and worst case values of $\sdCancel^2$. We perform 1000 runs and compute the MSE of the estimations with respect to the true average. 

\paragraph{Comparison of Different Topologies} Figure \ref{fig:drop.topologies} shows the results for $n=200$, $\epsilon=0.2$, $\delta=10^{-5}$ and proportion  of corrupted parties $\rho=0.1$. We tested \inca{} with parameters  $k \in \{1,3,5\}$ and $T \in \{10,20\}$. It is clear that, for the same amount of permanent drop-outs, higher $\iterCnt$ decreases the error. Parameter $k$ impacts different factors that determine the MSE. First, as for the case without dropouts, higher $k$ increases the number of observed messages, which result in a larger $\sdCancel^2$ to achieve DP guarantees. Second, higher $k$ also increases the mixing speed of the protocol reducing $\sdCancel^2$. 
%Third, increasing $k$ makes canceling terms $\cancelNoiseMat$ to spread  faster in the network, which is negative if they are not canceled. 
In experiments, these effects compensate for $k \in \{ 1,3,5\}$, providing similar results, although smaller $k$ yields smaller error and fewer amount of messages per iteration. 

%On the other hand, increasing $k$ also increases the mixing speed of the protocol, reducing the required $\sdCancel^2$, impacting positively on accuracy. In addition, $k$ has an effect on the amount of uncanceled terms $\cancelNoiseVec$ in the final average. New noise terms  $\cancelNoiseVec$ added by parties to the system spread more quickly with higher $k$. Also they cancel more quickly. In the experiments shown in Figure \ref{fig:drop.topologies}, these effects seem to partially compensate, showing a slight advantage for lower $k$ when $\iterCnt$ is fixed and corrupted nodes are present. If we consider that the amount of messages per party, which is equal to $k \iterCnt$, it is clear that \inca{} with $k=1$ minimizes the communication effort, as other cases will require $k$ times more messages and would not improve accuracy. 
\begin{figure}[htbp] 
	\begin{subfigure}[b]{0.3\textwidth}
		\includegraphics[width=\textwidth]{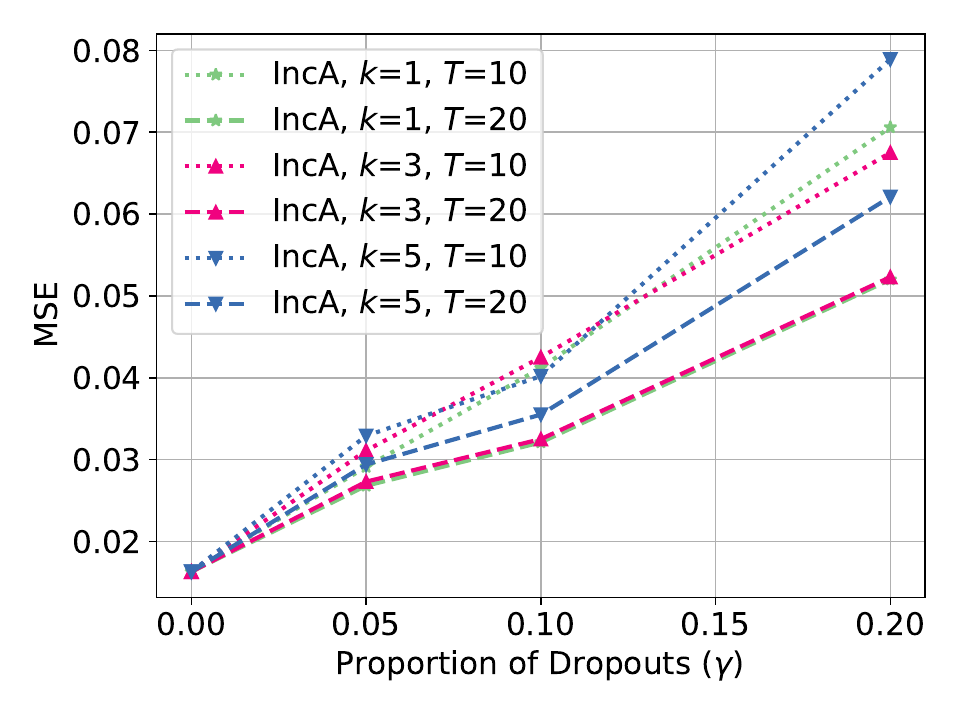}
		\caption{Parameters of \inca{}}
		\label{fig:drop.topologies}
	\end{subfigure}
	\begin{subfigure}[b]{0.3\textwidth}
		\includegraphics[width=\textwidth]{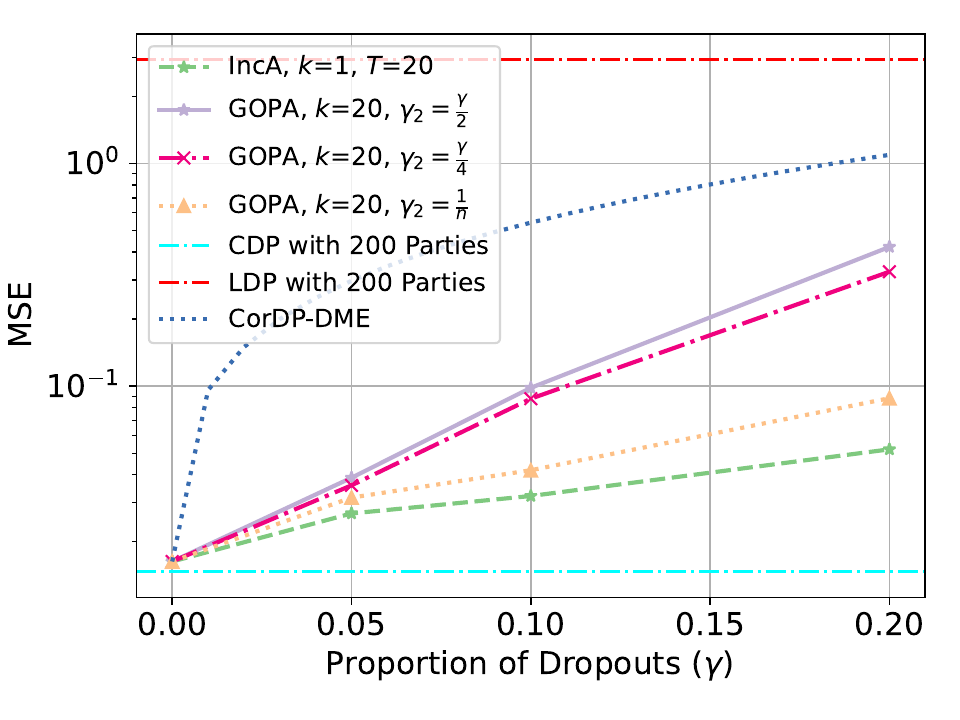}
		\caption{\inca{} and Previous Work}
		\label{fig:drop.sota}
	\end{subfigure}
	\caption{Mean Squared Error (MSE) in function of the proportion $\propDrop$ of dropouts for \inca{}, \gopa{} and \cordp{} protocols, over 1000 runs where $n=200$, $\epsilon=0.2$, $\delta=10^{-5}$ and proportion  of corrupted parties $\rho=0.1$.  As references we include the MSE of Central DP (CDP) and Local DP (LDP) with $n$ parties.} 
	\label{fig:drop}
\end{figure}

\paragraph{Comparison with Related Techniques} We compare our protocol with the GOPA  \cite{sabater2022accurate} and Cor-DP-DME \cite{vithana2025correlated}, that are fully decentralized and provide similar MSE when there are no dropouts. These protocols follow the same base protocol. In a first round parties  exchange pairwise canceling noise terms that they add to their private value to protect privacy. In a second round, they release the privatized noisy values, whose sum cancels the majority of the noise. In the case of dropouts, some of these terms remain uncanceled and  increase the error.  The more exchanges parties perform  with each other, the smaller is the variance of correlated noise and the lower is the error under dropouts. \cite{vithana2025correlated} is a particular case of GOPA when all parties communicate as much as possible to reduce the noise variance. 
%However, even in this case  uncanceled  noise terms significantly increase the MSE.
GOPA provides an additional dropout mitigation phase in the second round, where parties rollback their noise. However, if  dropouts occur in this round, uncanceled terms remain.

% The dynamics of these techniques  rely on a base 2 step structure. In the first round, parties interact among peers to exchange pairwise canceling noise terms between each other. They use this noise to obfuscate their private value, which they disseminate in the second round. As in our case, as pairwise noise is canceling, it does not harm accuracy. If dropouts occur in first round, then some noise terms will remain uncanceled and decrease accuracy. To reduce the impact of dropouts, the GOPA protocol \cite{sabater2022accurate} proposes a rollback mechanism that can be performed in the second round of dissemination. However, some noise terms can still stay uncanceled due to dropouts in the dropouts phase. In addition, both  GOPA and the Cor-DP-DME protocol \cite{vithana2025correlated} study the reduction of the variance of pariwise noise terms when the amount of exchanges increase. In particular Cor-DP-DME studies the harm of uncanceled noise terms when no rollback phase is performed, but all parties exchange pairwise noise with each other. The impact of noise terms without rollback is still very large. 
Figure \ref{fig:drop.sota} shows the comparison of \inca{}, GOPA and Cor-DP-DME. For GOPA, $k$ is the number of exchanges at each round. In Cor-DP-DME each party performs $n$ messages per round. Similarly to  \inca{}, GOPA is evaluated with two experiments of 1000 runs,  first determining worst-case empirical variance and then the MSE.
%which are significantly smaller than the theoretical values reported in \cite{sabater2023private}, and then computing the MSE. 
Both experiments are done in 1000 runs. For Cor-DP-DME, we ignore the error induced by missing private values due to drop-out and only compute the error induced by the variance of noise terms when they are uncanceled, which is a lower bound on the real error and can be computed theoretically. 

We evaluate all of the protocols for a proportion of dropouts $\gamma \in \{ 0.05, 0.1, 0.2\}$,  and set GOPA and \inca{} to send 20 messages: $k=20$ for GOPA and $(k,T)= (1,20)$ for \inca{}. GOPA is very sensitive to the distribution of drop-outs over its two rounds. We denote by $\propDrop_2 \in [0,\gamma]$ to the proportion of dropouts that happen in the rollback round. We consider  $\propDrop_2 \in \{\frac{\propDrop}{2}, \frac{\propDrop}{4}, \frac{1}{n}\}$, note that $\propDrop_2 = \frac{1}{n}$ assumes an extremely optimistic scenario for GOPA, in which only a single party drops out in the rollback round. In all cases, we observe that  \inca{} outperforms the other protocols for  the same proportion of dropouts. If $\propDrop_2 = \frac{1}{n}$ GOPA approaches \inca{}, but significantly degrades if  $\propDrop_2 \in \{ \frac{\propDrop}{2}, \frac{\propDrop}{4} \}$. Cor-DP-DME provides the worst performance even if sends substantially more messages, showing the impact when noise terms are not rolled back.

\paragraph{Discussion} Each protocol may experience a different  number and distribution of dropouts. 
In practice, this  is influenced by factors  such as the runtime complexity (which is discussed in more detail in   \cite{gupta2016model,dutta2018slow,wang2019matcha}) and the ability of parties to rejoin the computation. The former depends on the number of iterations  ($\iterCnt$ for \inca{}) and their duration. The duration of an iteration is influenced by the number of computations and  messages   ($k$ for \inca{}). \inca{} allows  dropped parties to rejoin the computation. Therefore, permanent dropouts do not necessarily increase with  $\iterCnt$, as longer runtime may increase the probability of a party to rejoin. GOPA requires only two iterations but substantially more messages per iteration than \inca{} and rejoining is only allowed until parties start the rollback phase.

\section{Related Work}
\label{sec:relatedwork}

Many distributed data processing tasks can be accomplished through local computations and subsequent aggregation of their outcomes across participating entities~\cite{bharadwaj2024federated,mcmahan2017communication,kairouz2021advances,nicolas2024differentially, nicolas2025secure,barczewski2025differentially}. This makes averaging a simple, yet powerful primitive. Averages can be computed in a federated way \cite{bharadwaj2024federated,mcmahan2017communication} or in a decentralized way using gossip protocols \cite{boyd2006randomized,berthier2020accelerated,pittel1987spreading,demers1987epidemic}. However, many early approaches required that parties share their data in the clear, which may be exploited by external attackers \cite{melis2019exploiting,geiping2020inverting,mrini2024privacy,touat2024scrutinizing,kariyappa2023cocktail,pasquini2023security}. 

\paragraph{Differentially Private Mean Estimation} We study approaches that prevent information leakage by providing differential privacy guarantees \cite{dwork2006calibrating}. Satisfying DP guarantees typically requires to perturb the computation with noise, which compromises accuracy. In the central model of DP (CDP), where a trusted curator performs the computation, one can achieve a mean squared error of $O(1/n^2)$ induced by privacy noise~\cite{dwork2014algorithmic}. If such a party is not available, parties could add noise before sharing their sensitive information, which further compromises accuracy. In the local model of DP (LDP)\cite{duchi2013local,kasiviswanathan2011can,kairouz2014extremal,kairouz2016discrete,chen2020breaking}, where parties fully privatize their data before using it in a collaborative computation,
%assume that their contributions are revealed as soon as they share them, 
the estimation error is a factor of $n$ greater than in CDP and only yields acceptable accuracy when the number of parties is massive \cite{erlingsson2014rappor,ding2017collecting}. 

\paragraph{Approaches based on Cryptographic Primitives}
%\janrm{If we relax LDP assumptions and consider that only a fraction of parties collude to compromise privacy, one could use}
%\janfoot{For example, additive/arithmetic secret sharing assumes all parties are malicious, no relaxation of assumption is needed.}
Instead of a trusted curator, one can also use
cryptographic primitives such as Secure Aggregation \cite{bonawitz2017practical,bell2020secure,taiello2024let,bell2023acorn}, other types of Multiparty Computation \cite{yao1982protocols,damgaard2013practical} or Shuffling \cite{ishai2006cryptography}. Using these techniques, one can recover similar privacy utility trade-offs of Central DP \cite{kairouz2021distributed,agarwal2021skellam,chen2022poisson,cheu2019distributed,balle2020private,ghazi2020private,heikkila2025using} or even exactly match them \cite{jayaraman2018distributed,dwork2006our,sabater2023private}. Implementing secure shuffling or general purpose multiparty computation approaches such as  
\cite{damgaard2013practical} requires a large communication cost or substantially  relaxes assumptions by considering multiple non-colluding servers or trusted parties that anonymize messages by shuffling them. Notably, some secure aggregation protocols \cite{bell2020secure,so2021turbo} require only $O(\log(n))$ communication per party, however, they rely on structured interaction approaches and  central coordination. 

\paragraph{Fully Decentralized Approaches} Fully decentralized approaches can reduce the noise required to achieve DP by relaxing adversary models or using correlated noise techniques. Existing relaxation techniques consider Network DP \cite{cyffers2022privacy} and Pairwise Network DP \cite{cyffers2022muffliato,cyffers2024differentially}. Both of these models substantially restrict the view of the adversary, making these models less applicable in practice (see our detailed discussion of Section \ref{sec:prelim.threat}). \cite{cyffers2022privacy,cyffers2024differentially} focus on decentralized stochastic gradient descent, while \cite{cyffers2022muffliato} is a protocol for private mean estimation. In Section \ref{sec:exp.nodrop}, we show that our approach outperforms \cite{cyffers2022muffliato} in accuracy despite considering a significantly stronger adversary. Correlated noise approaches \cite{imtiaz2019distributed,sabater2022accurate,allouah2024privacy,vithana2025correlated}  use pairwise canceling noise to protect privacy while keeping accuracy comparable to central DP when no dropouts occur. However, their accuracy degrades significantly if parties disconnect in the middle of the computation as pairwise noise does not cancel. This effect can be reduced up to some point with dropout mitigation techniques proposed by \cite{sabater2022accurate} but do not entirely solve the problem, as more dropouts can also occur in the mitigation phase. In Section \ref{sec:exp.dropouts} we provide a detailed comparison with \cite{sabater2022accurate,vithana2025correlated}. We do not compare with \cite{imtiaz2019distributed} and \cite{allouah2024privacy} since the former relies on Secure Aggregation, which was discussed previously, and the latter does not address the mean estimation task. Finally some correlated noise techniques also consider malicious adversaries who corrupt parties to actively deviate from the protocol to bias the computation \cite{sabater2022accurate,allouah2024privacy}.

% !TeX root = main.tex

\section{Conclusion and Future Work} 
\label{sec:conclusion}
%\cesar{draft conclusion, to proofread}
We propose \inca{}, a protocol for fully decentralized mean estimation, a key primitive for decentralized and privacy-preserving computation. We demonstrate that \inca{} not only satisfies theoretical privacy guarantees but also maintains accuracy despite party failures and collusion, all while incurring reasonable communication costs.
Directions of future work include the evaluation of our technique within larger systems such as the computation of federated learning models, the derivation of explicit bounds on the variance of correlated noise and the theoretical quantification of the probability that our conditions hold under random communication graphs.

\bibliographystyle{ACM-Reference-Format}
\bibliography{main}

%%
%% If your work has an appendix, this is the place to put it.
\appendix

%\input{symbol_table}

% !TeX root = main.tex

\section{Missing Theoretical Material}
\label{app:privacy} 

\subsection{Proof of Lemma \ref{lm:validgauss}} 
\label{app:protocol} 

Lemma \ref{lm:validgauss}. \lemmaValidGaussStm
\begin{proof}
	$\partDistrEarlySymb$ is a $(\xCoeffVec, \matrixPart)$-Gaussian where $\xCoeffVec$ and $\matrixPart$ are as follows:  
	\begin{itemize} 
		\item $\xCoeff{\anIter} = 1$ for $\anIter=0$ and  $\xCoeff{\anIter} = 0$ for $\anIter\in[1,\iterCnt]$,
		\item $\matrixPartEl{0}{:} = \oneVec^\top$,
		\item for all $t \in [1,\iterCnt]$,   \[\matrixPartEl{\anIter}{k} = \begin{cases} 
			-1 \text{\quad if $k = \anIter$} \\
			0 \text{\quad if $k \in [1,\iterCnt] \setminus \{\anIter\}$}.
		\end{cases}\]
	\end{itemize}

	$\partDistrIncremSymb$ is a $(\xCoeffVec, \matrixPart)$-Gaussian where 
	\begin{itemize}  
		\item $\xCoeffVec = \oneVec/(\iterCnt+1)$
		\item for all $\anIter \in [0,\iterCnt]$ and $k \in [1,\iterCnt]$
		\[ \matrixPartEl{\anIter}{k} = \begin{cases}  
			1 \text{\quad if $k=t+1$,}\\
			-1 \text{\quad if $k=t$,}\\ 
			0 \text{\quad otherwise.}
		\end{cases} 
		\] 
	\end{itemize} 
	For both $\partDistrEarlySymb$ and $\partDistrIncremSymb$, $(\xCoeffVec, \matrixPart)$ and $\matrixPartEl{-\iterCnt}{:}$ are invertible matrices.  
\end{proof}

\subsection{Knowledge in Linear Relations} 
\label{app:privacy.knowledge}

We now show how to construct the system of linear equations presented in Equation \eqref{eq:linear.adv}. Recall that in our privacy proofs, each party $\aParty \in \partySet$ uses  $\partDistrGauss{\valPriv{\aParty}+\indnoise{\aParty}}{\xCoeffOnlVec{\aParty}}{\matrixPartOnl{\aParty}}$ as the parameter for $\partDistrOnlParty{\aParty}$ in Algorithm \ref{alg:ourprotocol.drop}. Recall that $\xCoeffOnlVec{\aParty} \in \R^{\iterCnt+1}$ and $\matrixPartOnl{\aParty} \in \R^{(\iterCnt+1)\times\iterCnt}$ dictate how private values and canceling noise terms are introduced at each iteration for party $\aParty$.  
%
%Replacing $\valPart{\aParty}{\anIter}$ using the definition of $\partDistrGaussSymb$ in Section \ref{sec:protocol.gauss}, we have that 
%\begin{align}
%	\yVal{\aParty}{0} &= \xCoeffOnl{\aParty}{0} (\valPriv{\aParty}+\indnoise{\aParty}) + \sum_{k=1}^\iterCnt  \matrixPartOnlEl{\aParty}{0}{k} \cancelNoise{\aParty}{k} \\ 
%	 	\yVal{\aParty}{\anIter } &= \sum_{\partyB =1}^n \transMatOnlEl{\anIter}{\aParty}{\partyB} \yVal{\partyB}{\anIter} +  \xCoeffOnl{\aParty}{\anIter} (\valPriv{\aParty}+\indnoise{\aParty}) + \sum_{k=1}^\iterCnt  \matrixPartOnlEl{\aParty}{\anIter}{k} \cancelNoise{\aParty}{k} \\   & \qquad \text{for all $\anIter \in [1, \iterCnt].$}
%\end{align}  

\paragraph{Vectorized Gossip}
Let $\yVecIter{\anIter} \coloneq (\yVal{1}{\anIter}, \dots ,\yVal{n}{\anIter})^\top \in \xVectSpace^n$ for all $\anIter \in \{0,\dots, \iterCnt \}$, where $\yVal{\aParty}{\anIter}$ is defined for all $\aParty \in \partySet$ and $\anIter \in [0, \iterCnt]$ in equations \eqref{eq:iter.gauss.init} and \eqref{eq:iter.gauss.cancel}.  Recall that $\project{\aParty}$ is a column vector with a 1 in the $\aParty$th coordinate and $0$s elsewhere. Let $D^{(\anIter)}_c = \diag{\xCoeffOnl{1}{\anIter}, \dots,  \xCoeffOnl{n}{\anIter}} \in  \R^{n\times n}$ for all $\anIter \in [0,\iterCnt]$. The vectorized version of the \initPhase{} Phase (line  \ref{alg.line:ourprotocol.drop.init}) is 
\begin{equation}
	\yVecIter{0} \coloneq D^{(0)}_c (\vectPriv + \indNoiseVec) + \sum_{\aParty=1}^n \left( \matrixPartOnlEl{\aParty}{0}{:}  \noiseMatrixEl{\aParty}{:}^\top\right) \project{\aParty}
	\label{eq:init.vect} 
\end{equation}
and the vectorized version of the  \cancelPhase{} Phase (line \ref{alg.line:ourprotocol.drop.canceling}) is 
\begin{equation} 
	\yVecIter{\anIter} = 
	\transMatOnl{\anIter} \yVecIter{\anIter-1} + D^{(\anIter)}_c (\vectPriv + \indNoiseVec)  + \sum_{\aParty=1}^n \left(\matrixPartOnlEl{\aParty}{\anIter}{:}  \noiseMatrixEl{\aParty}{:}^\top \right) \project{\aParty}
	\label{eq:cancelIter.vect}
\end{equation}
for all  $\anIter \in [1, \iterCnt]$. 
Note that for all  $\anIter \in [0, \iterCnt]$ and $\aParty \in \partySet$, $\left(\matrixPartOnlEl{\aParty}{\anIter}{:}  \noiseMatrixEl{\aParty}{:}^\top \right) \project{\aParty}$ is a vector where only the $\aParty$-th value is different from 0.

Let $\nNoise \coloneq n\iterCnt$ be the total number of canceling noise terms and recall that   $\cancelNoiseVec \coloneq ( \noiseMatrixEl{1}{:},  \dots , \noiseMatrixEl{n}{:})^\top \in \xVectSpace^{\nNoise}$. 
For all $\anIter \in [0, \iterCnt]$, let
\[ \bigVecPart{\anIter} \coloneq \begin{ourmatrix}
	\project{1} \matrixPartOnlEl{1}{\anIter}{:}, &  \dots , &\project{n} \matrixPartOnlEl{n}{\anIter}{:} 
\end{ourmatrix} \ \in  \R^{n \times \nNoise}
\]
be  a block matrix formed by horizontally concatenating  $n$ blocks, each of size $n\times\iterCnt$ with only one nonzero row (i.e., the $\aParty$-th block only has the $\aParty$-th row with non-zero values). 
Equations \eqref{eq:init.vect} and \eqref{eq:cancelIter.vect} are equivalent to 
\begin{equation} 
	\yVecIter{0} =  D^{(0)}_c  (\vectPriv + \indNoiseVec)  +  \bigVecPart{0} \cancelNoiseVec
	\label{eq:init.vect.short} 
\end{equation} 
and
\begin{equation} 
	\yVecIter{\anIter} =  \transMatOnl{\anIter} \yVecIter{\anIter-1} +  D^{(\anIter)}_c  (\vectPriv + \indNoiseVec)  +  \bigVecPart{\anIter} \cancelNoiseVec
	\label{eq:cancelIter.vect.short} 
\end{equation} 
respectively.

\paragraph{Linear Equations} From Equation \eqref{eq:init.vect.short} and applying Equation \eqref{eq:cancelIter.vect.short} recursively 
we have that 
\begin{equation}
	\xMatIter{\anIter} (\vectPriv + \indNoiseVec)  + \nMatIter{\anIter} \cancelNoiseVec = \yVecIter{\anIter}
	\label{eq:cancelIter.recursive}
\end{equation}
for all $\anIter \in [0, \dots, \iterCnt]$ where  

\begin{align*}
	\xMatIter{0} &\coloneq  D^{(0)}_c , 
	&\xMatIter{\anIter} \coloneq \transMatOnl{\anIter} \xMatIter{\anIter-1} +  D^{(\anIter)}_c   \in \R^{n\times n}, \quad &\forall \anIter \in [\iterCnt] \\
	\nMatIter{0} &\coloneq \bigVecPart{0}, 
	&\nMatIter{\anIter} \coloneq \transMatOnl{\anIter} \nMatIter{\anIter-1} + \bigVecPart{\anIter}  \in \R^{n\times \nNoise},  \quad &\forall \anIter \in [\iterCnt].
\end{align*}
Let 
\[
\xMat \coloneq \begin{ourmatrix} 
	\xMatIter{0} \\
	\vdots\\
	\xMatIter{\iterCnt}
\end{ourmatrix} \in \R^{\nNoise \times n}
\text{, \quad}	
\nMat \coloneq \begin{ourmatrix} 
	\nMatIter{0} \\
	\vdots\\
	\nMatIter{\iterCnt}
\end{ourmatrix} \in \R^{\nNoise \times \nNoise}
\text{, and }
\yVec  \coloneq \begin{ourmatrix} 
	\yVecIter{0} \\
	\vdots\\
	\yVecIter{\iterCnt}
\end{ourmatrix} \in \R^{\nNoise}.
\]
The we can summarize Equation \eqref{eq:cancelIter.recursive} for all $t \in \{0, \dots, \iterCnt\}$ by 
\begin{equation}
	\xMat (\vectPriv + \indNoiseVec) + \nMat \cancelNoiseVec = \yVec
	\label{eq:linear.complete}
\end{equation}

\paragraph{Knowledge of the Adversary} 
Recall that $\corrnodes$ is the set of corrupted parties (if we are under the \colDPAcron{} setting), $\advValView$ is the set of observations made by the adversary,  $\partySetHonest$ is the set of honest parties with size $\nHonest$ and $\nNoiseHonest$ the total number of canceling noise terms generated by honest users. Additionally, recall that $\vectPrivHonest$, $\indNoiseVecH$ and  $\cancelNoiseVecHonest$ are the private vectors and noise terms of them. 
Let  $\vectPrivCorr$,  $\indNoiseVecCorr$ and $\cancelNoiseVecCorr$ be the vectors of values and noise terms generated by corrupted parties in $\corrnodes$ and
$[\aParty, \anIter]$ be the index of $\yVec$   such that $\yVec_{[\aParty,\anIter]} = \yVal{\aParty}{\anIter}$. Then,
\[ 
\xMatCorr^\advView (\vectPriv + \indNoiseVec) + \nMatCorr^\advView  \cancelNoiseVec = \yVecCorr' 
\]
is the linear system obtained by only keeping the rows with indexes $\{[\aParty,\anIter] : (\aParty,\anIter) \in \advValView \}$, which are observable by the adversary. Rearranging terms to separate the unknowns of the adversary, we generate the equivalent system
\[ 
\begin{ourmatrix} \xMatCorr & \xMatCorr^\corrSymb \end{ourmatrix}  \begin{ourmatrix}  \vectPrivHonest + \indNoiseVecH   \\  \vectPrivCorr + \indNoiseVecCorr \end{ourmatrix}    + \begin{ourmatrix}
	\nMatCorr  & \nMatCorr^\corrSymb  
\end{ourmatrix}  \begin{ourmatrix}  \cancelNoiseVecHonest \\ \cancelNoiseVecCorr \end{ourmatrix}   = \yVecCorr' . 
\]
Above, we  assume that rows which are linear combinations of the others are removed  such that $(\xMatCorr, \nMatCorr)$ is full rank.  Removing these rows does not reduce the  adversary's knowledge. 
The remaining system is also equivalent to
\[ 
\xMatCorr (\vectPrivHonest + \indNoiseVecH) +  \nMatCorr \cancelNoiseVecHonest = \yVecCorr' - \xMatCorr^\corrSymb (\vectPrivCorr + \indNoiseVecCorr) -  \nMatCorr^\corrSymb  \cancelNoiseVecCorr
\]
which is exactly the system we described in Equation \eqref{eq:linear.adv} for $\yVecCorr =  \yVecCorr' - \xMatCorr^\corrSymb (\vectPrivCorr + \indNoiseVecCorr) -  \nMatCorr^\corrSymb  \cancelNoiseVecCorr$, a vector known to the adversary.

% and  $\langle \aParty, \anIter \rangle$ the index of $\cancelNoiseVec$ such that $\cancelNoiseVecEl{\langle\aParty, \anIter\rangle} = \eta_{\aParty, \anIter}$. Let
%\[
%\strippedObs \coloneq \{[\aParty, \anIter] : (\aParty,\anIter) \in \advValView  \}
%\] 
%and 
%\[\strippedPHIt \coloneq \{ \langle \aParty, \anIter \rangle : \aParty, \in \partySetHonest \times [\iterCnt] \}.
%\] 
%We define 
%\[ 
%\xMatCorr \coloneq \xMat_{\strippedObs,\partySetHonest} \in \R^{\nObs \times \nHonest}
%\quad \text{ and } \quad 
%\nMatCorr \coloneq \nMat_{\strippedObs, \strippedPHIt} \in \R^{m \times\nNoiseHonest}.
%\] \cesarfoot{check that $\advValView$ only contains honest parties}
%Informally, $\xMatCorr$ and $\nMatCorr$ are constructed by keeping the rows of $\xMat$ and $\nMat$ respectively that correspond to elements of $\yVecCorr$ in Equation \eqref{eq:linear.complete}  
%and removing the columns of $\xMat$ and $\nMat$ respectively that correspond to coefficients of $(\valPriv{\aParty})_{\aParty \in \corrnodes}$ and $( \noiseMatrixEl{\aParty}{:})_{\aParty \in \corrnodes}$ (respectively) in Equation \eqref{eq:linear.complete}. Then, the knowledge gathered by the him is at most the set of linear equations described in Equation \eqref{eq:linear.adv}.

%where $(\xMatCorr, \nMatCorr, \yVecCorr)$ are the coefficients (i.e., are known by the adversary) and $(\vectPrivHonest, \cancelNoiseVecHonest)$ are the unknowns. 

\subsection{Preliminary Lemma}

Before presenting our privacy proofs, we prove a preliminary lemma.

\begin{lemma}  
	\label{lm:cDivEpsilon}
	Given $\epsilon, \delta \in (0,1)$, Equations 
	\begin{equation} 
		\epsilon -  \frac{1}{2}  \theta  \ge \sqrt{ \theta }
		\label{eq:cDivEpsilon.1}
	\end{equation} 
	and
	\begin{equation}
		\frac{1}{2} \frac{(\epsilon -  \frac{1}{2}   \theta) ^2}{ \theta} \ge  \ln\left(\frac{2}{\delta\sqrt{2\pi}}\right) 
		\label{eq:cDivEpsilon.2}
	\end{equation}  
	are satisfied for 
	\begin{equation}
		\theta  \le  \frac{\epsilon^2}{c^2}
	\end{equation}
	where $c^2 > 2\ln (1.25/\delta)$.

\end{lemma} 
\begin{proof}
	This lemma has been proven as part of the proof of bounds of the Gaussian Mechanism  \cite[Theorem A.1]{dwork2014algorithmic}. For $\theta  \le  \frac{\epsilon^2}{c^2}$,  we have that Equation \eqref{eq:cDivEpsilon.1} is implied if 
	\begin{align*} 
		&\epsilon -  \frac{\epsilon^2}{2c^2}   > \frac{\epsilon}{c} \\
		\iff &\epsilon > \frac{\epsilon}{c} + \frac{\epsilon^2}{2c^2} \\
		\iff & \frac{c}{\epsilon} \epsilon > \frac{c}{\epsilon} \left(\frac{\epsilon}{c} + \frac{\epsilon^2}{2c^2} \right)  \\
		\iff & c - \frac{\epsilon}{2c} > 1.
	\end{align*} 
	For $\epsilon \le 1$, the above is implied when $c > 3/2$. Moreover, for  $\theta  \le  \frac{\epsilon^2}{c^2}$ Equation\eqref{eq:cDivEpsilon.2} is implied if
	\begin{align*}
		&\frac{1}{2} \frac{\left(\epsilon -  \frac{\epsilon^2}{2c^2}   \right)^2}{ \frac{\epsilon^2}{ c^2} } \ge  \ln\left(\frac{2}{\delta\sqrt{2\pi}}\right) \\
		\iff &  \frac{1}{2} \left(\epsilon -  \frac{\epsilon^2}{2c^2}    \right)^2  \ge  \ln\left(\frac{2}{\delta\sqrt{2\pi}}\right)  \frac{\epsilon^2}{ c^2}\\
		\iff &  \frac{1}{2}\left(  \epsilon^2 -  \frac{\epsilon^3}{c^2} +  \frac{\epsilon^4}{4c^4} \right)    \ge  \ln\left(\frac{2}{\delta\sqrt{2\pi}}\right)  \frac{\epsilon^2}{ c^2} \\
		\iff &   c^2 -  \epsilon +  \frac{\epsilon^2}{4c^2}     \ge  2 \ln\left(\frac{2}{\delta\sqrt{2\pi}}\right). 
	\end{align*}  
	For $c > 3/2$ and $\epsilon \le 1$, the derivative of $c^2 -  \epsilon +  \frac{\epsilon^2}{4c^2}$ is positive. Therefore  
	\[ 
	c^2 -  \epsilon +  \frac{\epsilon^2}{4c^2} > c^2 - 8/9
	\] and thus Equation \eqref{eq:cDivEpsilon.2} is satisfied if $c^2 > 2 \ln(1.25 / \delta)$. 
\end{proof}

\subsection{Proof of Theorem \ref{thm:dp.abstract} and Corollary \ref{crl:sdp}} 
\label{app:privacy.abstract}

Now we prove Theorem \ref{thm:dp.abstract}. 
\\
Theorem \ref{thm:dp.abstract}. \textit{\abstractThmStatement}
\begin{proof}
	We start our analysis from Equation \eqref{eq:linear.complete}. 
	Matrices $\nMatCorr$, $\xMatCorr$ and $\vectPrivHonest$ are fixed, then honest parties draw $\indNoiseVecH$, $\cancelNoiseVecHonest$ and the adversary observes $\yVecCorr$. Let $\neighDatasetA$ and $\neighDatasetB$ be two possible values of $\vectPrivHonest$ that are neighboring as precised in Definition \ref{def:dp}. We prove that \[
	\Pr\left(\yVecCorr \mid \neighDatasetA\right) \le \epsilon \Pr\left(\yVecCorr \mid \neighDatasetB\right) + \delta
	\] for $\epsilon$ and $\delta$ according to Equation \eqref{eq:dp.abstract}. Recalling that $\indicator{\cdot}$ is the indicator function, we have that 
	{\small
		\begin{align*}
			&\Pr\left(\yVecCorr \mid \vectPrivHonest\right) \\
			&\quad=\int_{\noiseVec \in \xVectSpace^{\nHonest+\nNoiseHonest}} \indicator{\xMatCorr (\vectPrivHonest+ \indNoiseVecH)+  \nMatCorr \cancelNoiseVecHonest = \yVecCorr } \Pr\noiseVec  \text{d}\noiseVec 
		\end{align*} 
	}%
	
	As $(\xMatCorr, \nMatCorr)$ is full rank, then $(\xMatCorr, \nMatCorr)\noiseVec$ covers the complete space $\xVectSpace^{\nObs}$. Recall that $\varMatProd$ is the covaraince matrix of $(\xMatCorr, \nMatCorr)\noiseVec$. The probability of a certain value of $(\xMatCorr, \nMatCorr)\noiseVec = \yVecCorr - \xMatCorr \vectPrivHonest$ is given by the Gaussian distribution
	\[
	\Pr\left((\xMatCorr, \nMatCorr)\noiseVec = \dummyVec \right) =  \frac{\exp\left(-\dummyVec^\top \varMatProd^{-1} \dummyVec /2\right)}{\sqrt{(2\pi)^{\nObs} \det\left( \varMatProd\right)}}
	\]
	For two possible values $\dummyVecA = \yVecCorr-\xMatCorr \neighDatasetA$ and  $\dummyVecB = \yVecCorr-\xMatCorr \neighDatasetB$ we get the ratio:
	\begin{align*} 
		\frac{\Pr\left((\xMatCorr, \nMatCorr)\noiseVec = \dummyVecA \right)}{\Pr\left((\xMatCorr, \nMatCorr)\noiseVec = \dummyVecB \right)} &=
		\frac{\exp\left(-\dummyVecA^\top \varMatProd^{-1} \dummyVecA /2 \right)}{ \exp\left(-\dummyVecB^\top \varMatProd^{-1} \dummyVecB /2 \right)} \\
		&= 
		\exp \left( -(\dummyVecA+\dummyVecB)^\top \varMatProd^{-1} (\dummyVecA-\dummyVecB) / 2 \right). 
	\end{align*}

	Now we adapt the strategy of \cite{Dwork2014} for proving the general gaussian mechanism, computing first this ratio and showing it is bounded by $e^\epsilon$ with probability $1 - \delta/2$ (and above $e^{-\epsilon}$ with probability $1-\delta/2$). We have that 
	\begin{align}  
		\ln &\left|\frac{\Pr\left((\xMatCorr, \nMatCorr)\noiseVec = \dummyVecA \right)}{\Pr\left((\xMatCorr, \nMatCorr)\noiseVec = \dummyVecB \right)} \right| >\epsilon  \notag \\
		\iff &\left| -\frac{1}{2}(\dummyVecA+\dummyVecB)^\top \varMatProd^{-1} (\dummyVecA-\dummyVecB)  \right| > \epsilon.
		\label{eq:dp.abst.ratio}
	\end{align} 
	Now note that 
	\begin{align*}
		\dummyVecA-\dummyVecB &=\xMatCorr(\neighDatasetA-\neighDatasetB) \\ 
		\dummyVecA+\dummyVecB &= 2\yVecCorr -\xMatCorr(\neighDatasetA+\neighDatasetB).
	\end{align*}  
	Without loss of generality, we assume that our datasets $\neighDatasetA$ and $\neighDatasetB$ differ on the value of one party as much as possible. Let $\cancelNoiseDiff \vectPriv \coloneq \neighDatasetA - \neighDatasetB$. As private values lie in the interval $[0,1]$, we have that $\cancelNoiseDiff \vectPriv$ has one component equal to $1$ and  $0$ in all the other components. Equation \eqref{eq:dp.abst.ratio} is equivalent to 
	\begin{align*} 
		& \left| \frac{1}{2}\left(2\yVecCorr -\xMatCorr(\neighDatasetA+\neighDatasetB)\right)^\top \varMatProd^{-1}\xMatCorr \cancelNoiseDiff \vectPriv \right| >\epsilon  \\
		\iff & \left| \frac{1}{2}\left(2\yVecCorr - 2\xMatCorr \neighDatasetA  +\xMatCorr\cancelNoiseDiff \vectPriv\right)^\top \varMatProd^{-1}\xMatCorr \cancelNoiseDiff \vectPriv \right| >\epsilon \\
		\iff & \left| \frac{1}{2}\left(2\left(\xMatCorr, \nMatCorr\right)\noiseVec  +\xMatCorr\cancelNoiseDiff \vectPriv\right)^\top \varMatProd^{-1}\xMatCorr \cancelNoiseDiff \vectPriv \right|>\epsilon \\
		\iff & \left| \left(\left(\xMatCorr, \nMatCorr\right)\noiseVec\right)^\top \varMatProd^{-1}\xMatCorr \cancelNoiseDiff \vectPriv  + \frac{1}{2}\ (\xMatCorr\cancelNoiseDiff \vectPriv)^\top \varMatProd^{-1}\xMatCorr \cancelNoiseDiff \vectPriv \right|   > \epsilon.
	\end{align*} 
	We explore the conditions where 
	\begin{align*} 
		\left(\left(\xMatCorr, \nMatCorr\right)\noiseVec\right)^\top \varMatProd^{-1}\xMatCorr \cancelNoiseDiff \vectPriv + \frac{1}{2}\ (\xMatCorr\cancelNoiseDiff \vectPriv)^\top \varMatProd^{-1}\xMatCorr \cancelNoiseDiff \vectPriv > \epsilon 
	\end{align*}
	with probability smaller than $\delta/2$ (as said, the other side of the bound is analog). Note that  
	\[
	\xMatCorrCol \coloneq \xMatCorr \cancelNoiseDiff \vectPriv \in \R^{\nObs}
	\]  
	is the $\aParty$th column of $\xMatCorr$, where $\aParty$ is the coordinate of $\cancelNoiseDiff \vectPriv$ that is equal to $1$. The above is equivalent to 
	\begin{align*}
		\left(\left(\xMatCorr, \nMatCorr\right)\noiseVec\right)^\top \varMatProd^{-1} \xMatCorrCol >  \epsilon -  \frac{1}{2} \xMatCorrCol^\top \varMatProd^{-1}\xMatCorrCol .
	\end{align*}
	To bound the probability of the above to hold, we will use the tail bound 
	\begin{equation} 
		\Pr\left(\dummyTailBoundA > \dummyTailBoundB\right)  \le \frac{\sdTailBoundA}{\dummyTailBoundB \sqrt{2\pi}} \exp(-\dummyTailBoundB^2/2\sdTailBoundA^2)
		\label{eq:tailbound}  
	\end{equation}
	where $\sdTailBoundA$ is the standard deviation of $\dummyTailBoundA$. We set 
	\begin{align} 
		\dummyTailBoundA &= \left(\left(\xMatCorr, \nMatCorr\right)\noiseVec\right)^\top \varMatProd^{-1} \xMatCorrCol  \label{eq:tailbound.varA} \\
		\dummyTailBoundB &= \epsilon -  \frac{1}{2}  \xMatCorrCol^\top \varMatProd^{-1} \xMatCorrCol. \label{eq:tailbound.varB} 
	\end{align}
	and start by computing $\sdTailBoundA$:
	\begin{align} 
		\sdTailBoundA^2  &= \text{var}\left( \left(\left(\xMatCorr, \nMatCorr\right)\noiseVec\right)^\top \varMatProd^{-1} \xMatCorrCol \right) \notag \\
		%	&= \text{var}\left ( \left( \varMatProd^{-1}\xMatCorr \cancelNoiseDiff \vectPriv\right)^\top \left(\xMatCorr, \nMatCorr\right)\noiseVec \right) \\
		% &= \text{var}\left (  (\cancelNoiseDiff \vectPriv\xMatCorr)^\top \varMatProd^{-1}  \left(\xMatCorr, \nMatCorr\right)\noiseVec \right) \\
		&=   \xMatCorrCol^\top \varMatProd^{-1}  \text{var}\left(\left(\xMatCorr, \nMatCorr\right)\noiseVec\right)     \varMatProd^{-1}\xMatCorrCol  \notag \\
		&=  \xMatCorrCol^\top \varMatProd^{-1} \xMatCorrCol. 
		\label{eq:tailbound.sd} 
	\end{align} 
	By Equation \eqref{eq:tailbound}, proving that 
	\[ 
	\frac{\sdTailBoundA}{\dummyTailBoundB \sqrt{2\pi}} \exp\left(\frac{-\dummyTailBoundB^2}{2\sdTailBoundA^2}\right) \le \frac{\delta}{2}
	\]
	implies our privacy guarantee. The above is equivalent to 
	\begin{align*} 
		&\frac{\dummyTailBoundB}{\sdTailBoundA} \exp\left(\frac{\dummyTailBoundB^2}{2\sdTailBoundA^2}\right) \ge \frac{2}{\delta\sqrt{2\pi}}  
	\end{align*} 
	and, applying logarithms to both sides, to
	\begin{align*} 
		\ln \left( \frac{\dummyTailBoundB}{\sdTailBoundA} \right) + \frac{1}{2}\left(\frac{\dummyTailBoundB}{\sdTailBoundA}\right)^2 \ge \ln\left(\frac{2}{\delta\sqrt{2\pi}}\right).
	\end{align*} 
	The above is implied if 
	\begin{equation}  
		\ln \left( \frac{\dummyTailBoundB}{\sdTailBoundA} \right) \ge 0 
		\label{eq:dp.cond.1} 
	\end{equation}
	and
	\begin{equation}  
		\frac{1}{2}\left(\frac{\dummyTailBoundB}{\sdTailBoundA}\right)^2 \ge \ln\left(\frac{2}{\delta\sqrt{2\pi}}\right).	
		\label{eq:dp.cond.2} 	
	\end{equation}
	Below, we replace $\dummyTailBoundB$ and $\sdTailBoundA$ using Equations \eqref{eq:tailbound.varB} and \eqref{eq:tailbound.sd}. Then by noticing that Equation \eqref{eq:dp.cond.1} is equivalent to $\dummyTailBoundB \ge \sdTailBoundA$, we have that Equations \eqref{eq:dp.cond.1} and \eqref{eq:dp.cond.2} can be rewritten as  
	\begin{equation} 
		\epsilon -  \frac{1}{2}  \xMatCorrCol^\top \varMatProd^{-1} \xMatCorrCol  > \sqrt{ \xMatCorrCol^\top \varMatProd^{-1} \xMatCorrCol }
		\label{eq:dp.bound.1}
	\end{equation} 
	and
	\begin{equation}
		\frac{1}{2}\frac{ \left(  \epsilon -  \frac{1}{2}   \xMatCorrCol^\top \varMatProd^{-1} \xMatCorrCol   \right)^2}{ \xMatCorrCol^\top \varMatProd^{-1} \xMatCorrCol } \ge  \ln\left(\frac{2}{\delta\sqrt{2\pi}}\right) 
		\label{eq:dp.bound.2}
	\end{equation}  
	respectively. Finally, using Lemma \ref{lm:cDivEpsilon}, equations \eqref{eq:dp.bound.1} and \eqref{eq:dp.bound.2} are implied for 
	\[ 
	\xMatCorrCol^\top  \varMatProd^{-1} \xMatCorrCol \le  \frac{\epsilon^2}{c^2}. 
	\] 
	where $c^2 > 2 \ln(1.25 / \delta)$. 
\end{proof}

Now we prove Corollary \ref{crl:sdp}.

Corollary \ref{crl:sdp}.\textit{\crlConvexStm}
\begin{proof}
	Given that $(\xMatCorr, \nMatCorr)$ is full rank and $\varNoise$ is a positive diagonal matrix, then $\varMatProd \succ 0$. Given that 
	\[ 
	\begin{ourmatrix}
		\varMatProd & \xMatCorrCol \\
		\xMatCorrCol^\top & \epsilon^2/ c^2 
	\end{ourmatrix} \succeq 0,
	\]  by Schur's complement we have that 
	$\epsilon^2/ c^2  - \xMatCorrCol^\top 	\varMatProd^{-1} \xMatCorrCol \ge 0$ for each column $\xMatCorrCol$ of $\xMatCorr$. This is equivalent to the condition of Theorem \ref{thm:dp.abstract} to obtain $(\epsilon, \delta)$-DP. 
\end{proof}

\subsection{Proof of Lemma \ref{lm:dp.noiseDiff.abstract}}
\label{app:privacy.noiseDiff}

\emph{Lemma} \ref{lm:dp.noiseDiff.abstract}. \textit{\lemmaNoiseDiffStm} 

\begin{proof}
	It follows directly from  \cite[Theorem 1]{sabater2022accurate}. We adapt the proof to our setting.
	Let $\neighDatasetA$ and $\neighDatasetB$ be two neighboring datasets. Let $\yVecCorr$ be the view of the adversary after execution $\exec$.  Let 
	\[ S^{\adjASup} = \{ (\indNoiseVecH, \cancelNoiseVecHonest ) \in \xVectSpace^{\nHonest + \nNoiseHonest} : \advViewExec{\exec}{\neighDatasetA + \indNoiseVecH, \cancelNoiseVecHonest} =\yVecCorr \} 
	\]
	and analogously
	\[ S^{\adjBSup} = \{ (\indNoiseVecH, \cancelNoiseVecHonest ) \in \xVectSpace^{\nHonest + \nNoiseHonest} : \advViewExec{\exec}{\neighDatasetB + \indNoiseVecH, \cancelNoiseVecHonest} =\yVecCorr \}. 
	\]
	Let $t=(\indNoiseDiff, \cancelNoiseDiffVec)$. The precondition of the lemma implies that $S^{\adjASup} + t = S^{\adjBSup}$. Our $(\epsilon, \delta)$-differential privacy guarantee given by 
	\[ 
	\Pr( \yVecCorr | \neighDatasetA ) \le e^\epsilon  \Pr(\yVecCorr | \neighDatasetB) + \delta
	\]
	is implied if 
	\begin{equation}
		\Pr( (\indNoiseVecH, \cancelNoiseVecHonest ) ) \le e^\epsilon \Pr((\indNoiseVecH, \cancelNoiseVecHonest )+t ) + \delta
		\label{eq:diffSuffCond}
	\end{equation}
	Indeed if Equation \eqref{eq:diffSuffCond} holds we have 
	\newcommand{\diff}[1]{\mathrm{d}#1}
	\begin{align*}
		\Pr( \yVecCorr | \neighDatasetA )  &= \int_{(\indNoiseVecH, \cancelNoiseVecHonest) \in S^{\adjASup}} \Pr((\indNoiseVecH, \cancelNoiseVecHonest )) \diff{\indNoiseVecH} \diff{\cancelNoiseVecHonest} \\
		&\le  \int_{(\indNoiseVecH, \cancelNoiseVecHonest) \in S^{\adjASup}}  
		(e^\epsilon \Pr((\indNoiseVecH, \cancelNoiseVecHonest)+t) + \delta) \diff{\indNoiseVecH} \diff{\cancelNoiseVecHonest} \\
		&=  \int_{(\indNoiseVecH, \cancelNoiseVecHonest)-t \in S^{\adjASup}}  
		(e^\epsilon \Pr((\indNoiseVecH, \cancelNoiseVecHonest)) + \delta) \diff{\indNoiseVecH} \diff{\cancelNoiseVecHonest} \\
		&= \int_{(\indNoiseVecH, \cancelNoiseVecHonest) \in S^{\adjBSup}}  
		(e^\epsilon \Pr((\indNoiseVecH, \cancelNoiseVecHonest)) + \delta) \diff{\indNoiseVecH} \diff{\cancelNoiseVecHonest} \\
		&= e^\epsilon \Pr( \yVecCorr | \neighDatasetB ) + \delta.
	\end{align*}
	Therefore, it suffices to prove Equation \eqref{eq:diffSuffCond} to prove our $(\epsilon, \delta)$-DP guarantee.

	We will do that by proving that  $\Pr((\indNoiseVecH, \cancelNoiseVecHonest))\le e^\varepsilon P((\indNoiseVecH, \cancelNoiseVecHonest)+t)$ with probability at least $1-\delta$.
	\newcommand{\etadelta}{\gamma}
	Denoting $\etadelta=(\indNoiseVecH, \cancelNoiseVecHonest)$  for convenience, we need to prove that with
	probability $1-\delta$ it holds that 
	\[|\log(P(\etadelta)/P(\etadelta+t))| \le
	\varepsilon.
	\]
	We have that 
	\begin{eqnarray*}
		\Big|\log \frac{P(\etadelta)}{P(\etadelta+t)}\Big|
		&=& \Big|-\frac{1}{2}\etadelta^\top \varNoise^{-1} \etadelta + \frac{1}{2}
		(\etadelta+t)^\top \varNoise^{-1}(\etadelta+t)\Big|\\
		&=& \Big|\frac{1}{2}(2\etadelta+t)^\top\varNoise^{-1}  t\Big|.
	\end{eqnarray*}
	To ensure that $|\log(P(\etadelta)/P(\etadelta+t))| \le
	e^\varepsilon$ holds with probability at least $1-\delta$, since we are interested in the absolute value, we show that
	\[
	P\Big(\frac{1}{2}(2\etadelta+t)^\top\varNoise^{-1} t \ge \varepsilon\Big) \le
	\delta/2,
	\]
	the proof of the other direction is analogous.
	This is equivalent to
	\begin{equation}
		\label{eq:etadeltabound1}
		P(\etadelta\varNoise^{-1} t \ge \varepsilon - t^\top\varNoise^{-1} t/2) \le \delta/2.
	\end{equation}
	The variance of $\etadelta\varNoise^{-1} t$ is
	\begin{eqnarray*}
		\text{var}(\etadelta\varNoise^{-1} t) 
		&=&
		t^\top \varNoise^{-\top} \text{var}(\etadelta) \varNoise^{-1} t
		\\ 
		&=& t^\top \varNoise^{-\top} \varNoise \varNoise^{-1} t\\
		&=& t^\top \varNoise^{-1}  t.
	\end{eqnarray*}
	For any centered Gaussian random variable $Y$ with variance
	$\sigma_Y^2$, we have that
	\begin{equation}
		P(Y\ge \lambda) \le \frac{\sigma_Y}{\lambda\sqrt{2\pi}}\exp\left(-\lambda^2/2\sigma_Y^2\right).
		\label{eq:guassTailBound}
	\end{equation}
	Let $Y=\etadelta\varNoise^{-1} t$, $\sigma_Y^2=t^\top \varNoise^{-1} t$ and $\lambda=\varepsilon - t^\top\varNoise^{-1} t/2$, then satisfying
	\begin{equation}
		\frac{\sigma_Y}{\lambda\sqrt{2\pi}}\exp\left(-\lambda^2/2\sigma_Y^2\right)\le \delta/2
		\label{eq:gaussTailLeDelta}
	\end{equation}
	implies Equation \eqref{eq:etadeltabound1}.
	Equation \eqref{eq:gaussTailLeDelta} is equivalent to
	\[
	\frac{\lambda}{\sigma_Y}\exp\left(\lambda^2/2\sigma_Y^2\right)\ge
	2/\delta\sqrt{2\pi},
	\]
	or, after taking logarithms on both sides, to
	\[
	\log\left(\frac{\lambda}{\sigma_Y}\right)+\frac{1}{2}\left(\frac{\lambda}{\sigma_Y}\right)^2\ge
	\log\left(\frac{2}{\delta\sqrt{2\pi}}\right).
	\]
	To make this inequality hold, we require
	\begin{equation}
		\log\left(\frac{\lambda}{\sigma_Y}\right)\ge 0
		\label{eq:etadeltabound2}
	\end{equation}
	and 
	\begin{equation}
		\frac{1}{2}\left(\frac{\lambda}{\sigma_Y}\right)^2\ge
		\log\left(\frac{2}{\delta\sqrt{2\pi}}\right).
		\label{eq:etadeltabound3}
	\end{equation}
	Equation \eqref{eq:etadeltabound2} is equivalent to
	$
	\lambda\ge \sigma_Y
	$. 
	Substituting $\lambda$ and $\sigma_Y$ we get
	\begin{equation}
		\varepsilon - t^\top\varNoise^{-1} t/2 \ge 
		(t^\top \varNoise^{-1} t)^{1/2}.
		\label{eq:varDiff.bound.1}
	\end{equation}
	Substituting $\lambda$ and $\sigma_Y$ in Equation \eqref{eq:etadeltabound3}
	gives  
	\begin{equation}
		\frac{1}{2} \frac{(\epsilon -  \frac{1}{2}   t^\top \varNoise^{-1} t) ^2}{t^\top \varNoise^{-1} t} \ge  \ln\left(\frac{2}{\delta\sqrt{2\pi}}\right) 
		\label{eq:varDiff.bound.2} 
	\end{equation} 
	By Lemma \ref{lm:cDivEpsilon}, equations \eqref{eq:varDiff.bound.1} and \eqref{eq:varDiff.bound.2} are satisfied for
	$t^\top \varNoise^{-1} t \le \frac{\epsilon^2}{c^2}$ for $c^2 = 2\ln(1.25/\delta)$.

\end{proof}

\subsection{Proof of Lemma \ref{lm:edgeVec}}
\label{app:privacy.indist}

Lemma \ref{lm:edgeVec}. \textit{\stmIndLemma}
\begin{proof}
	For clarity, within the proof we ignore the superscript $\honestSymb$ of $\vectPrivHonest$, $\indNoiseVecH$ and $\cancelNoiseMatHonest$ and refer to them as $\vectPriv$, $\indNoiseVec$ and $\cancelNoiseMat$.

	Let $\vectPrivNoisy = \vectPriv + \indNoiseVec$. We want to prove that for any  $\beta \in \R$, 
	\begin{equation}
		\advViewExec{\exec}{\vectPrivNoisy, \cancelNoiseMat} = \advViewExec{\exec}{\vectPrivNoisy', \cancelNoiseMat'}
		\label{eq:indisinguish} 
	\end{equation}
	where $\vectPrivNoisy' = \vectPrivNoisy + \beta \edgeVec{\aParty}{\anIter}$ and 
	$\cancelNoiseMat' = \cancelNoiseMat - \beta \etaTransEl{\aParty}{\anIter}$ for some $\edgeVec{\aParty}{\anIter}$ and $\etaTransEl{\aParty}{\anIter}$. To achieve this, we will set $\valPrivNoisy{\aParty}' = \valPrivNoisy{\aParty} +  \beta$  and compute the   values of  $\vectPrivNoisy'$ and $\cancelNoiseMat'$ such that the view $\advView$ remains unchanged. 
	
	Let
	$\aggParty{\yVec} \in \xVectSpace^{(\iterCnt+1)\times n}$ be 
	\[ 
	\aggPartyEl{\anIter'}{:}{\yVec} \coloneq \left\{    
	\begin{array}{ll} 	
		0 &\text{for $\anIter' =0$} \\
		\left( \transMatOnl{\anIter'} \yVecIter{\anIter'-1}\right)^\top &\text{for $\anIter' \in [1, \iterCnt].$}  
	\end{array}
	\right.
	\]  
	\newcommand{\onlIter}[1]{\onlSymb_{#1}}
	The vector $\yVal{\partyB}{:}\in \xVectSpace^{\iterCnt+1}$ of messages of any party $\partyB \in \partySet$ follows 
	\begin{equation} 
		\aggPartyEl{:}{\partyB}{\yVec} + \xCoeffOnlVec{\partyB} \valPrivNoisy{\partyB} + \matrixPartOnl{\partyB}  \noiseMatrixEl{\partyB}{:}^\top  =  \yVal{\partyB}{:}.
		\label{eq:iterPerParty}
	\end{equation} 
	Let $\yVec'  = (\yValPrime{\partyB}{\anIter'})^\top_{\aParty \in \partySet, \anIter' \in [0,\iterCnt]}$ be the vector of messages of all parties with input $\vectPrivNoisy', \cancelNoiseMat'$. 
	For all $\partyB \in \partySet$ let
	\[ 
	\onlIter{\partyB} = \{\anIter' \in [\iterCnt] : \partyB \in \onlNodes{\anIter'} \}
	\]
	be the set of active iterations of a party.
	We are assuming that the only message that is not visible to the adversary  is $\yVal{\aParty}{\anIter}$. Therefore Equation \eqref{eq:indisinguish} is implied if 
	\begin{equation}
		\yValPrime{\aParty}{\anIter'} = \yVal{\aParty}{\anIter'} \quad \text{for all $\anIter' \in \onlIter{\aParty} \setminus \{ \anIter\}$}
		\label{eq:viewConstraint.1} 	\end{equation}
	and  
	\begin{equation} 
		\yValPrime{\partyB}{\anIter'} = \yVal{\partyB}{\anIter'}, \quad 	\text{for all  $\partyB \in \partySet \setminus  \{\aParty\},\anIter' \in \onlIter{\partyB}$}, 	
		\label{eq:viewConstraint.2}
	\end{equation}
	which means that all messages remain unchanged except for $\yVal{\aParty}{\anIter}$ when the input changes from $\vectPrivNoisy, \cancelNoiseMat$ to $\vectPrivNoisy', \cancelNoiseMat'$.

	\paragraph{Temporary Dropouts} We first analyze the case where there are no permanent dropouts.
	We will start by computing $\yValPrime{\aParty}{\anIter} - \yVal{\aParty}{\anIter}$, which we will use for the remainder of the proof. By Equation \eqref{eq:viewConstraint.1},   have that \begin{equation} 
		\sum_{\anIter' \in \onlIter{\aParty}} \yValPrime{\aParty}{\anIter'} - \sum_{\anIter' \in \onlIter{\aParty}} \yVal{\aParty}{\anIter'} =  \yValPrime{\aParty}{\anIter} - \yVal{\aParty}{\anIter}
		.
		\label{eq:cdp.yminus.1} 
	\end{equation}  
	By equations \eqref{eq:iter.gauss.init} and \eqref{eq:iter.gauss.cancel} we have that for all $\partyB \in \partySet$ 
	\begin{align} 
		\sum_{\anIter \in  \onlIter{\partyB}} \yVal{\partyB}{\anIter'} &= \sum_{\anIter' \in \onlIter{\partyB} \setminus \{0\} }  \transMatOnlEl{\anIter'}{\partyB}{:} \yVecIter{\anIter'-1} +  \sum_{\anIter' \in \onlIter{\partyB}}  \xCoeffOnl{\partyB}{\anIter'} \valPrivNoisy{\partyB} + \matrixPartOnlEl{\partyB}{\anIter'}{:} (\cancelNoise{\partyB}{:})^\top.  \label{eq:cdp.yminus.x}
	\end{align}  
	By definition of $\partDistrOnlSymb$ and Equation \eqref{eq:canceling.drop},  we know that 	 
	\begin{align*}
		&\sum_{\anIter' \in [0,T]}\xCoeffOnl{\partyB}{\anIter'} \valPrivNoisy{\partyB} + \matrixPartOnlEl{\partyB}{\anIter'}{:} (\cancelNoise{\partyB}{:})^\top \notag \\
		&=\sum_{\anIter' \in \onlIter{\partyB}} \xCoeffOnl{\partyB}{\anIter'} \valPrivNoisy{\partyB} + \matrixPartOnlEl{\partyB}{\anIter'}{:} (\cancelNoise{\partyB}{:})^\top =  \totalWeightParty{\partyB} \valPrivNoisy{\partyB} 
	\end{align*}
	when there are only temporary drop-outs. Then  Equation \eqref{eq:cdp.yminus.x} becomes equivalent to 
	\begin{align} 
		\sum_{\anIter \in  \onlIter{\partyB}} \yVal{\partyB}{\anIter'} &= \left( \sum_{\anIter' \in \onlIter{\partyB} \setminus \{0\} }  \transMatOnlEl{\anIter'}{\partyB}{:} \yVecIter{\anIter'-1} \right) + \totalWeightParty{\partyB} \valPrivNoisy{\partyB} 	\label{eq:cdp.yminus.2} 
	\end{align}
	for all $\partyB \in \partySet$. 
	Then by equations \eqref{eq:cdp.yminus.2}, \eqref{eq:viewConstraint.1}   and \eqref{eq:viewConstraint.2} we also have that 
	\begin{equation} 
		\sum_{\anIter' \in \onlIter{\partyB}} \yValPrime{\partyB}{\anIter'} - \sum_{\anIter' \in \onlIter{\partyB}} \yVal{\partyB}{\anIter'} =  \transMatOnlEl{\anIter+1}{\partyB}{\aParty}( \yValPrime{\aParty}{\anIter} - \yVal{\aParty}{\anIter})   +  \totalWeightParty{\partyB} (\valPrivNoisy{\partyB}' - \valPrivNoisy{\partyB})
		\label{eq:cdp.yminus.3} 
	\end{equation}  
	for all $\partyB \in \partySet$.
	By combining the above with Equation \eqref{eq:cdp.yminus.1} and since $\beta =\valPrivNoisy{\aParty}' - \valPrivNoisy{\aParty}$   have that  
	\[
	\yValPrime{\aParty}{\anIter} - \yVal{\aParty}{\anIter} = 
	\transMatOnlEl{\anIter+1}{\aParty}{\aParty}( \yValPrime{\aParty}{\anIter} - \yVal{\aParty}{\anIter})   +  \totalWeightParty{\aParty} \beta
	\]
	The above is equivalent to 
	\begin{equation}
		\yValPrime{\aParty}{\anIter} - \yVal{\aParty}{\anIter} = \frac{  \totalWeightParty{\aParty} \beta} { 1- \transMatOnlEl{\anIter+1}{\aParty}{\aParty}}
		\label{eq:msgdiff}
	\end{equation}
	which is the quantity we wanted to compute as a first step of our proof.
	
	Now we will compute $\cancelNoise{\aParty}{:}' - \cancelNoise{\aParty}{:}$. Using  equations \eqref{eq:iterPerParty} and \eqref{eq:viewConstraint.1}, we have that for all $\anIter' \in \onlIter{\aParty} \setminus \{\anIter\}$: 
	\begin{align}
		0 &= 	\yValPrime{\aParty}{\anIter'} - \yVal{\aParty}{\anIter'} \notag   \\
		&=     \transMatOnlEl{\anIter'}{\aParty}{:} ( \yVecIterPrime{\anIter'-1} -  \yVecIter{\anIter'-1}) + \xCoeffOnl{\aParty}{\anIter'} (\valPrivNoisy{\aParty}' - \valPrivNoisy{\aParty})+  \matrixPartOnlEl{\aParty}{\anIter'}{ :}  (\noiseMatrixEl{\aParty}{:}' - \noiseMatrixEl{\aParty}{:})^\top \label{eq:yminus} 
	\end{align}  
	The above is equivalent to 
	\begin{equation}
		\matrixPartOnlEl{\aParty}{\anIter'}{ :}  (\noiseMatrixEl{\aParty}{:}' - \noiseMatrixEl{\aParty}{:})^\top  = -  \xCoeffOnlVec{\aParty} \beta -   \transMatOnlEl{\anIter'}{\aParty}{:} ( \yVecIterPrime{\anIter'-1} -  \yVecIter{\anIter'-1})
		\label{eq:zi}
	\end{equation}
	for all $\anIter' \in \onlIter{\aParty} \setminus \{\anIter\}$. 
	
	Let $\hat{\anIter}$ be the smallest iteration of $\onlIter{\aParty}$ that is greater than $\anIter$ (it always exists, as $\anIter < \iterCnt$ and $\iterCnt \in \onlIter{\aParty}$).  We have that  \[\transMatOnlEl{\anIter'}{\aParty}{:} ( \yVecIterPrime{\anIter'-1} -  \yVecIter{\anIter'-1}) = 0
	\] for all $\anIter' \in \onlIter{\aParty} \setminus \{\anIter, \hat{\anIter}\}$. 
	If $\aParty$ dropped out from iterations $\anIter+1$ to $\hat{\anIter}$, $\yVecIter{\anIter'} = \yVecIter{\anIter+1}$ for all $\anIter' \in [\anIter+1, \hat{\anIter}]$, and no extra noise terms and private values were added due to inactivity. Therefore, by Equation \eqref{eq:msgdiff}  we have 
	\begin{align*}  
		\transMatOnlEl{\hat{\anIter}}{\aParty}{:} (\yVecIterPrime{\hat{\anIter}-1} -  \yVecIter{\hat{\anIter}-1}) &=  \transMatOnlEl{\anIter+1}{\aParty}{\aParty} ( \yValPrime{\aParty}{\anIter} -  \yVal{\aParty}{\anIter})  \\
		&=\frac{ \beta  \totalWeightParty{\aParty} \transMatOnlEl{\anIter+1}{\aParty}{\aParty} } { 1- \transMatOnlEl{\anIter+1}{\aParty}{\aParty}}.
	\end{align*}
	From the above and Equation \eqref{eq:zi} we have that 
	\begin{equation} 
		\matrixPartOnlEl{\aParty}{\anIter'}{ :}  (\noiseMatrixEl{\aParty}{:}' - \noiseMatrixEl{\aParty}{:})^\top = \beta h_{\anIter'} \text{\qquad for all $\anIter' \in \onlIter{\aParty} \setminus \{\anIter\}$} \label{eq:cancelDiff.1} 
	\end{equation}
	where $h_{\anIter'} = -  \xCoeffOnlVec{\aParty}$ for $\anIter'  \in \onlIter{\aParty} \setminus \{\anIter, \hat{\anIter}\}$ and  \[h_{\hat{\anIter}} = -  \xCoeffOnlVec{\aParty}  -  \frac{\totalWeightParty{\aParty} \transMatOnlEl{\anIter+1}{\aParty}{\aParty} } { 1- \transMatOnlEl{\anIter+1}{\aParty}{\aParty}}.
	\] Given that $\partDistrOnlSymb$ is Valid, we have that  $\hat{Z} = (\matrixPartOnlEl{\aParty}{\anIter'}{:})_{\anIter' \in \onlIter{\aParty}\setminus \{\anIter\}}$ is full row rank. Hence, there exist $v \in \xVectSpace^\iterCnt$ that satisfies
	\[
	\hat{Z} v   =  (h_{\anIter'})_{\anIter' \in \onlIter{\aParty} \setminus \{\anIter\}}.
	\] 
	Setting $(\noiseMatrixEl{\aParty}{:}' - \noiseMatrixEl{\aParty}{:})^\top =   \beta v$ satisfies Equation~\eqref{eq:cancelDiff.1}.
	
	We now compute $\valPrivNoisy{\partyB}'$ and $\cancelNoise{\partyB}{:}'$ for $\partyB \in \partySet \setminus \{\aParty\}$. Similarly to Equation \eqref{eq:yminus}, by equations \eqref{eq:iterPerParty} and \eqref{eq:viewConstraint.2}, we have that
	\begin{align}
		0 &= 	\yValPrime{\partyB}{\anIter'} - \yVal{\partyB}{\anIter'} \notag   \\
		&=     \transMatOnlEl{\anIter'}{\partyB}{:} ( \yVecIterPrime{\anIter'-1} -  \yVecIter{\anIter'-1}) + \xCoeffOnl{\partyB}{\anIter'} (\valPrivNoisy{\partyB}' - \valPrivNoisy{\partyB})+  \matrixPartOnlEl{\partyB}{\anIter'}{ :}  (\noiseMatrixEl{\partyB}{:}' - \noiseMatrixEl{\partyB}{:})^\top  \label{eq:yminus.2}.
	\end{align}
	for all $\anIter' \in \onlIter{\partyB}$.

	Let $k_{\anIter'} = -  \transMatOnlEl{\anIter'}{\partyB}{:} ( \yVecIterPrime{\anIter'-1} -  \yVecIter{\anIter'-1})$ for all $\anIter' \in \onlIter{\partyB}$. Let $\tilde{\anIter}$ be the earliest iteration of $\onlIter{\partyB}$ that is bigger than $\anIter$ (it always exist, as before, as $\iterCnt \in \onlIter{\partyB}$).  $k_{\anIter'}$ is equal to 0 for $\anIter' \in \onlIter{\partyB} \setminus \{ \tilde{\anIter}\}$ and by Equation \eqref{eq:msgdiff} 
	\[ 
	k_{\tilde{\anIter}} = \transMatOnlEl{\anIter+1}{\partyB}{\partyA}( \yVal{\partyA}{\anIter} -  \yValPrime{\partyA}{\anIter}) = - \frac{\beta  \totalWeightParty{\aParty} \transMatOnlEl{\anIter+1}{\partyB}{\aParty} } { 1- \transMatOnlEl{\anIter+1}{\aParty}{\aParty}}. 
	\] 
	For $\tilde{Z} = (\xCoeffOnl{\partyB}{\anIter'}, \matrixPartOnlEl{\partyB}{\anIter'}{:})_{\anIter \in \onlIter{\partyB}}$, Equation \eqref{eq:yminus.2} is equivalent to
	\begin{align}
		\tilde{Z} 
		\begin{ourmatrix}  \valPrivNoisy{\partyB}' - \valPrivNoisy{\partyB} \\
			\cancelNoise{\partyB}{:}' - \cancelNoise{\partyB}{:} 
		\end{ourmatrix}   
		= \beta( k_{\anIter'})_{\anIter' \in \onlIter{\partyB}} \label{eq:cancelDiff.2}.
	\end{align}
	As $\partDistrOnlSymb$ is a Valid Gaussian, we have that $\tilde{Z}$ full row rank.  Therefore there exist $v \in \xVectSpace^{\iterCnt+1}$ such that  $	\tilde{Z} v = ( k_{\anIter'})_{\anIter' \in [0,\iterCnt]}$. Setting 
	\[\begin{ourSmallMatrix}  \valPrivNoisy{\partyB}' - \valPrivNoisy{\partyB} \\
		\cancelNoise{\partyB}{:}' - \cancelNoise{\partyB}{:} 
	\end{ourSmallMatrix}   = \beta v
	\] 
	satisfies Equation \eqref{eq:cancelDiff.2}. Then, setting $\edgeVec{\aParty}{\anIter} = \vectPrivNoisy' - \vectPrivNoisy$  and $\etaTransEl{\aParty}{\anIter} = \cancelNoiseMat' - \cancelNoiseMat$ finishes the first part of our claim.

	We now prove the remaining part of the claim, which is that there exist $\edgeVecHat{\aParty}{\anIter}_\aParty$ such that   
	\[ 
	\edgeVecHat{\aParty}{\anIter}_\aParty = (\transMatOnlEl{\anIter}{\aParty}{\aParty}-1)/\totalWeightParty{\aParty}\]
	and
	\[
	\edgeVecHat{\aParty}{\anIter}_\partyB = \transMatOnlEl{\anIter}{\partyB}{\aParty}/\totalWeightParty{\partyB} \text{\quad for all $\partyB \in \partySetHonest \setminus \{\partyB\}$}\] satisfies the lemma.

	For $K \in \R$, any re-scaling  $(K\edgeVec{\aParty}{\anIter}, K\etaTransEl{\aParty}{\anIter})$ of $(\edgeVec{\aParty}{\anIter}, \etaTransEl{\aParty}{\anIter})$  satisfies the lemma. 
	We will prove the claim for  a  re-scaling $\edgeVecHat{\aParty}{\anIter}$ of $\edgeVec{\aParty}{\anIter}$.

	%We start by computing  $\valPrivNoisy{\partyB}' - \valPrivNoisy{\partyB}$ for all $\partyB \in \partySetHonest - \{\aParty\}$.  
	
	For all  $\partyB \in \partySetHonest \setminus \{\aParty\}$, by equations \eqref{eq:cdp.yminus.3} and \eqref{eq:viewConstraint.2} we have that 
	\begin{align}
		0 &=  \sum_{\anIter' \in \onlIter{\partyB}} \yValPrime{\partyB}{\anIter'} - \sum_{\anIter' \in  \onlIter{\partyB}} \yVal{\partyB}{\anIter'} \notag \\
		&=  \transMatOnlEl{\anIter+1}{\partyB}{\aParty}( \yValPrime{\aParty}{\anIter} - \yVal{\aParty}{\anIter})   +  \totalWeightParty{\partyB} (\valPrivNoisy{\partyB}' - \valPrivNoisy{\partyB}) \notag \\
		(\text{by Eq. \eqref{eq:msgdiff}}) &=  \transMatOnlEl{\anIter+1}{\partyB}{\aParty}\frac{  \totalWeightParty{\aParty} \beta} { 1- \transMatOnlEl{\anIter+1}{\aParty}{\aParty}}     +  \totalWeightParty{\partyB} (\valPrivNoisy{\partyB}' - \valPrivNoisy{\partyB}).
	\end{align}  
	The above is equivalent to 
	\begin{equation}
		\valPrivNoisy{\partyB}' - \valPrivNoisy{\partyB} = \beta \frac{  \totalWeightParty{\aParty} } {\transMatOnlEl{\anIter+1}{\aParty}{\aParty}-1}     \frac{\transMatOnlEl{\anIter+1}{\partyB}{\aParty}}{\totalWeightParty{\partyB}}  = \beta \edgeVec{\aParty}{\anIter}_{\partyB}
		\label{eq:edgeVecRescale} 
	\end{equation}
	for all $\partyB \in \partySetHonest \setminus \{\aParty\}$. 
	Let $K = \frac {\transMatOnlEl{\anIter+1}{\aParty}{\aParty}-1}{  \totalWeightParty{\aParty} }$ and  $\edgeVecHat{\aParty}{\anIter} = K \edgeVec{\aParty}{\anIter}$. We deduce that $\edgeVecHat{\aParty}{\anIter}$ satisfies the lemma as it is a re-scaling of $\edgeVec{\aParty}{\anIter}$. Recalling that $\valPrivNoisy{\aParty}' - \valPrivNoisy{\aParty} = \beta$ we know that $\edgeVec{\aParty}{\anIter}_{\aParty}=1$ and hence 
	\[ 
	\edgeVecHat{\aParty}{\anIter}_{\aParty} = K \edgeVec{\aParty}{\anIter}_{\aParty} = \frac {\transMatOnlEl{\anIter+1}{\aParty}{\aParty}-1}{  \totalWeightParty{\aParty}}. 
	\]
	By equation \eqref{eq:edgeVecRescale} we have that  
	\[ 
	\edgeVecHat{\aParty}{\anIter}_{\partyB} =  K \edgeVec{\aParty}{\anIter}_{\partyB} =   \frac{\transMatOnlEl{\anIter+1}{\partyB}{\aParty}}{\totalWeightParty{\partyB}}.
	\] for all $\partyB \in \partySetHonest \setminus \{\aParty\}$. The last two equations are what we wanted to prove.

	\paragraph{Permanent Dropouts} 
	%The case where there permanent drop-outs is more harmful for utility but also more privacy preserving due to non-canceled extra noise. We will however not account for the uncanceled correlated noise of $\cancelNoiseMat$ as a privacy enhancer, as we can never ensure the presence of drop-outs. Therefore we will consider that the adversary has extra knowledge. 

	%we will not account for the privacy amplification given by uncanceled noise terms of matrix $\cancelNoiseMat$. We will assume that the adversary has more knowledge than he actually has. 
	%In particular,  for all $\aParty'  \in \partySet \setminus  \onlNodes{\iterCnt}$ the adversary canot see $\yVal{\aParty'}{\iterCnt}$ as the party dropped out. 
	For simplicity, in the case of permanent drop-outs we will assume that for each permanently dropped party $\aParty' \in \partySet \setminus \onlNodes{\iterCnt}$ the  adversary knows a special message: 
	\begin{equation}
		\yVal{\aParty'}{\iterCnt} = \sum_{\partyB \in \partySet} \transMatOnlEl{\iterCnt}{\aParty'}{\partyB} \yVal{\partyB}{\aParty'} + \sum_{k=1}^\iterCnt \matrixPartSpecialOnlEl{\aParty'}{\iterCnt}{k} \cancelNoise{\aParty'}{k}
		\label{eq:yspecial.drop}  
	\end{equation}  
	where 
	\[ 
	\matrixPartSpecialOnlEl{\aParty'}{\iterCnt}{k} =  -\sum_{\anIter'=0}^{\iterCnt-1}\matrixPartSpecialOnlEl{\aParty'}{\anIter'}{k} \quad \text{for all $k \in [1, \iterCnt]$}.
	\]
	The special message of Equation \eqref{eq:yspecial.drop} cancels remaining noise that party $\aParty'$ didn't cancel due to drop-out. 
	Note that in the original setting, the adversary does not know $\yVal{\aParty'}{\iterCnt}$ as party $\aParty'$ dropped out. Therefore if we prove that there exist $\edgeVec{\aParty}{\anIter}$ and $\etaTransEl{\aParty}{\anIter}$ that satisfy the lemma with the extra knowledge of  Equation \eqref{eq:yspecial.drop} then this will also hold with the original knowledge of the adversary where these special messages are unknown. Then, for each party $\aParty' \in \partySet \setminus  \onlNodes{\iterCnt}$ we define $\matrixPartSpecialOnl{\aParty'} \in \R^{(\iterCnt+1)\times \iterCnt}$ such that 
	$\matrixPartSpecialOnlEl{\aParty'}{\anIter}{:} =\matrixPartOnlEl{\aParty'}{\anIter}{:}$ for all $\anIter \in [0, \iterCnt-1]$ and  $\matrixPartSpecialOnlEl{\aParty'}{\iterCnt}{:}$ defined as above. 
	
	Running Algorithm \ref{alg:ourprotocol.drop} with distributions 
	\[ \partDistrOnlParty{\aParty'} = \partDistrGauss{\valPriv{\aParty'}+\indnoise{\aParty'}}{\xCoeffOnlVec{\aParty'}}{\matrixPartSpecialOnl{\aParty'}}
	\]  for all $\aParty' \in \partySet \setminus  \onlNodes{\iterCnt}$ is equivalent to an execution without permanent dropouts. Hence we can compute  $\edgeVec{\aParty}{\anIter}$ and $\etaTransEl{\aParty}{\anIter}$ as in the temporary dropout case. Note that the value of $\edgeVec{\aParty}{\anIter}$ only depends on $\transMatEl{\anIter+1}{:}{\aParty}$ and $\totalWeightVec$ which only depends on $(\xCoeffOnlVec{\aParty'})_{\aParty' \in \partySet}$. Therefore its value is not modified by the changes introduced by  $\matrixPartSpecialOnl{\aParty'}$ for all $\aParty' \in \partySet$. 
\end{proof}

\subsection{Proofs of theorems \ref{thm:cdp.nodrop}, \ref{thm:totalsum} and \ref{thm:coalition}}
\label{app:privacy.cdp} 

In this appendix, we prove theorems \ref{thm:cdp.nodrop}, \ref{thm:totalsum} and \ref{thm:coalition}. We start by proving Lemma \ref{lm:partialWeightSum}. Then we prove Theorem \ref{thm:coalition} and show that theorems \ref{thm:cdp.nodrop} and \ref{thm:totalsum} are a special case of Theorem \ref{thm:coalition}.

\newcommand{\weightedSumSet}[1]{\mathcal{K}^{(#1)}}
\begin{lemma} 
	\label{lm:partialWeightSum} 
	Let $\exec = (\transitionset,\onlSymb,\advValView)$ be defined as in Theorem \ref{thm:dp.abstract}. Let $\coalSet$, $\totalWeightVec$,  $\hidden$, $\{\edgeVec{\aParty}{\anIter}\}_{(\aParty,\anIter) \in  \hidden}$ 
	be as defined in Theorem \ref{thm:coalition}. 
	If   $\{ \edgeVec{\aParty}{\anIter} \}_{(\aParty,\anIter) \in \hidden}$ has at least $|\coalSet| -1$ independent vectors, then for any $\vectPriv$, $\indNoiseVec$, 
	$\vectPriv'$ and $(\indNoiseVec)'$ such that 
	\[ 
	\sum_{\aParty \in \coalSet} \totalWeightParty{\aParty} (\valPriv{\aParty} + \indnoise{\aParty})
	=  \sum_{\aParty \in \coalSet} \totalWeightParty{\aParty} (\valPriv{\aParty}' + (\indnoise{\aParty})'  ).
	\] 
	there exist $\cancelNoiseDiff \in \xVectSpace^{\nHonest \times \iterCnt}$ such that 
	\[
	\advViewExec{\exec}{\vectPriv+ \indNoiseVec, \cancelNoiseMat}
	=\advViewExec{\exec}{\vectPriv'+ (\indNoiseVec)', \cancelNoiseMat+\cancelNoiseDiff}
	\]
	for any $\cancelNoiseMat$. 
\end{lemma}  

\begin{proof}  
	%	As noise of $\cancelNoiseMat$ cancels for temporary drop-outs, we have that 
	%	\begin{equation}
		%		\sum_{\aParty \in \partySetHonest} \totalWeightParty{\aParty} \valPrivNoisy{\aParty}  
		%		\label{eq:weightSum} 
		%	\end{equation} 
	%	is part of the view of the adversary.  
	%where $\totalWeightParty{\aParty} = \sum_{\anIter=0^\iterCnt} \xCoeffOnl{\aParty}{\anIter}$ and $\valPrivNoisy{\aParty} = \valPriv{\aParty} + \indnoise{\aParty}$. 
	
	We first define 
	\[ 
	\weightedSumSet{\vectPrivNoisy} = \left\{\vectPrivNoisy' \in \xVectSpace^{\nHonest} :  \sum_{\aParty \in \coalSet} \totalWeightParty{\aParty} \valPrivNoisy{\aParty}'  =  \sum_{\aParty \in \coalSet} \totalWeightParty{\aParty} \valPrivNoisy{\aParty}   \right\}. 
	\]
	and 
	%Note that for any $\vectPrivNoisy' \in \weightedSumSet{\vectPrivNoisy}$ 
	\begin{align*}
		&\hidden^{(\vectPrivNoisy)} \\
		&=  \left\{ \vectPrivNoisy + \sum_{(\aParty, \anIter) \in \hidden}  \beta_{(\aParty, \anIter)} \edgeVec{\aParty}{\anIter}  \in \xVectSpace^{\nHonest}:   \beta_{(\aParty,\anIter)} \in \R \quad \forall  (\aParty,\anIter) \in \hidden \right\}. 
	\end{align*}
	
	By definition, $\transMatOnlEl{\anIter+1}{\partyB}{\aParty} \not{=} 0$ only if $\partyB \in \outNeigh{\aParty}{\anIter}$ or if $\partyB = \aParty$. By Equation \eqref{eq:colstoch} and since $\outNeigh{\aParty}{\anIter} \in \coalSet$ we have that for any $(\aParty, \anIter) \in \hidden$, 
	\begin{equation}
		\sum_{\partyB \in \coalSet} \transMatOnlEl{\anIter+1}{\partyB}{\aParty} = \sum_{\partyB \in \partySet} \transMatOnlEl{\anIter+1}{\partyB}{\aParty} = 1.
		\label{eq:colstoch.coalition} 
	\end{equation}
	Then for any $(\aParty,\anIter) \in \hidden$, $\noisyVec \in \xVectSpace^{\nHonest}$ and $\beta \in \R$, 
	\begin{align*}
		\sum_{\partyB \in \coalSet}\totalWeightParty{\partyB}( \noisyVec_{\partyB} + \beta \edgeVec{\aParty}{\anIter}_\partyB)  &=  \sum_{\partyB \in \coalSet} \totalWeightParty{\partyB}\noisyVec_{\partyB} + \sum_{\partyB \in \coalSet} \totalWeightParty{\partyB} \edgeVec{\aParty}{\anIter}_{\partyB} \\
		&=  \sum_{\partyB \in \coalSet} \totalWeightParty{\partyB}\noisyVec_{\partyB} + \totalWeightParty{\aParty} \edgeVec{\aParty}{\anIter}_{\aParty} +  \sum_{\partyB \in \coalSet \setminus \{\aParty\}}\edgeVec{\aParty}{\anIter}_{\partyB} \\
		(\text{by Lemma \ref{lm:edgeVec}})	   &= \sum_{\partyB \in \coalSet} \totalWeightParty{\partyB}\noisyVec_{\partyB} + \totalWeightParty{\aParty} \edgeVec{\aParty}{\anIter}_{\aParty} + \sum_{\partyB \in \coalSet \setminus \{\aParty\}} \totalWeightParty{\partyB} \frac{\transMatOnlEl{\anIter+1}{\partyB}{\aParty}}{\totalWeightParty{\partyB}} \\	   
		&=  \sum_{\partyB \in \coalSet} \totalWeightParty{\partyB}\noisyVec_{\partyB}+  \totalWeightParty{\aParty} \frac{\transMatOnlEl{\anIter+1}{\aParty}{\aParty}-1}{\totalWeightParty{\aParty}} +  \sum_{\partyB \in \coalSet \setminus \{\aParty\}} \transMatOnlEl{\anIter+1}{\partyB}{\aParty}\\
		&=\sum_{\partyB \in \coalSet} \totalWeightParty{\partyB}\noisyVec_{\partyB} -1 + \sum_{\partyB \in \coalSet } \transMatOnlEl{\anIter+1}{\partyB}{\aParty} \\
		\text{(by Eq.  \eqref{eq:colstoch.coalition})}	&= \sum_{\partyB \in \coalSet} \totalWeightParty{\partyB}\noisyVec_{\partyB}.
	\end{align*}
	Therefore  $\vectPrivNoisy + \beta \edgeVec{\aParty}{\anIter} \in \weightedSumSet{\vectPrivNoisy}$.  This implies that   $\hidden^{(\vectPrivNoisy)} \subseteq \weightedSumSet{\vectPrivNoisy}$. We know that  $\hidden^{(\vectPrivNoisy)}$ has dimension $|\coalSet| -1$. Since $\weightedSumSet{\vectPrivNoisy}$ also has dimension $|\coalSet|-1$, it must be that $\hidden^{(\vectPrivNoisy)} = \weightedSumSet{\vectPrivNoisy}$. 
	
	This means that for any $\vectPrivNoisy' \in \weightedSumSet{\vectPrivNoisy}$ there exists $(\beta_{(\aParty,\anIter)})_{(\aParty,\anIter) \in \hidden} \in \R^{|\hidden|} $ such that 
	\[\vectPrivNoisy + \sum_{(\aParty, \anIter) \in \hidden}  \beta_{(\aParty, \anIter)} \edgeVec{\aParty}{\anIter}  = \vectPrivNoisy'.
	\] 
	By successively applying Lemma \ref{lm:edgeVec} for each $(\aParty, \anIter) \in \hidden$ we have that 
	\[ 
	\advViewExec{\exec}{\vectPrivNoisy, \cancelNoiseMat} = \advViewExec{\exec}{\vectPrivNoisy', \cancelNoiseMat + \cancelNoiseDiff}
	\] 
	for any $\cancelNoiseMat$, where $\cancelNoiseDiff = - \sum_{(\aParty,\anIter) \in \hidden} \beta_{(\aParty,\anIter)} \etaTransEl{\aParty}{\anIter}$ and $\etaTransEl{\aParty}{\anIter}$ as defined in Lemma \ref{lm:edgeVec}. This completes the proof. 
\end{proof}  

% Previous statement of the next lemma
% 
%Let $\exec$ be an execution of Algorithm \ref{alg:ourprotocol.drop} where $\partDistrOnlSymb$ is  $(\xCoeffOnlVec{\aParty}, \matrixPartOnl{\aParty})_{\aParty \in \partySet}$-\emph{Robust}. Let $\coalSet \subseteq \partySetHonest$ be a coalition that contain all honest users that did not drop-out permanently (i.e. $\onlNodes{\iterCnt} \cap \partySetHonest \subseteq \coalSet)$. Let $\hidden = \{ (\aParty, \anIter) \in \partySetHonest\times [0,\iterCnt] \setminus \advValView : \aParty \in \coalSet \land  \outNeigh{\aParty}{\anIter} \subseteq \coalSet \}$ the set of pairs $(\aParty,\anIter)$ such that $\yVal{\aParty}{\anIter}$ is not seen by the adversary and only sent to members of the coalition $\coalSet$.  Let for all $(\aParty,\anIter) \in  \hidden$, let $\edgeVec{\aParty}{\anIter}$ be as defined in Lemma \ref{lm:edgeVec}. 
%If   $\{ \edgeVec{\aParty}{\anIter} \}_{(\aParty,\anIter) \in \hidden}$ has at least $|\coalSet| -1$ independent vectors, 
Theorem \ref{thm:coalition}. \textit{\thmPermDropStm{}}

%\begin{theorem}
%	\thmPermDropStm{}
%	\label{thm:coalition} 
%\end{theorem}  

%\begin{lemma}
%	Let $\epsilon, \delta \in (0,1)$.
%		Let $\exec = (\transitionset,\onlSymb,\advValView)$ be defined as in Theorem \ref{thm:dp.abstract}. Let $\coalSet$, $\totalWeightVec$, $\wOnl$, $\hidden$, $\{\edgeVec{\aParty}{\anIter}\}_{(\aParty,\anIter) \in  \hidden}$ 
%	be as defined in Theorem \ref{thm:coalition}. 
%	If   $\{ \edgeVec{\aParty}{\anIter} \}_{(\aParty,\anIter) \in \hidden}$ has at least $|\coalSet| -1$ independent vectors,
%	then
%	there exist $\indNoiseDiff$ and $\cancelNoiseDiffVec$ as defined in Lemma \ref{lm:dp.noiseDiff.abstract} and  $\exec$ satisfies  \obsDP{$\epsilon$}{$\delta$}{$\advValView$}  (or \colDP{$\epsilon$}{$\delta$}{$\corrnodes$})
%	for any  
%	\[ 
%	\sdInd^2 >  \frac{(|\coalSet|-1)c^2}{(\wOnl-1)^2 \epsilon^2} 
%	\] 
%	and  
%	\[ 
%	\sdCancel^2 \ge \frac{\|\cancelNoiseDiffVec\|^2_2}{
	%		\frac{\epsilon^2}{c^2} - 		\frac{(|\coalSet|-1)}{ (\wOnl-1)^2 \sdInd^2}},
%	\] 
%	where $c^2 > 2\ln(1.25/\delta)$. 
%\end{lemma}  

\begin{proof}
	Let $\neighDatasetA$ and $\neighDatasetB$ be two neighboring datasets. 
	Let $\aParty \in \partySetHonest$ be the coordinate in which $\neighDatasetB$ and $\neighDatasetA$ differ. Without loss of generality,  we assume $\neighDatasetB_\aParty -\neighDatasetA_\aParty = 1$. We split the proof into two cases, depending on whether $\aParty \in \coalSet$ or not. We start with the case where $\aParty \in \coalSet$. 
	
	Define  $\indNoiseDiff$ such that $\indNoiseDiff_{\partyB} = -\frac{\totalWeightParty{\aParty}}{\|\totalWeightVec\|_2^2}  \totalWeightParty{\partyB}$ if $\partyB \in \coalSet$ and $\indNoiseDiff_{\partyB} = 0$ otherwise.
	For any $\indNoiseVec$, we have that 
	\[ 
	\sum_{\partyB \in \partySetHonest} \totalWeightParty{\partyB} (\neighDatasetA_{\partyB} + \indnoise{\partyB})
	=  \sum_{\partyB \in \partySetHonest} \totalWeightParty{\partyB} (\neighDatasetB_{\partyB} + \indnoise{\partyB} + \indNoiseDiff_{\partyB})  .
	\] 
	By Lemma \ref{lm:partialWeightSum}, there exists $\cancelNoiseDiff$ such that 
	\begin{equation} 
		\advViewExec{\exec}{\neighDatasetA + \indNoiseVec, \cancelNoiseMat}
		=\advViewExec{\exec}{\neighDatasetB + \indNoiseVec + \indNoiseDiff, \cancelNoiseMat+\cancelNoiseDiff}. 
		\label{eq:indist} 
	\end{equation}
	for any $\cancelNoiseMat$.
	Therefore we can apply Lemma \ref{lm:dp.noiseDiff.abstract}. We have that execution $\exec$ is $(\epsilon, \delta)$-DP if $\sdInd^2$, $\sdCancel^2$, $\|\indNoiseDiff\|^2_2$, $\|\cancelNoiseDiff\|^2_2$, $\epsilon$ and $\delta$ satisfy Equation 
	\eqref{eq:dpThetaMax}. 
	
	We have that 
	\[ 
	\|\indNoiseDiff\|^2_2 =   \frac{\totalWeightParty{\aParty}^2}{  \| (\totalWeightParty{\aParty})_{\aParty \in \coalSet} \|_2^2 } 
	\]  
	From the above, Equation \eqref{eq:dpThetaMax} is implied if 
	\[ 
	\frac{\|\cancelNoiseDiff\|^2_2}{\sdCancel^2}  \le
	\frac{\epsilon^2}{c^2} - 	\frac{\totalWeightParty{\aParty}^2}{ \| (\totalWeightParty{\aParty})_{\aParty \in \coalSet} \|_2^2 \sdInd^2}.
	\]  
	We know that  
	\[\sdInd^2 > \frac{ w_{max}^2}{ \| (\totalWeightParty{\aParty})_{\aParty \in \coalSet} \|_2^2}\frac{c^2}{ \epsilon^2}.
	\] 
	Then,  we can deduce that 
	\[ \frac{\epsilon^2}{c^2} - 	\frac{\totalWeightParty{\aParty}^2}{ \| (\totalWeightParty{\aParty})_{\aParty \in \coalSet} \|_2^2 \sdInd^2} > 0.
	\] 
	The above is equivalent to
	\begin{equation} 
		\sdCancel^2 \ge \frac{\|\cancelNoiseDiff\|^2_2}{
			\frac{\epsilon^2}{c^2} - 		\frac{\totalWeightParty{\aParty}^2}{ \| (\totalWeightParty{\aParty})_{\aParty \in \coalSet} \|_2^2 \sdInd^2}}.
		\label{eq:sigmaDelta}
	\end{equation}
	In other words, with $\sdCancel^2$ lower bounded by the above expression   execution $\exec$ is 
	$(\epsilon,\delta)$-DP as Equation  \eqref{eq:dpThetaMax} is 
	satisfied. 
	
	We now prove the lemma when $\aParty \not\in \coalSet$. It must be that $\aParty$ dropped out permanently, which means that $\matrixPartOnlEl{\aParty}{\iterCnt}{:} = 0$. As $\partDistrSymb$ is a Valid $(\xCoeffVec, \matrixPart)$-Gaussian, we have that $\matrixPartOnlEl{\aParty}{-\iterCnt}{:}$ is full row rank. The $\iterCnt$-th row of $\matrixPartOnl{\aParty}$ is equal to 0. Then,  $\matrixPartOnl{\aParty}$ and $\matrixPartOnlEl{\aParty}{-\iterCnt}{:}$ share the same solutions. Hence $\matrixPartOnl{\aParty}$ is also full row rank. 
	
	We first prove that for any $\indNoiseVec$ and $\cancelNoiseMat$ there exist $\indNoiseDiff$ and $\cancelNoiseDiff$ such that 
	\[\|  \indNoiseDiff \|_2^2 \le \frac{\totalWeightParty{\aParty}^2}{\| (\totalWeightParty{\aParty})_{\aParty \in \coalSet}\|_2^2 }
	\]  
	and Equation \eqref{eq:indist} is satisfied. 
	The former requirement is met by setting  
	\[ \indNoiseDiff_{\aParty} =  \frac{\totalWeightParty{\aParty}^2}{\| (\totalWeightParty{\partyB})_{\partyB \in \coalSet}\|_2^2 }
	\] and $\indNoiseDiff_{\partyB} = 0$ for all $\partyB \in \partySetHonest \setminus \{ \aParty\}$. Equation \eqref{eq:indist} holds if $\yVal{\aParty}{\anIter}$ does not change for all $\anIter \in [0, \iterCnt]$ and the algorithm is executed with input $(\neighDatasetA, \indNoiseVec, \cancelNoiseMat)$ or with input $(\neighDatasetB, \indNoiseVec + \indNoiseDiff, \cancelNoiseMat + \cancelNoiseDiff)$. This is equivalent to proving that 
	\begin{align*}
		&\xCoeffOnlVec{\aParty} (  \neighDatasetA_{\aParty} + \indnoise{\aParty} ) + \matrixPartOnl{\aParty} (\cancelNoise{\aParty}{:})^\top \\
		&=
		\xCoeffOnlVec{\aParty} ( \neighDatasetB_{\aParty} + \indnoise{\aParty} + \indNoiseDiff_{\aParty} ) + \matrixPartOnl{\aParty} (\cancelNoise{\aParty}{:} + \cancelNoiseDiff_{\aParty,:})^\top. 
	\end{align*}
	The above is equivalent to
	\begin{align*}
		\matrixPartOnl{\aParty} ( \cancelNoiseDiff_{\aParty,:})^\top  = - \xCoeffOnlVec{\aParty} (  1 +   \indNoiseDiff_{\aParty}).
	\end{align*}
	As 	$\matrixPartOnl{\aParty}$ is  full row rank, there exists $\cancelNoiseDiff_{\aParty,:}$ that satisfies such equation. 
	
	Then by setting $\cancelNoiseDiff_{\partyB,:} = 0$ for all $\partyB \in \partySetHonest \setminus \{\aParty\}$ we apply Lemma \ref{lm:dp.noiseDiff.abstract} for $\indNoiseDiff$ and $\cancelNoiseDiff$ as defined above.  We do the same reasoning as the last part of the derivation where $\aParty \in \coalSet$ (after Equation \eqref{eq:indist}), showing that Equation  \eqref{eq:dpThetaMax} is satisfied for \[\sdInd^2 > \frac{ w_{max}^2}{ \| (\totalWeightParty{\aParty})_{\aParty \in \coalSet} \|_2^2}\cdot\frac{c^2}{ \epsilon^2}\] 
	and  $\sdCancel^2$ as in Equation \eqref{eq:sigmaDelta}. This concludes the proof. 
	
\end{proof}

Theorem \ref{thm:cdp.nodrop}. \textit{ \thmCdpNodropStm }
\begin{proof}
	If there are no dropouts, then $\totalWeightParty{\aParty} = 1$ for all $\aParty \in \partySetHonest$ and $\onlNodes{\anIter}=\partySet$ for all $\anIter \in [0,\iterCnt]$. Then the proof is a direct application of Theorem \ref{thm:coalition} with $\coalSet = \partySetHonest$.   
\end{proof}

Theorem \ref{thm:totalsum}. 	\thmTotalSumStm{} 
\begin{proof}
	We have that  
	\[ \frac{w_{max}^2}{ \|\totalWeightVec\|_2^2 } \le  \frac{\nHonest-1}{(\wOnl-1)^2}\] for all $\totalWeightVec$ such that  $\sum_{\aParty \in \partySetHonest} \totalWeightParty{\aParty}= \wOnl$. 
	Then, it follows directly by applying Theorem \ref{thm:coalition} with $\coalSet = \partySetHonest$.
\end{proof}

\subsection{Proof of theorems \ref{thm.stronglyConnected.totalSum} and \ref{thm:static}} 
\label{app:privacy.topologies} 

Theorem \ref{thm.stronglyConnected.totalSum}. \textit{
	\thmStronglyConnectedTotSumStm{}
}

\newcommand{\colSymb}{F}
\newcommand{\colSet}[1]{\colSymb^{(#1)}}
\begin{proof} 
	We first construct   a weighted adjacency matrix  $\transMatHidden$ of graph $\flatGraphHidden$.
	For each party $\aParty \in \partySetHonest$, let 
	\[
	\colSet{\aParty} = \left\{ \transMatOnlEl{\anIter+1}{\partySetHonest}{\aParty}  \in \R^{\nHonest}:  (\aParty,\anIter) \in 
	\partySetHonest \times [0, \iterCnt-1] \setminus \advValView  \right\}
	\] the set whose elements are columns \[ \transMatOnlEl{\anIter+1}{\partySetHonest}{\aParty}= (\transMatOnlEl{\anIter+1}{\partyB}{\aParty})_{\partyB \in \partySetHonest} \in \R^{\nHonest},
	\] which determine the weights of outgoing messages $\yVal{\aParty}{\anIter}$ that have not been seen by the adversary. 
	
	For each party $\aParty \in \partySetHonest$,  we set \[\transMatHiddenEl{:}{\aParty} = \frac{1}{|\colSet{\aParty}|} \sum_{v \in \colSet{\aParty}}v. 
	\]
	For all $\aParty, \partyB \in \partySetHonest$, we have that $\transMatHidden_{\aParty,\partyB} > 0$ if $(\aParty,\partyB) \in E(\flatGraphHidden)$ and $\transMatHidden_{\aParty,\partyB}= 0$ otherwise. Therefore, $\transMatHidden$ is a weighted adjacency matrix of $\flatGraphHidden$. As $\flatGraphHidden$ is strongly connected then  $\transMatHidden$ is irreducible. This means that  for all $\aParty, \partyB \in \partySetHonest$, $(\transMatHidden)^k_{\aParty,\partyB} > 0$ for a sufficiently large integer $k$. 
	
	Therefore we can apply the Perron-Frobenius Theorem for non-negative irreducible matrices. This implies that the largest eigenvalue of $\transMatHidden$ is smaller or equal to 1 and has multiplicity 1.  The same applies to $(\transMatHidden)^\top$. 
	
	As columns of $\transMatHidden$ are the average of column-stochastic matrices, then $\transMatHidden$ is column stochastic. Therefore, \[(\transMatHidden)^\top \oneVec = \oneVec.
	\] This means that 1 is an eigenvalue of $(\transMatHidden)^\top$ associated to the eigenvector $\oneVec$. 
	
	By the results of our application of the Perron-Frobenius Theorem, it must be that 1 is the largest eigenvalue of $(\transMatHidden)^\top$ and has multiplicity 1. This means that the nullspace of $(\transMatHidden)^\top-I$ has dimension $1$, which implies that $(\transMatHidden)^\top - I$ has rank $\nHonest -1$. Then 
	\[
	\bar{\aTransMat} = \transMatHidden - I
	\] also has rank $\nHonest - 1$. Let
	\[
	\colSet{\aParty}_I = \left\{ \transMatOnlEl{\anIter+1}{\partySetHonest}{\aParty} - \project{\aParty}  \in \R^{\nHonest}:  (\aParty,\anIter) \in \partySetHonest \times[0, \iterCnt-1] \setminus \advValView  \right\}. 
	\]
	We can deduce that 
	\begin{align*} 
		\colSymb_I &= \bigcup_{\aParty \in \partySetHonest} \colSet{\aParty}_I \\
		&= \left\{\totalWeightVec_{\partySetHonest}^\top \edgeVec{\aParty}{\anIter}_{\partySetHonest} :  (\aParty,\anIter) \in \partySetHonest \times[0, \iterCnt-1] \setminus \advValView  \right\}
	\end{align*}
	where $\totalWeightVec$  is defined in Equation \eqref{eq:totalWeight},  $\edgeVec{\aParty}{\anIter}$ in Lemma \ref{lm:edgeVec}, \[ \totalWeightVec_{\partySetHonest} = (\totalWeightVec_{\partyB})_{\partyB \in \partySetHonest} \in \R^{\nHonest}
	\] 
	and
	\[\edgeVec{\aParty}{\anIter}_{\partySetHonest} = (\edgeVec{\aParty}{\anIter}_{\partyB})_{\partyB \in \partySetHonest} \in \R^{\nHonest}.
	\] 
	We can also deduce that for each $\aParty \in\partySetHonest$
	\[ \bar{\aTransMat}_{:,\aParty} = \frac{1}{|\colSet{\aParty}_I|} \sum_{v \in \colSet{\aParty}_I}v.
	\]
	\newcommand{\totalWeightVecHonest}{\totalWeightVec_{\honestSymb}}
	\newcommand{\diagWInv}{D^{(-1)}_w}
	Given that  $\bar{\aTransMat}$ has rank $\nHonest -1$, $\colSymb_I$ contains at least $\nHonest-1$ linearly independent vectors.
	$\colSymb_I$ can be obtained by multiplying each vector of 
	\[\hat{\colSymb}_I = \left\{  \edgeVec{\aParty}{\anIter}_{\partySetHonest} :  (\aParty,\anIter) \in \partySetHonest \times[0, \iterCnt-1] \setminus \advValView  \right\}
	\]
	by $\totalWeightVec_{\partySetHonest}^\top$. Therefore, $\hat{\colSymb}_I$ also has $\nHonest-1$ linearly independent vectors. For all $(\aParty,\anIter)  \in \partySetHonest \times[0, \iterCnt-1] \setminus \advValView$,  we have that the vector $\edgeVec{\aParty}{\anIter}_{\partyB} = 0$  for all $\partyB \in \corrnodes$. This means that we can obtain 
	\[\tilde{\colSymb}_I = \left\{  \edgeVec{\aParty}{\anIter} :  (\aParty,\anIter) \in \partySetHonest \times[0, \iterCnt-1] \setminus \advValView  \right\}
	\] 
	by adding $0$-valued entries to each vector $\hat{\colSymb}_I$ in the same indexes. Therefore, the number of linearly independent vectors in $\tilde{\colSymb}_I$ is the same as that of $\hat{\colSymb}_I$. Given that 
	$\tilde{\colSymb}_I$ has $\nHonest-1$ linearly independent vectors, we can apply Theorem \ref{thm:totalsum}, which gives the required bound on $\sdInd^2$ to achieve $(\epsilon,\delta)$-DP, for sufficiently large $\sdCancel^2$. 
	\end{proof}

	Now we prove Theorem \ref{thm:static}. 
	
	Theorem \ref{thm:static}. \textit{\thmStaticStm} 
			\begin{proof}
		According to Lemma \ref{lm:edgeVec}, for each $(\aParty, \anIter) \in \hidden$, $\edgeVec{\aParty}{\anIter}$ only depends on $\transMatOnlEl{\anIter+1}{:}{\aParty}$ and $\totalWeightVec$.  If $\transMatOnl{\anIter} = \transMatOnl{1}$ for all $\anIter \in [1,\iterCnt]$, we have that for all $\aParty \in \partySetHonest$, if $(\aParty, \anIter) \in \hidden$ then  $\edgeVec{\aParty}{\anIter} = \edgeVec{\aParty}{0}$ for all $\anIter \in [0,\iterCnt-1]$.  In addition, from the theorem we know that there exist $\partyA,\partyB \in \partySetHonest$ such that $(\partyA,\anIter) \in \advValView$ and $(\partyB,\anIter) \in \advValView$ for all $\anIter  \in [0,\iterCnt]$. Therefore  $\{ \edgeVec{\aParty}{\anIter}\}_{(\aParty, \anIter) \in \hidden}$ is contained in $\{ \edgeVec{\aParty}{0}\}_{\aParty \in \partySetHonest \setminus \{\partyA,\partyB\}}$, which has at most $\nHonest-2$ vectors.
	\end{proof}

\section{Calculations for \muffliato{}}
\label{app:exp.muffliato}

\newcommand{\muffEpsSymb}{\bar{\epsilon}}
\newcommand{\muffEpsParty}[1]{\muffEpsSymb_{#1}}
\newcommand{\muffAlpha}{\alpha_{M}}
\newcommand{\muffDeg}[1]{d_{#1}}
\newcommand{\muffSens}{\Delta_{M}}
\newcommand{\muffSd}{\sigma_{M}}

In this appendix, we show how  the accuracy of \muffliato{} is computed for the experiment in Figure \ref{fig:accuracy1}.

Essentially, for a given function $f: \partySet \times \partySet \to \R^+$ and $\muffAlpha > 1$, a mechanism satisfies $(\muffAlpha, f)$-Pairwise Network DP (PNDP) \cite[Definition 5]{cyffers2022muffliato} if for all pairs of parties $\aParty,\partyB \in \partySet$, it satisfies $(\muffAlpha, f(\aParty, \partyB))$-Renyi~DP (RDP) \cite{mironov2017renyi} assuming the adversary has access only to party $\partyB$'s view.

Results of \cite[Theorem 1]{cyffers2022muffliato} defines the mean privacy loss $\muffEpsParty{\aParty}$ of a party $\aParty$. Given that private values are within $[0,1]$ we have that the local sensitivity is given by $\muffSens = 1$ and  for  any party $\aParty \in \partySetHonest$ we have that  
\begin{equation}
	\muffEpsParty{\aParty} = \frac{\muffAlpha T \muffDeg{\aParty}}{2n \muffSd^2} \label{eq:muff.avgloss}
\end{equation}
by \cite[Equation (7)]{cyffers2022muffliato},  where $\muffDeg{\aParty}$ is the degree of party $\aParty$,  $\muffAlpha$ is the desired parameter $\alpha$ of Renyi DP, $T$ is the number of iterations and $\muffSd^2$ is the variance of the Gaussian noise that each party adds to its private value.  \muffliato{} reports $\muffEpsSymb = \max_{\aParty \in \partySet} \muffEpsParty{\aParty}$ as its measure of privacy loss. 

For the sake of comparison, we convert PNDP guarantees given by $\muffEpsSymb$, which are essentially the mean privacy loss of a RDP guarantee, to a mean privacy loss of an $(\epsilon, \delta)$-DP. By \cite[Proposition 3]{mironov2017renyi}, a mechanism that satisfies $(\alpha, \muffEpsSymb)$-RDP it also satisfies $(\epsilon, \delta)$-DP with 
\begin{equation}
	\epsilon = \muffEpsSymb + \frac{\ln(1/\delta)}{\muffAlpha-1}. \label{eq:rdp.dp}
\end{equation}  

Typically, when guarantees are satisfied for multiple pairs $(\alpha, \muffEpsSymb)$ one uses the pairs that minimizes $\epsilon$ for a fixed  $\delta$. This  optimizes the privacy/accuracy trade-offs. In our experiments of Figure \ref{fig:accuracy1}, we fix $(\epsilon, \delta) = (0.1, 10^{-5})$. Here we do the same and choose $\muffEpsSymb$,  $\muffAlpha$ such that it minimizes $\muffSd^2$. By Equation \eqref{eq:rdp.dp}, we set 
\begin{equation} 
	\muffAlpha = \frac{\log(1/\delta)}{\epsilon - \muffEpsSymb} +1 \label{eq:rdp.dp.eps}. 
\end{equation}
We use the same parameters as those in  Figure \ref{fig:accuracy1} and set \muffliato{} under hypercube graphs, which give the smallest $\muffEpsSymb$ (see \cite[Figure 1a]{cyffers2022muffliato}). For $n = 2^{10} = 1024$, the degree of each node $\aParty$ in a hypercube is $\muffDeg{\aParty} = 10$ for all $\aParty \in \partySet$. 
By plugging Equation \eqref{eq:rdp.dp.eps} into Equation \eqref{eq:muff.avgloss} we have that 
\begin{equation}
	\muffSd^2 = \left(  \frac{\ln(1/\delta)}{\epsilon - \muffEpsSymb} +1  \right) \frac{10 T }{2n \muffEpsSymb } \label{eq:muff.sigma}.
\end{equation}

%To determine the number of iterations $\iterCnt$ of \muffliato{}, we use the same criterion as in \cite[Figure 1a]{cyffers2022muffliato}, which relies on gossip acceleration \cite{berthier2020accelerated} and compute $\iterCnt$ such that the gossip error is smaller than the error induced the noise with variance $\muffSd^2$ when this variance is bigger than 0.25 \footnote{see \href{https://github.com/totilas/muffliato}{https://github.com/totilas/muffliato}}:
We set $\iterCnt$ according to \cite[Theorem 2]{cyffers2022muffliato}. Then we have that 
\[ \iterCnt = \left\lceil \frac{\ln(n)}{\sqrt{\lambda_2}}  \right\rceil    \]
when $\muffSd^2 \ge 1$.  $\lambda_2$ is the second largest eigenvalue of the gossip matrix of a hypercube of $n=1024$ nodes and is equal to  $\frac{9}{11}$. Given these parameters, we deduce that $T = 8$. 

With the current parameters,  $\muffSd^2$ is always greater than $180$. We  ignore the gossip error by assuming that each party converged to the true average of noisy values. Therefore we approximate the MSE of Muffliato by $\muffSd^2/n$ which is at least $0.1764$.

%Let $M_I \in \R^{\nHonest \times |\colSymb_I| }$ be the matrix whose columns are the vectors of $\colSymb_I$. Let $M = \diagWInv M_I \in \R^{\nHonest \times |\colSymb_I|}$ where $\diagWInv = \diag{( \frac{1}{\totalWeightParty{\aParty}})_{\aParty \in \partySet}}$. As $M_I$ has at least rank  $\nHonest-1$, we know that $M$ also has at least rank  $\nHonest-1$. 
%
% We have that for each column $d \in \R^\nHonest$ of $\diagWInv$, there exist $(\aParty,\anIter) \in [0, \iterCnt-1] \setminus \advValView$ such that
%$d_{\aParty} =  (\transMatOnlEl{\anIter+1}{\aParty}{\aParty} - 1)/\totalWeightParty{\aParty}$ and $d_{\partyB} =  \transMatOnlEl{\anIter+1}{\partyB}{\aParty}/\totalWeightParty{\partyB}$ for all $\partyB \in \partySetHonest \setminus \{\aParty\}$. Therefore $d = \edgeVec{\aParty}{\anIter}$ for some $(\aParty,\anIter) \in [0, \iterCnt-1] \setminus \advValView$. The set $\{\edgeVec{\aParty}{\anIter}\}_{
	%(\aParty,\anIter) \in [0, \iterCnt-1] \setminus \advValView}$ is equal to the set of columns of $M$. Therefore $\{\edgeVec{\aParty}{\anIter}\}_{	(\aParty,\anIter) \in [0, \iterCnt-1] \setminus \advValView}$ has at least $\nHonest -1$ linearly independent vectors. 
%	
%We conclude our theorem by applying Theorem \ref{thm:totalsum}, which gives the required bound on $\sdInd^2$ to achieve $(\epsilon,\delta)$-DP for sufficiently large $\sdCancel^2$. 

%\input{app_exp}

\end{document}